\documentclass[pdflatex,sn-mathphys-ay]{sn-jnl}

\usepackage{graphicx}%
\usepackage{multirow}%
\usepackage{amsmath,amssymb,amsfonts}%
\usepackage{amsthm}%
\usepackage{mathrsfs}%
\usepackage{xcolor}%
\usepackage{booktabs}%
\usepackage{algorithm}%
\usepackage{listings}%
\usepackage{pifont}
\usepackage[section]{placeins}
\usepackage[normalem]{ulem}

\DeclareMathAlphabet{\pazocal}{OMS}{zplm}{m}{n}

\newif\ifshowchanges
\showchangesfalse  

\newcommand{\added}[1]{%
  \ifshowchanges
    \textcolor{blue}{#1}%
  \else
    #1%
  \fi
}

\newcommand{\deleted}[1]{%
  \ifshowchanges
    \ifmmode
      \text{\textcolor{red}{\sout{$#1$}}} 
    \else
      \textcolor{red}{\sout{#1}}        
    \fi
  \fi
}

\usepackage{mathtools} 
\usepackage{float}
\usepackage{overpic}   
\usepackage{subcaption} 

\theoremstyle{thmstyleone}%
\newtheorem{theorem}{Theorem}

\newtheorem{proposition}[theorem]{Proposition}
\newtheorem{lemma}[theorem]{Lemma}
\newtheorem{definition}[theorem]{Definition}
\newtheorem{assumption}[theorem]{Assumption}

\newtheorem{corollary}[theorem]{Corollary}
\newcommand\undersetbrace[2]{\underset{#1}{\underbrace{#2}}}
\newcommand\defeq\triangleq
\newcommand\defequiv\coloneqq

\begin{document}

\title[Article Title]{Co-Investment under Revenue Uncertainty Based on Stochastic Coalitional Game Theory}


\author*[1]{\fnm{Amal} \sur{Sakr}}
\email{amal.sakr@telecom-sudparis.eu}

\author[1]{\fnm{Andrea} \sur{Araldo}}\email{andrea.araldo@telecom-sudparis.eu}

\author[1]{\fnm{Tijani} \sur{Chahed}}\email{tijani.chahed@telecom-sudparis.eu}
\author[2]{\fnm{Daniel} \sur{Kofman}}\email{daniel.kofman@telecom-paristech.fr}


\affil*[1]{SAMOVAR, Télécom SudParis, Institut Polytechnique de Paris, 91120 Palaiseau, France}

\affil[2]{Télécom Paris, Institut Polytechnique de Paris, 91120 Palaiseau, France}



\abstract{The introduction of new services, such as Mobile Edge Computing (MEC), requires a massive investment that cannot be assumed by a single stakeholder, for instance the Infrastructure Provider (\textit{InP}). Service Providers (SPs) however also have an interest in the deployment of such services. We hence propose a co-investment scheme in which all stakeholders, i.e., the \textit{InP} and the SPs, form the so-called grand coalition composed of all the stakeholders with the aim of sharing costs and revenues and maximizing their payoffs. The challenge comes from the fact that future revenues are uncertain. We devise in this case a novel stochastic coalitional game formulation which builds upon robust game theory and derive a lower bound on the probability of the stability of the grand coalition, wherein no player can be better off outside of it. 
In the presence of \deleted{some correlated} \added{highly dependent} fluctuations of revenues however, stability can be too conservative. In this case, we make use also of profitability, in which payoffs of players are non-negative, as a necessary condition for co-investment\added{, and we derive a lower bound on the probability that co-investment is profitable}.
The proposed framework is showcased for MEC deployment, where computational resources need to be deployed in nodes at the edge of a telecommunication network. Numerical results show high lower bound on the probability of stability when the SPs' revenues are of similar magnitude, even with high levels of uncertainty. In the case where revenues are highly variable however, the lower bound on stability can be trivially low whereas co-investment is still profitable.}


\keywords{Co-investment, Uncertainty, Coalitional game, Stability, Profitability, \added{Mobile} Edge Computing}


\maketitle
\section{Introduction}\label{section: intro}

Several examples exist of technologies that, from a purely technical point of view, would be ready to emerge or scale up, but whose development is hindered by a big economic barrier.
This is the case of technologies requiring the deployment of a massive infrastructure, for which massive investments need to be engaged upfront, and whose precise future profits remain uncertain.
We call Infrastructure Provider (\textit{InP}) the actor responsible for deploying, maintaining and operating the physical infrastructure. We consider that, on top of such an infrastructure, third-party Service Providers (SPs)  run their applications and collect revenues from end users. 

An example of this setting is the electric vehicle charging infrastructure~\citep{badia2019investment}, where the electric grid operator, acting as the \textit{InP}, is tasked with deploying the grid infrastructure, e.g., bringing high-capacity electricity lines to road networks. On top of such grid infrastructure,  electric vehicle charging station providers, which take the role of SPs, can operate their business and get revenues from end-users (drivers)~\citep{fredriksson2019optimal}.

Another example, which is also the one  we consider in the numerical results,
involves the deployment and operation of \added{Mobile} Edge Computing (\added{M}EC)~\citep{satyanarayanan2017emergence}. \added{M}EC is a computing architecture that decentralizes data processing by bringing it to the edge of the network, e.g., base stations~\citep{shi2016edge}. This proximity to end-users reduces latency, improves bandwidth efficiency and enables real-time data processing and analytics \citep{mohan2020pruning}. \added{M}EC brings new revenue opportunities: third-party SPs, e.g., augmented or virtual reality applications~\citep{ong2013virtual}, \citep{reif2008augmented},
automated car function providers~\citep{faisal2019understanding}, \citep{parekh2022review}, can run low-latency applications on the edge nodes, benefiting from proximity to end-users, and collect revenues from them. 
In the \added{M}EC context, the \textit{InP} is the network operator~(NO), which is the only entity owning the nodes at the edge. Any SP willing to run its services at the edge has to pass through a NO. This is thus an unprecedented opportunity for NOs, who could capitalize on their privileged position of exclusive owners of edge nodes~\citep{khan2019edge}.
However, deploying and maintaining \added{M}EC infrastructure at network edges is highly costly~\citep{santoyo2018edge}. This is because \added{M}EC requires establishing and managing distributed computing resources across multiple locations. Setting up these infrastructures demands significant investment in hardware, networking, and ongoing maintenance, making large-scale implementation financially challenging. 
As a result, \added{M}EC has not yet been deployed at scale in today's network.

In the cases we consider in this paper, the \textit{InP} is reluctant to bear the significant capital (CAPEX) and operational expenditures (OPEX) alone.
This reluctance is compounded by the market dynamics, where the revenues generated from the edge services tend to flow to SPs, which directly interact with end users. Consequently, the \textit{InP} often faces challenges in capturing a share of the profits from these services, particularly in highly competitive markets. 
To remove this economic barrier, we propose a co-investment scheme in which players (or stakeholders), i.e., one \textit{InP} and several SPs, co-invest to share costs and revenues. By joining a coalition, they aim to maximize their payoff by deploying infrastructure with the optimal amount of resources and allocating it dynamically among players, based on how much revenue they could get by using the allocated resources.
In this scheme, SPs are strongly incentivized to contribute to infrastructure costs, as it enables them to provide high-quality or innovative services, attract new users, improve satisfaction, and thereby increase revenue.
 
The challenge is whether there is an opportunity for co-investment between the stakeholders under uncertainty. To do so, we first tackle the question of whether this co-investment is the most profitable investment that each stakeholder can make.
Determining this is highly complex due to the uncertainty of revenues, which depend themselves on the stochasticity of end-user demand. 
Demand is influenced by external factors such as market trends, technological shifts, and competitive dynamics, making future revenues inherently difficult to predict. Even small deviations from the expected demand can significantly alter potential revenues and, consequently, the stakeholders’ payoffs, making it difficult to determine whether co-investment is indeed the most profitable option for each of them.
To address this challenge, we cast the problem as a novel stochastic coalitional game, extending the principles of Robust Coalitional Game Theory (RCGT)~\citep{doan2014robust}.
In its basic form, RCGT requires coalitions to be formed and remain stable, i.e., no subset of players would be better off outside the coalition, even in the worst-case scenario. While this approach ensures robustness, it is overly restrictive in the context of co-investment, where risk is an inherent feature of the investments. Requiring no risk, as in the original RCGT framework, would effectively preclude any meaningful investment opportunities. 
To address this limitation, we reformulate RCGT to allow for an acceptable level of risk, and derive a lower bound on the probability that the grand coalition, composed of all the players, is stable, i.e., no subset of players is better off outside of the grand coalition. 

Second, in the case where revenues are highly variable, which stems from \deleted{possible correlations}\added{strong dependence} in time between demand fluctuations, even this probabilistic approach to stability can yield trivially low value for the lower bound.
In this case, profitability, defined as the condition where no player's individual payoff is negative, presents a more significant measure for the opportunity of co-investment. \added{To this end, we derive a lower bound on the probability that co-investment is profitable.} The profitability condition resembles the concept of satisficing~\citep{simon1997models}, where a decision is acceptable if it meets a given threshold, when that threshold is set to zero. In our case however, we determine the infrastructure plan (resources and shares) by maximizing total payoffs, not through a satisficing search. Profitability here reflects satisficing in outcome, but not in the decision process.

Our contributions are:
\begin{itemize}
    \item A co-investment framework under uncertainty:
    we introduce a co-investment model where an \textit{InP} and multiple SPs form a coalition to share the costs and the revenues under uncertainty. The uncertainty we consider in this paper derives from uncertain user demand, which in turn causes stochastic fluctuations in revenues. 
    \item A lower bound on grand coalition stability: we establish a probabilistic guarantee that the grand coalition is stable despite revenue fluctuations, and we prove this result within a stochastic coalitional game. \deleted{We also define and formulate the probability that the co-investment is profitable.}
    \added{\item A lower bound on co-investment profitability: we establish a probabilistic guarantee for profitable co-investment under revenue uncertainty.}
    \item Application to \added{Mobile} Edge Computing infrastructure: we apply the proposed framework to a\deleted{n} \added{M}EC setting, where an \textit{InP} and several SPs cooperate to deploy and maintain edge infrastructure over a defined investment period. Costs and revenues are shared among players.
\end{itemize}
We show numerically that when \deleted{the investment period is sufficiently long (e.g., 5 years) and }the SPs are of similar size, the lower bound on the probability of stability reaches~1. In cases where this lower bound is uninformative, the \added{lower bound on the} probability of profitable co-investment is also computed and found to be high. We release our code as open source to ensure reproducibility.\footnote{\url{https://github.com/amalsakr-tsp/Anor}}

The rest of this paper is as follows. Section~\ref{section: related} reviews related work. Section~\ref{Section: meth} delves into the co-investment scheme and coalitional game model. Section~\ref{Section: ana-results} provides results in the context of MEC. Section~\ref{section: conclusion}  presents the conclusion and perspectives.
\section{Related work} \label{section: related}

As our work consists of modeling co-investment using coalitional game theory with application to \added{M}EC deployment, our related work is categorized into three sub-sections: co-investment strategies, coalitional game theory and the economic aspects of \added{M}EC.

\subsection{Co-investment strategies}

Cooperation and infrastructure sharing among industry stakeholders have become essential to address the high costs of infrastructure deployment. The Connected Collaborative Computing Network~\citep{digitalinfrastructure} emphasizes the importance of cooperation across stakeholders (e.g., telecom, cloud) to reduce investment burdens.

Co-investment agreements have been studied under various structural and strategic assumptions. The author in~\citep{jeanjean2022co} formalizes co-investment as a form of network sharing, where two \textit{InP}s co-invest in infrastructure, including passive and active components. The results show that co-investment and infrastructure sharing reduce costs, boost investment and improve consumer welfare. 
In contrast, our work focuses on co-investment between heterogeneous players, an \textit{InP} and multiple SPs.

Strategic misalignment in co-investment is addressed in~\citep{lopez2022co}, where incumbents may increase reported deployment costs to deter entry, harming welfare. Our model instead emphasizes mutual benefit, ensuring transparency and aligning incentives so that all stakeholders are motivated to join the co-investment.

Infrastructure sharing under uncertainty and regulation is also studied in~\citep{INDERST201428}, where one firm invests in infrastructure and another decides later whether to access it, based on regulatory pricing and realized demand. However, we model pre-investment cooperation, where all stakeholders commit upfront to sharing the deployment and maintenance costs. Resources are then dynamically allocated among SPs based on their individual needs at each point in the investment period.

The impact of uncertainty on co-investment decisions is studied in~\citep{bourreau2021co}, which analyses how demand uncertainty affects the timing of investment between incumbents and entrants in network infrastructure. The authors show that allowing entrants to co-invest ex-post, after demand is revealed, can distort investment incentives and reduce network coverage. To address this, they propose combining ex-ante co-investment options, before demand is known, with ex-post risk premia. While their model focuses on bilateral co-investment, our work extends the analysis to coalitions involving heterogeneous players, one \textit{InP} and multiple SPs, using a stochastic coalitional game framework. Our setting aligned with the ex-ante case, where all players commit to co-investment before demand is realized, based on expectations.

The study in~\citep{azevedo2018rivalry} focuses on investment timing in a competitive duopoly under both revenue and cost uncertainty. Similarly,~\citep{kauppinen2018investing} examines a firm's sequential investment decisions, considering whether to initiate, suspend, or continue investment in response to uncertain future revenues and costs. However, our work considers only revenue uncertainty, with fixed timing, as stakeholders commit at the start of the co-investment strategy.

Cost, profit, and effort allocation mechanisms are examined in~\citep{bhaskaran2009effort} using a bargaining-based approach. This work focuses on incentive alignment and efficient cooperation between two firms. In contrast, we adopt the nucleolus in a stochastic environment to allocate payoff in a way that supports coalition stability under uncertainty.

The work in~\citep{he2024co} proposes a cooperative framework for infrastructure co-investment using Nash bargaining to share benefits, assuming predictable demand. While effective in structured settings, the model does not address demand uncertainty or the challenge of sustaining  cooperation under risk. Our framework accounts for stochastic demand and offers probabilistic guarantees on cooperation profitability.

\subsection{Coalitional game theory under uncertainty}
Several studies have explored cooperative game theory frameworks to analyze investment decisions and coalition stability under uncertainty. For instance, the work in~\citep{lin2021revenue} investigates how independent stakeholders working on a shared project can cooperate by reallocating their resources to speed up completion and earn extra rewards. They apply cooperative game theory to ensure the resulting gains are fairly shared, mainly through Shapley value and proportional sharing rules. While their setting is deterministic and focused on time savings in project delivery, the spirit of cooperation they promote aligns closely with our goal. We extend these ideas to a stochastic setting, where investment decisions and resource sharing must account for uncertain demand and fluctuating revenues.

In~\citep{ketelaars2023dynamic}, cooperative game theory is combined with real options analysis to address timing and coalition stability under uncertainty. The authors propose a proportional investment scheme to maximize project value while keeping investment timing stable. Their approach accounts for strategic deferral and uses option value to represent flexibility. In contrast, our framework assumes full upfront investment by all players and models uncertainty directly using a stochastic coalitional game, without relying on real options or proportional mechanisms.

The work in~\citep{kiedanski2020discrete} extends coalitional games to model cooperative investment, incorporating stochastic load profiles, to minimize expected costs. It proves the existence of a non-empty core, the set of all possible distributions of total profit where no subset of players would be better off forming their own coalition. However, the formulation is tailored to linear cost minimization, whereas our setting involves nonlinear, concave revenue maximization based on dynamic resource allocation.

The concept of the expected value core is proposed in~\citep{pantazis2023distributionally}, ensuring stability in expectation under probabilistic distributions. Similarly, the study in~\citep{akkarajitsakul2011coalition} applies a Bayesian coalitional game to model cooperation under uncertainty, where players form coalitions based on expected payoffs derived from beliefs about others’ behavior. However, in a co-investment context, players require stronger guarantees than stability in expectation. We instead introduce a nominal game, an expected version of the stochastic game, as an intermediate step, and we prove that the grand coalition is stable in this setting. Since this is not sufficient on its own, we further derive analytical lower bounds on the probability that the grand coalition is stable under stochastic demand.

To handle uncertainty robustly,~\citep{raja2021payoff} introduces the robust core, defined as the intersection of instantaneous deterministic cores across discrete time slots. While this approach ensures stability at each time slot under bounded uncertainty, our framework differs by modeling a unified coalitional game over the entire investment horizon, incorporating dynamic resource allocation across time.

A related approach is proposed in~\citep{borrero2016stable}, where the authors study coalitional games under uncertainty by introducing several stability concepts. One of their key contributions is the common core, which requires the existence of a single weighting over scenarios such that all coalitions are stable under the corresponding weighted game. Although their framework allows for partial information about these weights, it still relies on the existence of such a weighting to assess stability. In contrast, our model makes no assumptions about scenario weights.  Instead, we work with a bounded set of realizations and evaluate coalition stability uniformly over that set.

The Robust Coalitional Game Theory (RCGT) model in~\citep{doan2014robust} ensures grand coalition stability by evaluating payoffs under worst-case realizations. However, this strict robustness limits its practical applicability in uncertain co-investment contexts. We extend this work by replacing worst-case guarantees with probabilistic ones, computing a lower bound on the probability of stability that accounts for stochasticity in demand and revenue.

\subsection{Economic aspects of \added{Mobile} Edge Computing}
The economic viability of \added{M}EC depends on efficient
resource allocation, cost-sharing mechanisms, and investment models that account for uncertainty. 
Several studies focus on resource rental models. In~\citep{li2020data} and~\citep{xu2017zenith}, SPs rent computational resources from \textit{InP}s, based on historical demand or latency-aware workload optimization. While these works efficiently allocate existing resources, they assume the infrastructure is pre-deployed and entirely funded by the \textit{InP}s. In contrast, our model targets co-investment, where \textit{InP} and SPs jointly deploy and maintain infrastructure, shifting the focus to long-term profit sharing. Moreover, unlike~\citep{xu2017zenith}, we explicitly model demand uncertainty in coalition formation.

Deployment under demand uncertainty is studied in~\citep{nguyen2021two, chen2023building}, where an optimization framework is proposed to handle uncertain demand or user mobility. The model optimizes edge server locations and capacities before demand realization. However, both works assume unilateral decision-making by either a single SP~\citep{nguyen2021two} or a single \textit{InP}~\citep{chen2023building}, with SPs treated as passive users. Our approach differs by modeling a cooperative strategy where both \textit{InP} and SPs actively co-invest and share decision-making responsibilities under uncertainty.

Cost-efficient infrastructure deployment is also addressed in~\citep{zhao2018deploying}, which minimizes the number of edge nodes required for Internet of Things service coverage. However, this work neither considers which entities invest in infrastructure nor supports cost sharing among parties. Similarly, the work in~\citep{cong2022coopedge} proposes a cooperative model where base stations share resources via overlapping placement. Yet, the model does not clarify whether cooperation spans multiple \textit{InP}s or is internal to a single operator, nor does it consider joint financial investment.

A co-investment strategy is proposed in~\citep{patane2023coalitional}, where an \textit{InP} and multiple SPs deploy edge infrastructure. However, the model assumes static resource allocation and known user demands. Our work extends this by incorporating dynamic resource allocation based on fluctuating demand and providing probabilistic guarantees on co-investment profitability under uncertainty.

A preliminary version of this work has been presented in~\citep{sakr2025co} where we propose co-investment between the \textit{InP} and SPs under uncertainty, for the case where the stochastic fluctuations in the revenues are \deleted{uncorrelated} \added{independent}. In this work, however, we tackle the case where they are \deleted{correlated} \added{dependent} and establish a \deleted{necessary condition} \added{lower bound on the probability} of \deleted{the} profitab\added{le}\deleted{ility} \deleted{of} co-investment.

\section{Co-investment model}\label{Section: meth}
In this section, we first provide a schematic overview of the proposed
co-investment strategy (Section~\ref{sec:high-level}). We then model the
co-investment within a stochastic framework
(Section~\ref{sec:modeling}) and introduce the stochastic coalitional game
that represents it (Section~\ref{sec:stochastic-coalitional-game}). 
Next, we define the nominal game (Section~\ref{sec:nominal-game}), a
deterministic version of the stochastic game, and establish key properties
necessary for analyzing co-investment within the stochastic setting. 
After that, we present the main result of the paper: a lower bound on the
probability that the grand coalition is stable under uncertainty
(Section~\ref{sec:stability}), followed by a lower bound on the probability that the co-investment is profitable for all players~(Section~\ref{sec:profitability}). We then show the relation between the two lower bounds~(Section \ref{sec:relationship}). Finally, we apply the nucleolus concept to distribute costs and
revenues among the players (Section~\ref{sec:shapley}).  

\subsection{High-level description of the co-investment plan}
\label{sec:high-level} 
First, players combine their funds into a common wallet to deploy and maintain \deleted{the} \added{an} infrastructure \deleted{, consisting of a certain amount of resources, such as computational resources, with capacity C,} over an investment period $I$. \added{The infrastructure provides a certain amount of multiple resource types (e.g., computational resources) distributed across nodes.}
These resources are then dynamically partitioned among SPs in a manner that
we will make explicit in the sequel. Each SP receives a stochastic load of
requests from its end users and collects revenues by serving those
requests. Revenues increase with load and with the amount of resources
allocated to the SP. All revenues are then redistributed among the players.
Note that, although the \textit{InP} is the direct responsible for the
installation and the maintenance of the infrastructure,
it does not serve any user requests, and thus it does not collect any
revenues from users. However, since no resources can be installed without
it, the \textit{InP} benefits from the redistribution of the revenues from
SPs.

We aim to answer the following questions:
\begin{enumerate}
    \item \label{it:question1} Can we ensure, with high probability, that
    there is an opportunity for co-investment despite the uncertainty about
    the future loads?
    \item \label{it:question2} How to decide the amount of resources that
    should be installed and the dynamic resource sharing strategy?
    \item \label{it:question3} How to distribute costs and revenues
    among players?
\end{enumerate}
To answer Question~\ref{it:question1}, we formulate a novel stochastic game
theoretical framework, suited to the particular case of co-investment,
where a big cost must be borne at the beginning (to install resources), and
then revenues over the investment period are uncertain.
To answer Question~\ref{it:question2}, we resort to convex optimization to
determine the optimal amount of resources to be purchased as well as its
optimal dynamic sharing between players.
As for Question~\ref{it:question3}, we compute a redistribution of costs and
revenues based on the nucleolus.

\tablename~\ref{tab:notations} summarizes the frequently used notations in
this paper.
 
\begin{table*}
    \centering
    \caption{Notation}
    \label{tab:notations}
    \resizebox{0.9\textwidth}{!}{
    \begin{tabular}{l p{0.70\textwidth}}
    \toprule
    Notation & Definition \\
    \midrule
    
    $t$ & Specific time slot in the set $\pazocal T$
    \eqref{eq:h_k_m} \\

    $\pazocal T$ & Set of time slots in the investment period $I$
    \eqref{eq: generic expression of payof} \\

    $I$ & Investment period \eqref{eq:cost-sharing} \\

    \added{$K$} & 
    \added{Number of resource types~\eqref{eq:matrix_c}} \\

    \added{$M$} & 
    \added{Number of nodes~\eqref{eq:matrix_c}} \\

    \deleted{$C$}
    \added{$\mathbf C = \{C_{k,m}\}_{1\le k\le K,\,1\le m\le M}$}
    &
    \added{Installed} \deleted{C}\added{c}omputational resource capacity \added{matrix} (e.g. CPU, GPU, RAM, etc.)
    \eqref{eq:matrix_c} \\

    $\text{Cost}(I,\deleted{C}\added{\mathbf C})$
    & Deployment and maintenance costs over the investment period
    \eqref{eq:cost-sharing} \\

    $\deleted{h}\added{\mathbf h}_{i}^t$
    \added{$= \{h_{i,k,m}^t\}_{1\le k\le K,\,1\le m\le M}$}
    &
    Resource allocat\deleted{ed}\added{ion matrix} for SP $i$ at time slot $t$~\eqref{eq:h_k_m}\\

    $\deleted{\vec{h}} \added{\mathbf H}^t$
    \added{$= \{\mathbf h_{i}^t\}_{i\in\pazocal N} \in \mathbb R_{+}^{N \times K \times M}$}
    & \deleted{Resources allocated to players}
    \added{Multi-dimensional allocation array collecting all players' allocations} at time slot $t$
    \eqref{eq: generic expression of payof} \\

    $\mathbf H = \{\mathbf H^t\}_{t\in\pazocal T}$
    &
    Set of resource allocation \deleted{vectors} \added{arrays over all time slots}
    \eqref{eq: generic expression of payof} \\

    $\deleted{C}\added{\mathbf C}_{\pazocal S}^*,\;
\deleted{h_{i,}}\added{\{\mathbf H}_{\pazocal S}^{t*}\added{\}_{t \in \pazocal T}}$

    &
    Results of \added{the} optimization problem
    \eqref{equation: optimal solution}-\eqref{eq:constraint-positive} \\

    $u_{i,\omega}^t(\deleted{h} \added{\mathbf h}_i^t)$
    & Revenues collected by player~$i$ from end-users at time slot~$t$
    in realization~$\omega$ \eqref{equation:zero_ut} \\

    $\bar u_i^t(\deleted{h} \added{\mathbf h}_i^t)$
    & Expected value of the above
    \eqref{equation: expected_utility} \\

    $v_{\omega}(\pazocal{S})$
    & Coalition payoff in realization~$\omega$ for coalition~$\pazocal{S}$
    \eqref{eq: generic expression of payoff-realizations_s} \\ 

    $\bar{v}(\pazocal{S})$
    & Expected coalition payoff for $\pazocal{S}$ over all realizations
    \eqref{equation:v_bar_s} \\

    $z_{i,\omega,S}$
    & Deviation of player $i$'s revenue from expected value in realization
    $\omega$ \eqref{eq:additional} \\

    $\hat{\delta}$
    & \added{Stability} \deleted{Uncertainty threshold}~\eqref{equation: delta hat} \deleted{(Definition~\ref{def:uncertainty-bound})}  \\

    $\text{range}(u_i^t)$
    & Range of uncertainty in the utility of player $i$ at time slot $t$
    \eqref{eq:delta-u} \\

    $l_{i,\omega}^t$
    & Amount of requests for SP $i$ at time slot $t$ in realization~$\omega$
    \eqref{equation: utility function} \\

    $\bar{l}_i^t$
    & Expected amount of requests for SP $i$ at time slot $t$
    \eqref{equation: expectedload} \\

    \bottomrule
    \end{tabular}}
\end{table*}

\subsection{Stochastic modeling of co-investment}
\label{sec:modeling} 
Let~$\pazocal N$ be the set of players composed of one Infrastructure
Provider~(\textit{InP}) and~$N$ Service Providers~(SPs).
Let~$I$ be the investment period (e.g., $I=5$~years). We discretize time~$I$
in set~$\pazocal T$ of time slots, having each length~$\Delta$.
At each time slot~$t\in\pazocal T$, SP~$i$ is allocated a quantity
\deleted{$h$}
$\added{\mathbf h}_i^t$
of infrastructure resources.
Exploiting the allocated resources, SP~$i$ collects at each time slot
~$t\in\pazocal T$ revenue
$u_i^t(\deleted{h}\added{\mathbf h}_i^t)$
(which we also call utility), which is a random function of the amount of
resources $\deleted{h}\added{\mathbf h}_i^t$.
The more the allocated resources, the more revenues will be collected.
Revenue $u_i^t(\deleted{h}\added{\mathbf h}_i^t)$
is a random function in the sense that, for any value $\deleted{h}\added{\mathbf h}_i^t$, revenue $u_i^t(\deleted{h}\added{\mathbf h}_i^t)$
is a random variable defined on probability space~$(\Omega, \pazocal F,
\mathbb P)$, where~$\Omega$ is the sample space, $\pazocal F$ is the set of
events and $\mathbb P$ is the probability measure
~\citep{ghahramani2024fundamentals}. This randomness arises because revenue
$u_i^t(\deleted{h}\added{\mathbf h}_i^t)$
depends on the traffic load received by SP~$i$, which is inherently
uncertain and varies over time.
Let~$u_{i, \omega}^t(\deleted{h}\added{\mathbf h}_i^t)$
$\omega\in\Omega$, denote a realization of random variable $u_i^t(\deleted{h}\added{\mathbf h}_i^t)$.

A coalition is a subset~$\pazocal S \subseteq \pazocal N$ of players. If a
coalition is formed, its participants decide, altogether, the amount \deleted{$C$}\added{$\mathbf C$} of infrastructure resources to be deployed \added{and maintained over the investment period $I$}. 
\added{We denote by}
\begin{equation}
\added{
\mathbf C \defeq \{C_{k,m}\}_{1\le k\le K,\,1\le m\le M}\in\mathbb R_+^{K\times M}
}
\label{eq:matrix_c}
\end{equation}
\added{the installed capacity matrix across the $K$ resource types and $M$
nodes, where $C_{k,m}$ represents the total available amount of resource type
$k$ at node $m$ installed by the coalition.}
The coalition incurs an upfront cost
$\text{Cost}(I,\deleted{C}\added{\mathbf C})$, which incorporates the cost
of purchase and of maintenance of the resources over the investment period
~$I$. To cover this cost, each player~$i \in \pazocal{S}$ contributes an
initial payment~$p_i$. The contributions must collectively cover the entire
cost incurred by the coalition, that is:
\begin{equation}
\sum_{i \in \pazocal{S}} p_i = \text{Cost}(I,\deleted{C}\added{\mathbf C})
\label{eq:cost-sharing}
\end{equation}
$\text{Cost}(I,\deleted{C}\added{\mathbf C})$ is a deterministic function
that increases with $\mathbf C$
(the more the purchased resources, the higher the cost) and with~$I$ (the
longer the period of guaranteed maintenance, the higher the cost).
\footnote{In reality, the maintenance cost may be stochastic, but it can be
assumed to be deterministic in case the coalition (or \textit{InP} on its
behalf) signs a contract with a resource vendor and pays upfront to that
vendor the price of resources and a maintenance service guaranteed over
period~$I$.} We make the following reasonable assumption:
\begin{align}
\label{eq:cost-zero}
    \text{Cost}(I,\deleted{C}\added{\mathbf C})=0 \text{     if }
    \deleted{C}\added{\mathbf C}=\deleted{0} \added{\mathbf 0}
\end{align}
\added{where $\mathbf 0$ denotes the zero matrix, i.e., $C_{k,m}=0$ for all resource types $k=\{1 \dots K\}$ and nodes $m=\{1\dots M\}$.}

\added{For each player $i\in\pazocal N$ and time slot $t\in\pazocal T$, we
define the \emph{resource allocation matrix}}
\begin{equation}
\added{
\mathbf h_i^t \defeq \{h_{i,k,m}^t\}_{1\le k\le K,\,1\le m\le M}\in\mathbb R_+^{K\times M},
}
\label{eq:h_k_m}
\end{equation}
\added{where $h_{i,k,m}^t$ denotes the amount of resource type $k$ allocated to
player $i$ at node $m$ and time slot $t$.}

At each time slot~$t$, let us denote by~
$\deleted{\vec h}\added{\mathbf H}^t=\{\deleted{h}\added{\mathbf h}_i^t\}_{i\in\pazocal S}
\added{\in \mathbb{R}_+^{|\pazocal{S}|\times K \times M}}$
the multi-dimensional allocation array of resources allocated to each player.
With total amount $\mathbf C$ of resources, we have:
\deleted{$\vec h^t \mathbf{\cdot} \vec 1 =$} $ \sum_{i\in\pazocal N} h_{i\added{, k, m}}^t \leq  C_{\added{k,m}}$\added{$
, \forall k=1,\dots,K,\ \forall m=1,\dots,M$},
$\forall t\in\pazocal T$.

Let~
$\{\deleted{\vec h}\added{\mathbf H}^t\}_{t\in\pazocal T}$
be a sequence of such arrays over time slots. 
Note that while all players participate in the co-investment, only SPs
interact with end-users and generate revenue by serving requests. The
\textit{InP} is responsible for installing and maintaining the
infrastructure but does not serve any user requests directly. Motivated by
this distinction, we make the following assumption:

\begin{assumption}
\label{ass:inp-zero}
The \textit{InP} does not receive any infrastructure resources during the
investment period. That is,
\begin{equation}
\deleted{h}\added{\mathbf h}_{\textit{InP}}^t = \deleted{0} \added{\mathbf 0},
\quad \forall t \in \pazocal T
\label{equation:no_resources_inp}
\end{equation}
\added{equivalently, $h_{\textit{InP},k,m}^t = 0$ for all $k \in \{1,\dots,K\}$, $m \in \{1,\dots,M\}$, and $t \in \pazocal T$.}
\end{assumption}
To make our analysis realistic, we consider the following assumption:
\begin{assumption}
\label{ass:nonnegative}
Utilities are non negative. Moreover, player~$i\in\pazocal N$ does not
collect revenues from the use of the infrastructure, when the share of
resources it gets is~$0$, i.e.:
\begin{equation}
\begin{aligned}
u_{i,\omega}^t(\deleted{h_i^t} \added{\mathbf h_i^t})
&\ge 0, \quad \forall
\deleted{h}\added{\mathbf h}_i^t \added{\succeq} \deleted{\ge} \deleted{0}\added{\mathbf 0}; \quad
u_{i,\omega}^t(\deleted{h_i^t} \added{\mathbf h_i^t})
 = 0, \ \text{if }
\deleted{h}\added{\mathbf h}_i^t = \deleted{0}\added{\mathbf 0},
\\
&\forall i \in \pazocal N, \ t \in \pazocal T, \ \omega \in \Omega
\end{aligned}
\label{equation:zero_ut}
\end{equation}
\end{assumption}
\added{Here, inequalities and equalities involving $\mathbf h_i^t$ are interpreted
component-wise, that is, $h_{i,k,m}^t \ge 0$ and $h_{i,k,m}^t = 0$ for all
resource types $k$ and nodes $m$.
}

The following proposition restricts our attention to reasonable cases:
\begin{proposition}
In the coalition, the \textit{InP} is a veto player, i.e., a coalition
without it is pointless. 
Moreover, since we assume only SPs directly collect revenues by serving the
end-users via the infrastructure, then the set of all SPs, considered
altogether, is a veto player as well. 
\label{proposition:pointless}
\end{proposition}
The proof is provided in Appendix~\ref{appendix:pointless-propo}.

The payoff of a coalition, i.e., how much it is worth in terms of net
benefit, is given by the following two definitions.

\begin{definition}
If coalition~$\pazocal S\subseteq \pazocal N$ installs \deleted{an amount $C$}
\added{a capacity matrix $\mathbf C$}
of resources and allocates them among players applying set of resource allocation arrays
$\added{\{\mathbf H}^t\}_{t\in\pazocal T}$,
the \emph{value} of coalition~$\pazocal S$ under
$\deleted{\{\vec h}\added{\{\mathbf H}^t\}_{t\in\pazocal T}$,
denoted by~$v \left(
\pazocal S,\,
\deleted{C}\added{\mathbf{C}},\,
\deleted{\{\vec h}\added{\{\mathbf H}^t\}_{t\in\pazocal T}
\right)$,
is the payoff that it obtains, computed as the sum of collected revenues
minus the cost, i.e.,
\begin{equation}
v \left(
\pazocal S,\,
\deleted{C}\added{\mathbf{C}},\,
\deleted{\{\vec h}\added{\{\mathbf H}^t\}_{t\in\pazocal T}
\right)
=
\sum_{t \in \pazocal T} \sum_{i \in \pazocal S}
u_i^t\!\left(\deleted{h_i^t}\added{\mathbf{h}_i^t}\right)
- \mathrm{Cost}\!\left(I,\, \deleted{C}\added{\mathbf C}\right)
\label{eq: generic expression of payof}
\end{equation}

The quantity above is a random variable, as it depends on random variables
$u_i^t\!\left(\deleted{h}\added{\mathbf h}_i^t\right)$,
$\forall i\in\pazocal N,t\in\pazocal T$. Its realizations are:
\begin{align}
v_{\omega}\left(
\pazocal S,\,
\deleted{C}\added{\mathbf{C}},\,
\deleted{\{\vec h}\added{\{\mathbf H}^t\}_{t\in\pazocal T}
\right)
=
\sum_{t \in \pazocal T} \sum_{i \in \pazocal S}
u_{i, \omega}^t\!\left(\deleted{h_i^t}\added{\mathbf{h}_i^t}\right)
- \mathrm{Cost}\!\left(I,\, \deleted{C}\added{\mathbf C}\right)
\label{eq: generic expression of payoff-realizations-complete}
\end{align}
\end{definition}

\begin{definition}
\label{def:vtilde}
Let us assume that there is a criterion that associates, to each
coalition~$\pazocal{S} \subseteq \pazocal{N}$, an amount of resources
$\deleted{C}\added{\mathbf{C}}_{\pazocal{S}}^*$ to be installed and \added{ a set of resource} \deleted{an} allocation \added{arrays}
$\{ \deleted{\vec  h}\added{\mathbf H}_\pazocal S^{t*}\}_{t\in\pazocal T}
= \left\{ \{\deleted{h}\added{\mathbf h}_{i, \pazocal S}^{t*}\}_{i\in\pazocal S} \right\}_{t\in\pazocal T}
\added{
= \left\{ \left\{ \{h_{i,k,m}^{t*}\}_{1\le k\le K,\,1\le m\le M} \right\}_{i\in\pazocal S} \right\}_{t\in\pazocal T}}$
(this criterion will be specified in
Assumption~\ref{assumption: optimal}).
We define the payoff (or value) of coalition~$\pazocal S$ as the following
random variable:
\begin{equation}
    v (\pazocal S)
    \defeq
    v \left(\pazocal S, \deleted{C_\pazocal S^*} \added{\mathbf C_\pazocal S^*}, \deleted{\{\vec h}\added{\{\mathbf H}_{\pazocal S}^{t*}\}_{t\in\pazocal T}\right)
    \label{eq: generic expression of payof_S}
\end{equation}
Its realizations are:
\begin{align}
    v_\omega (\pazocal S)
    \defeq
    v_\omega \left(\pazocal S, \deleted{C}\added{\mathbf C}_\pazocal S^*, \deleted{\{\vec h}\added{\{\mathbf H}_{\pazocal S}^{t*}\}_{t\in\pazocal T}\right),
    && \forall \omega\in\Omega
    \label{eq: generic expression of payoff-realizations_s}
\end{align}

Set function~$ v(\cdot):2^\pazocal N\rightarrow \mathbb R$, where
~$2^{\pazocal N}$ is the set of all possible subsets of~$\pazocal N$, is
called \emph{value function}. Observe that function~$ v(\cdot)$ is a random
function.
\label{Definition:c_star}
\end{definition}

The main question we aim to answer is whether there is an opportunity for co-investment.
This is challenging due to two factors: the long investment period, and uncertainty in revenue fluctuations over time. Since payoffs are realized at the end of the period, they depend on cumulative random variations, making it infeasible to track all possible outcomes explicitly. 
To address this challenge, we model the co-investment scheme as a stochastic coalitional game and analyze first whether the grand coalition can be formed and remains stable under uncertainty. The players are willing to co-invest if the grand coalition is stable. In a second step, we study the case where uncertainty results from \deleted{correlated} \added{strong dependent} fluctuations in revenues,  and where the notion of stability becomes too conservative. Hence, we resort to the profitability condition. 

\subsection{Model of the stochastic coalitional game}\label{sec:stochastic-coalitional-game}
We borrow the following definition of stochastic game from~\citep[Definition 1]{charnes1976coalitional} and we add the concept of restriction of the stochastic game. Such a restriction is useful, since we will be able to prove some important properties for large enough restrictions, but not for all the possible realizations of the stochastic game.
\begin{definition} 
A stochastic game is a pair~$(\pazocal N, \pazocal V)$, where~$\pazocal N$ is the set of players and~$\pazocal V$, called \emph{uncertainty set}, 
is the set of realizations of random function~$v_\omega(\cdot)$:
     \begin{equation}
         \pazocal V \defequiv \{ v_\omega(\cdot):2^\pazocal N \rightarrow \mathbb R | \omega\in\Omega\}
         \label{eq:set-V}
     \end{equation}
     \label{def:stochastic-game}
For any subset~$\Omega'\subseteq \Omega$ of the sample space, we define a restriction~$(\pazocal N, \pazocal V(\Omega'))$ of the original game as a new stochastic game, with the same set of players and the set of possible value functions defined as:
\begin{equation}
    \pazocal V(\Omega') \defequiv \{ v_\omega(\cdot):2^\pazocal N\rightarrow \mathbb R | \omega\in\Omega' \}
    \label{eq:set-V-restricted}
\end{equation}
\end{definition}
The challenge is to evaluate the stability of the grand coalition under uncertainty. To show this, we need to show that a set of ``good'' allocations of payoff among players exist. 
The following definition from~\citep[Definition 4]{doan2014robust} formalizes what are the ``good'' payoff allocations.
\begin{definition}
\label{def:rcore}
The robust core of the stochastic game $(\pazocal N, \pazocal V)$ is:
\begin{equation}
\small
\begin{aligned}
rcore(\pazocal N,\pazocal V) \defequiv \Bigg\{ \vec y = (y_1,\dots,y_{|\pazocal N|}) \in \mathbb{R}^\pazocal N \ \Big|\ 
&\underbrace{\sum_{i=1} ^{|\pazocal N|} y_i = 1}_{\text{efficiency}},\underbrace{\sum_{i \in \pazocal S} v_\omega(\pazocal N)\cdot y_i \geq v_\omega(\pazocal S)}_{\text{stability}}, \\
&\forall v_\omega \in \pazocal V, \forall \pazocal S \subseteq \pazocal N 
\Bigg\}
\end{aligned}
\end{equation}
\end{definition}

Allocations~$\vec y$ in the robust core are ways in which the payoff of the grand coalition can be shared among players. Each such allocation is efficient, in the sense that 100\% of the payoff is divided among the players. They are stable, in the sense that, considering any subset~$\pazocal S\subseteq \pazocal N$ of players (including singleton subsets), players in~$\pazocal S$ are always better off joining the grand coalition, rather than forming a coalition between themselves, no matter the realization of payoff. This is why we define the grand coalition as \emph{stable}, if it respects the following definition.

\begin{definition}
Given a stochastic game~$(\pazocal N,\pazocal V)$, the grand coalition is stable if and only if~$\textit{rcore}(\pazocal N, \pazocal V)\neq \emptyset$, i.e., if and only if there exists at least one allocation of the payoff that is efficient and stable.
\label{def:stable-stochastic}
\end{definition}

Stability of the grand coalition $\pazocal N$ is obtained through non emptiness of the robust core. The model in~\citep{doan2014robust} however is based on worst-case assumption for the uncertainty. Instead, we aim to derive conditions for the stability of the grand coalition based on the subset of realizations for which this condition holds. We do so by introducing a probability measure for the existence of the grand coalition.

\begin{definition}
Given a stochastic game $( \pazocal N,\pazocal V)$, the grand coalition is stable with probability at least $\nu^{\text{LB}}$ if and only if there exists a set of realizations~$\Omega'$ such that:
\begin{itemize}
\item $v_{\omega}(\pazocal N)>0, \forall \omega \in \Omega'$,
\item The grand coalition is stable in game~$\left(\pazocal N, \undersetbrace{\text{\eqref{eq:set-V-restricted}}}{\pazocal V(\Omega')} \right)$ (which we call \emph{restriction}), i.e., $\textit{rcore}(\pazocal N,\pazocal V(\Omega'))\neq \emptyset$ and
\item $\mathbb P (\Omega') \deleted{=} \added{\ge} \nu^{\text{LB}}$
\end{itemize}
\label{definition: robust_restriction}
\end{definition}

The main result\added{s} we will prove in this paper (Theorem\added{s}~\ref{thm:main-theorem} \added{and~\ref{thm:prof}}) \deleted{is} \added{are} that the grand coalition is stable with probability at least~$\nu^{\text{LB}}$ \added{and that co-investment is profitable with probability at least~$\nu_{\pazocal{N}}^{\text{LB}}$}, which we can bring close to~$1$ under some conditions. 
To prove th\added{e}\deleted{i}s\added{e} result\added{s}, we need first to analyze the so-called ''nominal game'', which is a deterministic version of the aforementioned stochastic game. 

\subsection{Nominal game}
\label{sec:nominal-game}

The nominal version of our game is constructed by reasoning in terms of
expected values. Let:
\begin{equation}
    \bar u_i^t(\deleted{h}\added{\mathbf h}_i^t)
    \defeq \mathbb E_\omega
    \big[
    u_{i,\omega}^t(\deleted{h}\added{\mathbf h}_i^t)
    \big]
    \label{equation: expected_utility}
\end{equation}
be the expected utility collected at time slot~$t$ by player~$i$, if it is
allocated resources
$\deleted{h}\added{\mathbf h}_i^t$. The expected payoff of coalition~$\pazocal S$, when
it installs an amount $\mathbf C$ of resources which are dynamically shared by SPs following
$\deleted{\{\vec h}\added{\{\mathbf H}^t\}_{t\in\pazocal T}$, is
\begin{equation}
\begin{aligned}
\bar{v}(\pazocal S,\deleted{C} \added{ \mathbf C}, \deleted{\{\vec h}\added{\{\mathbf H}^t\}_{t\in\pazocal T})
&\defeq \mathbb{E}_\omega \big[
v_{\omega}(\pazocal S,\deleted{C} \added{ \mathbf C},\deleted{\{\vec h}\added{\{\mathbf H}^t\}_{t\in\pazocal T})] \\
&\overset{\eqref{eq: generic expression of payoff-realizations-complete}
\eqref{equation: expected_utility}}{=} 
\sum_{t \in \pazocal T} \sum_{i \in \pazocal S}
\bar u_i^t(\deleted{h}\added{\mathbf h}_i^t)
- \text{Cost}(I,
\deleted{C}
\added{\mathbf C})
\end{aligned}
\label{equation: expected payoff}
\end{equation}

By considering
\deleted{$C$}\added{$\mathbf C$}$_\pazocal S^*$
and
\{\deleted{$ \vec{h}$}\added{$\mathbf H$}$_{\pazocal S}^{t*} \}_{t \in \pazocal T}$
from Definition~\ref{def:vtilde}, the nominal value function of
coalition~$\pazocal S$ is given by:
\begin{equation}
\begin{aligned}
\bar{v}(\pazocal S) 
&\overset{\eqref{eq: generic expression of payoff-realizations_s}}{\defeq}
\mathbb{E}_\omega \big[v_\omega (\pazocal S) \big] \\
&=
\sum_{t \in \pazocal T} \sum_{i \in \pazocal S}
\bar{u}_i^t(\deleted{h} \added{\mathbf h}_{i,\pazocal S}^{t*})
- \text{Cost}(I,
\deleted{C}\added{\mathbf C}_\pazocal S^*)\end{aligned}    
\label{equation:v_bar_s}
\end{equation}

We now introduce the nominal game.
\begin{equation}
\begin{aligned}
\text{Nominal game: } 
&(\pazocal N, \underbrace{\bar v}_{\eqref{equation:v_bar_s}}), \\
&\text{where } \pazocal N \text{ is the same set of players as before and } \\ &
\bar v: 2^\pazocal N \rightarrow \mathbb{R} \text{ is the value function}.
\end{aligned}
\label{eq:deterministic-game}
\end{equation}

This game follows the standard definition of coalitional game theory~
\citep[Definition~257.1]{Osborne1994}.

We have previously introduced in Definition~\ref{Definition:c_star}, that
once a coalition $\pazocal S$ is formed, its players collectively decide on
an amount of resources
$\deleted{C}\added{\mathbf{C}}_\pazocal S^*$
to install and a sequence of dynamic resource allocation \added{array}s
$\{\deleted{\vec h}\added{\mathbf H}_\pazocal S^{t*}\}_{t\in \pazocal T}$.
The following assumption states the practical criterion adopted in this
paper to determine the amount of resources to deploy and how to allocate
them among the coalition's players.
\begin{assumption}
\label{assumption: optimal}
Given a coalition~$\pazocal S \subseteq \pazocal N$, the \textit{InP}
precomputes amount
$\deleted{C}\added{\mathbf C_{\pazocal S}^*}$
of resources that would be installed and \deleted{vectors} \added{the allocation arrays}
$\left\{\deleted{\vec h}\added{\mathbf H}_{\pazocal S}^{t*}\right\}_{t\in\pazocal T}$
with which resources would be dynamically shared, by solving the following
optimization problem:
\begin{align}
\deleted{C}\added{\mathbf C}_{\pazocal S}^*,
\left\{\deleted{\vec h}\added{\mathbf H}_{\pazocal S}^{t*}\right\}_{t \in \pazocal T} 
&=
\operatorname*{argmax}_{
\deleted{C}\added{\mathbf C},
\left\{\deleted{\vec h}\added{\mathbf H}^t \right\}_{ t \in \pazocal T}}
\bar{v}\big(\pazocal S,\deleted{C}\added{\mathbf C},
\{\deleted{\vec h}\added{\mathbf H}^t\}_{t\in\pazocal T}\big)
\label{equation: optimal solution}
\\
\text{such that } \quad
\sum_{i\in\pazocal S} 
h_{i\added{,k,m}}^t
&\le
C_{\added{k,m}},
\quad \added{\forall k\in K,\ \forall m\in  M,\ \forall t\in\pazocal T}
\label{eq:constraint-sum}
\\
h_{i\added{,k,m}}^t & \ge 0,
\quad \added{\forall k\in K,\ \forall m\in M,\ \forall i\in\pazocal S,\ \forall t\in\pazocal T}
\label{eq:constraint-positive}
\end{align} 
\end{assumption}

This problem falls under stochastic programming~\citep{birge2011introduction},
as it involves optimizing the expected payoff under uncertainty, assuming a
known and uniform distribution over demand realizations.

In the following proposition, we formally define how the quantities
$\deleted{C}\added{\mathbf C}^*_{\pazocal S}$
and
$\{\deleted{\vec h}\added{\mathbf H}_{\pazocal S}^{t*}\}_{t\in\pazocal T}$
are derived from the optimal solution (see
Assumption~\ref{assumption: optimal}), depending on whether the
\textit{InP} is part of the coalition~$\pazocal S$.

\begin{proposition}
\label{prop:allocation}
Let~$\pazocal S \subseteq \pazocal N$ be a coalition. The optimal resource
capacity
$\deleted{C}\added{\mathbf C}^*_{\pazocal S}$
and the dynamic allocations
$\deleted{h}\added{\mathbf h}_{i,\pazocal S}^{t*}$
for all~$i \in \pazocal S$, $t \in \pazocal T$, are determined as follows:
\begin{itemize}
    \item If the \textit{InP} does not belong to coalition~$\pazocal S$, then
    no resources are installed, i.e.:
\begin{equation}
\deleted{C}\added{\mathbf C}^*_{\pazocal S} = \deleted{0}\added{\mathbf 0},
\quad
\deleted{h}\added{\mathbf h}_{i,\pazocal S}^{t*} = \deleted{0}\added{\mathbf 0},
\quad \forall i \in \pazocal S,\ \forall t \in \pazocal T
\label{eq:InPnot}
\end{equation}

    \item If the \textit{InP} belongs to coalition~$\pazocal S$, 
    then resources
    $\deleted{C}\added{\mathbf{C}}_{\pazocal S}^*$
    can be potentially installed and used by players in~$\pazocal S$,
    according to the dynamic shares precomputed in
    \eqref{equation: optimal solution}. 
\end{itemize}
\end{proposition}
See Appendix~\ref{appendix:proof_allocation} for the proof of
Proposition~\ref{prop:allocation}.

The following proposition ensures that
$\deleted{C}\added{\mathbf C}_\pazocal S^*$
and
$\{\deleted{\vec h}\added{\mathbf H}_{\pazocal S}^{t*}\}_{t\in\pazocal T}$
are uniquely defined for every coalition~$\pazocal S \subseteq \pazocal N$,
and can be computed in polynomial time
~\citep[Section~11]{Boyd2004}.

\begin{proposition}
\label{prop:convexity-of-problem}
Problem~\eqref{equation: optimal solution}-\eqref{eq:constraint-positive} is convex, under the following assumptions:
\begin{enumerate}
\item Function~
$\deleted{h}\added{\mathbf h}_i^t$ $\rightarrow \bar{u}_i^t(\deleted{h_i^t} \added{\mathbf h_i^t})$
is increasing concave, $\forall i\in\pazocal N, t\in\pazocal T$.
\item Function~
\deleted{$C$} \added{$\mathbf C$} $\rightarrow \text{Cost}(I,\deleted{C}\added{\mathbf C})$ is convex.
\end{enumerate}
As a result, the optimal resources can be determined analytically by solving the Karush-Kuhn-Tucker (KKT) conditions~\citep[Section~5.5.3]{Boyd2004}.
\end{proposition}
The proof of Proposition~\ref{prop:convexity-of-problem} is provided in Appendix~\ref{appendix:prop_convex}.


We now introduce the notion of the core of the nominal game, which is the deterministic counterpart of the robust core introduced for the stochastic game. The core is the set of payoff allocations that are both efficient and stable, i.e., allocations that distribute the total payoff of the grand coalition among the players in such a way that no subset of players has an incentive to deviate and form a coalition on its own. The non-emptiness of the core implies that the grand coalition in the nominal game is stable~\citep{saad2009coalitional}. Formally, the core of the nominal game $(\pazocal N,\bar{v})$ is defined as~\citep[Definition~2.1]{maschler1979geometric}):

\begin{equation}
\label{eq:core}
\begin{aligned}
\textit{core}(\pazocal N,\bar{v}) \defequiv 
\bigg\lbrace 
\vec{x} \in \mathbb{R}^\pazocal N \ \bigg| \ 
&\sum_{i \in \pazocal N} x_i = \bar{v}(\pazocal N), \ 
\sum_{i \in \pazocal S} x_i \geq \bar{v}(\pazocal S),\forall \pazocal S \subsetneq \pazocal N 
\bigg\rbrace 
\end{aligned}
\end{equation}

Theorem \ref{thm:convexity-of-game} implies that the grand coalition $\pazocal N$ in the nominal game is stable.

\begin{theorem}
\label{thm:convexity-of-game}
Nominal game $(\pazocal N, \bar{v})$ (\eqref{eq:deterministic-game}) has a non-empty core. 
\end{theorem}
The proof of Theorem \ref{thm:convexity-of-game} is provided in Appendix~\ref{appendix:th_nominal}.

\subsection{Stability of the grand coalition under uncertainty}
\label{sec:stability}
After determining the optimal resource\deleted{s} \added{matrix} $\deleted{C}\added{\mathbf C}_{\pazocal S}^*$ to deploy and the corresponding dynamic allocation \added{array}s~$\{\deleted{\vec{h}}\added{\mathbf H}_{\pazocal S}^{t*}\}_{t \in \pazocal T}$ based on expected values, we must acknowledge that due to the stochasticity in the load, the actual revenue collected by the SP could deviate from the expected one. We introduce the following definition to quantify this deviation:
\begin{definition}
The revenue deviation collected by SP~$i$ over time in coalition $\pazocal S \subseteq \pazocal N$ in realization $\omega\in\Omega$ is defined as: 
\begin{equation}
z_{i,\omega, \pazocal S}=
\sum_{t \in \pazocal T} 
\left(
u_{i,\omega}^t(
\deleted{h}\added{\mathbf h}_{i, \pazocal S}^{t*}
)
-
\bar u_i^t(
\deleted{h}\added{\mathbf h}_{i, \pazocal S}^{t*}
)
\right)
\label{eq:additional}
\end{equation}
\label{def: add rev}
\end{definition}

Due to~\eqref{eq:InPnot} and Assumption~\ref{ass:nonnegative}, $z_{i,\omega,\pazocal S}=0$ if~$\textit{InP}\notin \pazocal S$.

The stability of the grand coalition~$\pazocal N$ in the stochastic game depends on whether~$z_{i, \omega, \pazocal N}$~\eqref{eq:additional} is within a certain bound, $\forall i \in \pazocal N, \forall \omega \in \Omega$~\citep[Theorem 2]{doan2014robust}. 
The bound on the probability that the grand coalition is stable can be computed based on the following definition, applicable to each player.

\begin{definition}
\label{def:uncertainty-bound}
The uncertainty bound of player~$i\in\pazocal N$ within coalition~$\pazocal N$ quantifies the likelihood that the revenue deviation~$z_{i,\omega,\pazocal N}$ remains within a predefined threshold~$\hat{\delta}$. 
Formally, it is given by:

\begin{align}
\mathbb{P}_{i,\pazocal N}^{\text{LB}} 
&\defeq
\mathbb P (|z_{i,\omega,\pazocal N}|< \hat{\delta})
\label{equation: lb}
\end{align}
where~$\hat{\delta}>0$
is designed as a restrictive threshold that ensures both efficiency and stability within the grand coalition and is defined as:
\begin{equation}
    \hat{\delta} \defeq \min 
    \left ( \frac{\Bar{v}(\pazocal N)}{|\pazocal N|}, \frac{\hat{\sigma}(\pazocal{N}, \Bar{v})}{\max\limits_{\substack{\pazocal{S} \subsetneq \pazocal{N} \\ \pazocal{S} \neq \emptyset}} \left( |\pazocal{S}|  +  (|\pazocal{N}| - 2  |\pazocal{S}|) \cdot y_{\hat{\sigma}}(\pazocal{S}) \right)}\right)
    \label{equation: delta hat}
\end{equation}
\end{definition}
To better understand this bound, let us break it down:
\begin{enumerate}
    \item The first term serves as a uniform upper bound on the allowed deviation per player. It ensures that no individual perturbation exceeds the average share of the grand coalition's value. It also ensures the efficiency condition, because $\bar{v}(\pazocal{N})$ is fully distributed among the players in the grand coalition~$\pazocal{N}$.
    \item The second term quantifies how much better off players are within the grand coalition compared to any possible subset, by incorporating the stability value~$\hat{\sigma}(\pazocal N, \bar{v})$. 
    The denominator normalizes this effect by considering both coalition size and the extent of payoff variability. We define the stability value~$\hat{\sigma}(\pazocal N, \bar{v})$ as follows:
\begin{equation}
    \hat{\sigma}(\pazocal N, \bar{v}) \defeq \min_{\pazocal S \subsetneq \pazocal N} \left( \sum_{i\in\pazocal S} \bar{ x}_i - \bar{v}(\pazocal S) \right)
    \label{equation:stab_hut}
\end{equation}
where $\bar{x}_i$ denotes the payoff assigned to player $i$, with $\bar{x} = (\bar{x}_i)_{i \in N}$ being a payoff allocation in the core of the nominal game $(\pazocal N, \bar{v})$, distributing the expected payoff $\bar{v}(\pazocal N)$ among the players.

We now explain the intuition behind~$\hat{\sigma}(\pazocal N, \bar{v})$~\eqref{equation:stab_hut}:
\begin{enumerate}
    \item If the sum of expected payoffs of coalition~$\pazocal S$ exceeds its coalition payoff~$\bar{v}(\pazocal S)$, then the players in~$\pazocal S$ have a stronger incentive to remain in the grand coalition~$\pazocal N$.
    \item By minimizing over all coalitions~$\pazocal S$, this function identifies the weakest coalition configuration, ensuring that the uncertainty bound accounts for the most unstable scenario.
\end{enumerate}

To quantify how the payoff of coalition~$\pazocal S$ compares to its payoff in the grand coalition~$\pazocal N$, we introduce the normalized payoff~$y_{\hat{\sigma}}(\pazocal S)$:
\begin{align}
    y_{\hat{\sigma}}(\pazocal S) & \defeq 
    \begin{cases} 
        \frac{\bar{v}(\pazocal S) + \hat{\sigma}(\pazocal N, \bar{v})}{\bar{v}(\pazocal N)} & \text{for all } \pazocal S \subsetneq \pazocal N, \\
        1 & \text{for } \pazocal S = \pazocal N
    \end{cases}
\end{align}
If a coalition~$\pazocal S$ has a coalition payoff~$\bar{v}(\pazocal S)$ significantly greater than its payoff in~$\pazocal N$, its players might have an incentive to leave and form their own coalition, potentially reducing the stability of the grand coalition.

\end{enumerate}

We include the following Lemma to justify our use of~$\hat{\delta}$, showing that it remains safely below the threshold used in earlier work and reinforces the soundness of our approach.
\begin{lemma}
\label{lemma: tightness-bound}
The threshold $\hat{\delta}$~\eqref{equation: delta hat} is tighter than the one presented in~\citep[Theorem~2]{doan2014robust}, i.e,
\begin{equation}
    \hat{\delta} < \delta 
    \label{eq:tighter-bound}
\end{equation}
where $\delta$ denotes the threshold introduced in~\citep[Theorem~2]{doan2014robust}.
\end{lemma}
The proof of Lemma~\ref{lemma: tightness-bound} is provided in Appendix~\ref{appendix:delta_hut}.

We now present the \added{first} main result of our work, which establishes a lower bound on the probability of stability of the grand coalition. We first address the case where, for each SP~$i$, the utilities $u_i^t(\cdot)$ may be dependent across time slots $t \in \pazocal{T}$.
\added{We then consider the case where} the utilities $u_i^t(\cdot)$ are independent \added{and bounded} across time slots $t \in \pazocal{T}$. \added{In some cases, it is nevertheless reasonable to assume boundedness: (i)~if the communication channel capacity is finite, the requests that arrive to the SP cannot grow unbounded, and (ii)~if the willingness to pay of the end users is bounded (as is reasonable to assume), an SP setting prices tending to infinity would become completely non-competitive.} \deleted{This temporal independence is required to apply}\added{In the dependent case, we use Chebyshev’s inequality~\citep[Theorem~18.2]{harchol2023introduction}, whereas in the independent and bounded case, we rely on} Hoeffding’s inequality~\cite[Theorem~2.8]{Boucheron2013}.\deleted{, which}\added{ Both inequalities} provide\deleted{s} a bound on the deviation of the total revenue from its expected value~\eqref{eq:additional} (\added{i.e.,} on the sum of utility fluctuations over time).

\added{Moreover}\deleted{Second}, we assume that utilities of the SPs are independent from each other. This is justified by the fact that SPs often serve distinct user populations and handle traffic from unrelated applications. 
Consequently, the uncertainties influencing their revenues arise from separate sources. 
This cross-player independence enables us to express the joint probability that all revenue deviations~\eqref{eq:additional} remain within a threshold~$\hat{\delta}$~\eqref{equation: delta hat} as a product over individual probabilities, which is a key step in deriving the final bound on coalition stability.


\begin{theorem}
\label{thm:main-theorem}
Let \deleted{us assume that} $u_i^t(\cdot)$ \deleted{are independent}\added{denote the} random \added{utility }function\deleted{s} \added{of player $i \in \pazocal N$ and time slot $t \in \pazocal T$, and let $\hat{\delta}>0$ denote the coalition stability threshold defined in \eqref{equation: delta hat}} \deleted{and that they are bounded, $\forall i\in\pazocal N, t\in\pazocal T$}. \added{Then}\deleted{I}\added{i}n stochastic game ($\pazocal N, \pazocal V$), the grand coalition is stable with probability at least $\nu^\text{LB}$ (see Definition~\ref{definition: robust_restriction}), where \added{\( \nu^{\mathrm{LB}} \) can be obtained as follows}:
\begin{enumerate}
    \item \added{If $u_i^t(\cdot)$ are random functions, possibly dependent, with finite variance $\forall i\in\pazocal N, t\in\pazocal T$. Then:
    \begin{equation}
    \nu^\text{LB}
    \defeq
    \prod_{i \in \pazocal N}
    \left(
    1 - 
    \min\left(
    \frac{\sigma_i^{2}}
    {\hat{\delta}^{2}}
    ,1\right)
    \right), 
    \label{eq:lower_bound_chebyshev}
    \end{equation}
    where
    \begin{equation}
    \label{eq:var-u}
    \sigma_i^{2}
    \defeq
    \mathrm{Var}\!\left(
    z_{i,\omega,\pazocal{N}}
    \right)
    \overset{\eqref{equation: lb}}{=}
    \mathrm{Var}\!\left(
    \sum_{t\in\pazocal T}
    \bigl(
    u_{i,\omega}^t(\mathbf h_{i,\pazocal{N}}^{t*})
    -
    \bar{u}_i^t(\mathbf h_{i,\pazocal{N}}^{t*})
    \bigr)
    \right)
    \end{equation}
    represents the variance of the revenue deviation $z_{i,\omega,\pazocal{N}}$ of player~$i\in\pazocal{N}$ over all time slots~$t\in\pazocal T$, $u_{i,\omega}^t(\mathbf h_{i,\pazocal{N}}^{t*})$ is the revenue collected by player $i$ from end-users using the allocated resources $\mathbf h_{i,\pazocal{N}}^{t*}$, in time slot $t$ in realization $\omega$, and $\bar{u}_i^t(\mathbf h_{i,\pazocal{N}}^{t*})$ is the corresponding expected value.}
    \item \added{If $u_i^t(\cdot)$ are independent random functions and are bounded $\forall i \in \pazocal{N}, t \in \pazocal{T}$. Then:}
    \begin{equation}
    \nu^\text{LB}
    \defeq
    \prod_{i \in \pazocal N}
    \deleted{\max} 
    \left(
    1 - \added{\min \left(}
    2 
    \exp\frac{-2\hat{\delta}^{2}}
    {
    \sum_{t\in\pazocal T} 
    (\text{range}(u_i^t) )^2
    },\deleted{0}
    \added{1 \right)}
    \right) 
    \label{equation: lower bound}
    \end{equation}
    where:
    \begin{equation}
    \label{eq:delta-u}
    \text{range}(u_i^t)
    \defeq
    \max_{\omega\in\Omega} u_{i,\omega}^{t}(\deleted{h}\added{\mathbf h}_{i,\pazocal N}^{t*})
    - 
    \min_{\omega\in\Omega} u_{i,\omega}^{t}(\deleted{h}\added{\mathbf h}_{i,\pazocal N}^{t*})
    \end{equation}
    represents the range of uncertainty of revenues collected by player~$i\in\pazocal N$ at time slot~$t\in\pazocal T$.
\end{enumerate}
\end{theorem}

\begin{proof}
We first find uncertainty bounds~$\mathbb{P}_{i,\pazocal N}^{\text{LB}},\forall i\in\pazocal N$, according to Definition~\ref{def:uncertainty-bound}. 
Since the \textit{InP} does not directly interact with end users, it does not receive any resource allocation, i.e., $\deleted{h}\added{\mathbf h}_{\text{\textit{InP}}, \pazocal N}^{t*} = \deleted{0}\added{\mathbf 0}$ for all $t \in \pazocal T$~\eqref{equation:no_resources_inp}, as stated in Assumption~\ref{ass:inp-zero}. Consequently, via~\eqref{equation:zero_ut} and~\eqref{eq:additional}, we have ~$z_{\text{\textit{InP}}, \omega, \pazocal N} = 0, \forall \omega \in \Omega$, which implies that~$\mathbb{P}_{\text{\textit{InP}}, \pazocal N}^{\text{LB}} = 1$~\eqref{equation: lb}.
Let us now consider~$i\neq$ \textit{InP}. 
\added{We first express the probability that the revenue deviation $|z_{i,\omega,\pazocal N}|$~\eqref{eq:additional} does not exceed the threshold $\hat{\delta}$~\eqref{equation: delta hat}, namely,
\(
\mathbb{P}\big(|z_{i,\omega,\pazocal N}| < \hat{\delta}\big)
\) \eqref{equation: lb}.}
Since values~$h_{i, \pazocal N}^{t*}$ have been precomputed~(Assumption~\ref{assumption: optimal}), random variables~$z_{i,\omega,\pazocal N}$ are statistically independent among players. Therefore:
\begin{equation}
    \begin{aligned}
        \mathbb{P}(|z_{i,\omega, \pazocal N}|\textless \hat{\delta}, \forall i \in \pazocal N) 
        = 
        \prod_{i \in \pazocal N} \left(1 - \mathbb{P}(|z_{i, \omega, \pazocal N}| \ge \hat{\delta})\right) 
        \end{aligned}
\label{eq:prod-of-Pi}
\end{equation}
Let $\Omega' \subseteq \Omega$ be the set of realizations in which all players' revenue deviations remain within the threshold~$\hat{\delta}$, that is,
\begin{equation}
\Omega' \triangleq \left\{ \omega \in \Omega : |z_{i,\omega,\pazocal N}| < \hat{\delta},\ \forall i \in \pazocal N \right\}    
\end{equation}
and applying the independence of the random variables $z_{i,\omega,\pazocal N}$ across players, we can write:
\begin{equation}
\mathbb{P}(\Omega') \defeq \mathbb{P}(|z_{i,\omega, \pazocal N}|\textless \hat{\delta}, \forall i \in \pazocal N) \overset{\added{\eqref{eq:prod-of-Pi}}}{=}
\prod_{i \in \pazocal N} \left(1 - \mathbb{P}(|z_{i,\omega,\pazocal N}| \ge \hat{\delta})\right)
\label{eq:prob_cheb_prod}
\end{equation}
\added{We now use two concentration inequalities to obtain $ \mathbb{P}(|z_{i,\omega,\pazocal N}| \ge \hat{\delta})$:}
\begin{enumerate}
    \item \added{By assumption, the utility functions $u_i^t(\cdot)$ have finite variance for each player $i$ and each time slot $t \in \pazocal{T}$. Random variables $u_{i,\omega}^t(\mathbf h_{i,\pazocal{N}}^{t*})$ admit expectations $\bar{u}_i^t(h_{i,\pazocal N}^{t*})$ and finite variances. Hence, the revenue deviation $z_{i,\omega,\pazocal{N}}$~\eqref{eq:additional} is a finite-variance random variable, being a finite sum over time slots of finite-variance terms. Therefore, it satisfies the conditions for applying Chebyshev’s inequality \citep[Theorem 18.2]{harchol2023introduction}: namely, it is a random variable with finite variance $\sigma_i^2$ as defined in~\eqref{eq:var-u}. The probability that $|z_{i,\omega,\pazocal N}|$ exceeds a given threshold $\hat{\delta}$~\eqref{equation: delta hat} is thus bounded via Chebyshev’s inequality as:
\begin{align}
\mathbb{P}\!\left(\left|z_{i,\omega,\pazocal N} 
      - \mathbb{E}_{\omega}[z_{i,\pazocal N}]\right|\ge \hat{\delta}\right) = \mathbb{P}\!\left(|z_{i,\omega,\pazocal N}|\ge \hat{\delta}\right)
\le
\min\!\left(
      \frac{\sigma_i^{2}}{\hat{\delta}^{2}},\,
      1
\right),
\qquad
\forall i\in\pazocal N,
\label{eq:proba_under_chebyshev}
\end{align}
where $z_{i,\pazocal N}$ is the random variable corresponding to the
realizations $z_{i,\omega,\pazocal N}$.  
Its expectation is
\begin{equation}
\mathbb{E}_{\omega}[z_{i,\pazocal N}]
=
\sum_{t\in\pazocal T}
\Big(
\mathbb{E}_{\omega}\big[u_{i}^{t}(\mathbf h_{i,\pazocal N}^{t*})\big]
-
\bar u_i^{t}(\mathbf h_{i,\pazocal N}^{t*})
\Big)
= 0,
\end{equation}
Then, substituting~\eqref{eq:proba_under_chebyshev} into~\eqref{eq:prob_cheb_prod},  we obtain the following lower bound:
\begin{equation}
 \mathbb{P}(\Omega') \ge 
\prod_{i \in \pazocal N} \left(1 - \min\left(\frac{\sigma_i^{2}}{\hat{\delta}^{2}}, 1 \right)\right) \defeq \nu^{\text{LB}} 
\label{eq:chebyshev}
\end{equation}}
\item By assumption, the utility functions $u_i^t(\cdot)$ are independent across time slots $t \in \pazocal{T}$ for each player $i$, and the random variables $u_{i,\omega}^t(\deleted{h}\added{\mathbf h}_{i,\pazocal{N}}^{t*})$ are bounded. It follows that the deviations $u_{i,\omega}^t(\deleted{h}\added{\mathbf h}_{i,\pazocal{N}}^{t*}) - \bar{u}_i^t(\deleted{h}\added{\mathbf h}_{i,\pazocal{N}}^{t*})$ are independent, zero-mean, and bounded for each $t$. 
Therefore, revenue deviation $z_{i,\omega,\pazocal{N}}$~\eqref{eq:additional} satisfies the conditions for applying Hoeffding’s inequality ~\cite[Theorem~2.8]{Boucheron2013}: namely, it is a sum of independent, bounded random variables. The probability that $|z_{i,\omega,\pazocal{N}}|$ exceeds a given threshold $\hat{\delta}$~\eqref{equation: delta hat} is thus bounded \added{ via Hoeffding’s inequality} as:
\begin{align}
\mathbb{P}(|z_{i,\omega, \pazocal N}|\ge \hat{\delta})
\le
\min \left(
2\exp
\frac{-2 \hat{ \delta}^{2}}
{\sum_{t\in\pazocal T} 
(\text{range}(u_i^t))^2}
        , 1
        \right),
&& \forall i\in\pazocal N
\label{equation: proba uncer}
\end{align}
Then, \added{substituting} \deleted{using the upper bound from Hoeffding’s inequality}~\eqref{equation: proba uncer} \added{into~\eqref{eq:prob_cheb_prod}}, we obtain the following lower bound:
\begin{equation}
\mathbb{P}(\Omega') 
\ge 
\prod_{i \in \pazocal N} \left(1 - \min\left(2 \exp\left( \frac{-2 \hat{\delta}^2}{\sum_{t \in \pazocal T} (\text{range}(u_i^t))^2} \right), 1 \right)\right) \added{\defeq \nu^{\text{LB}}}
\label{eq:probability-omega-prime}
\end{equation}
\end{enumerate}
If, for all realizations, we have that~$\added{|}z_{i,\omega,\pazocal N} \added{|}< \hat{\delta}$, then the grand coalition is stable in the stochastic game~\citep[Theorem 2]{doan2014robust}. This condition is verified for all realizations in~$\Omega'$. Therefore, we can affirm that the grand coalition~$\pazocal N$ is stable in stochastic game~$(\pazocal N,\pazocal V(\Omega'))$ (according to Definition~\ref{definition: robust_restriction}). Thanks to \added{\eqref{eq:chebyshev} and} \eqref{eq:probability-omega-prime}, we obtain the theorem.
\end{proof}

\subsection{Profitability of the co-investment under uncertainty}\label{sec:profitability}
Theorem~\ref{thm:main-theorem} provides a lower bound on the probability of grand coalition stability\deleted{under the assumption that revenues are uncorrelated and bounded}. When revenues are \deleted{correlated} \added{dependent}, \deleted{one can instead} \added{one can} apply Chebyshev’s inequality~\added{Theorem~\ref{thm:main-theorem}~(1)}.\deleted{\citep{feller1991introduction}} However, in the presence of strong temporal \deleted{correlations} \added{dependence}, Chebyshev’s bound may become too loose to be informative. This raises a key question: Can co-investment remain profitable for all players despite high revenue variability?

Answering this question requires verifying a necessary condition, namely, whether the co-investment yields a non-negative payoff for each player. Specifically, we must evaluate whether the individual payoff $x_{i,\omega}$ for each player $i \in \pazocal{N}$, as well as the realized coalition payoff $v_{\omega}(\pazocal{S})$, for all $\omega \in \Omega$ and all $\pazocal{S} \subseteq \pazocal{N}$, is non-negative.
To formalize what we mean by a profitable co-investment, we introduce the following definition.
\begin{definition}
We say that the co-investment of a coalition~$\pazocal S \subseteq \pazocal N$ is profitable for all its players with probability at least~$\nu_{\pazocal S}^{\text{LB}} \in [0,1]$, if the probability that every player~$i \in \pazocal S$ receives a non-negative payoff is at least~$\nu_{\pazocal S}^{\text{LB}}$, i.e.,
\begin{equation}
    \mathbb{P}\left( x_{i,\omega} \geq 0, \forall i \in \pazocal S \right) \geq \nu_{\pazocal S}^{\text{LB}}
\end{equation}
where~$x_{i,\omega}$ denotes the random payoff of player~$i$ in realization~$\omega$, resulting from the co-investment by coalition~$\pazocal S$.
\label{definition:profitable_coinv}
\end{definition}

\added{
We now present the second main result of our work, which establishes a lower bound on the probability of profitable co-investment.
\begin{theorem}
\label{thm:prof}
Let \( \{ u_{i,\omega}^t(\cdot) \}_{i \in \pazocal{N},\, t \in \pazocal{T}} \) 
be non-negative random variables. 
We define the total utility in realization \( \omega \in \Omega \) as
\begin{equation}
    U_{\omega} \defeq \sum_{t \in \pazocal{T}} \sum_{i \in \pazocal{N}} 
    u_{i,\omega}^t\bigl(\mathbf h_{i,\pazocal{N}}^{t*}\bigr),
    \label{eq:ut}
\end{equation}
and denote by \(U\) the corresponding random variable. 
Let \( \mathbb{E}_{\omega}[U] \) and \( \mathrm{Var}_{\omega}(U) \) be,
respectively, the expected value and variance of \(U\) over all realizations 
\(\omega \in \Omega\). 
Assuming that \( \mathrm{Var}_{\omega}(U) < \infty \), 
the co-investment is profitable with at least the probability:
\begin{equation}
\nu_{\pazocal{N}}^{\text{LB}}
\;\defeq\;
\frac{
    1
}{
    1
    +
    \dfrac{
        \mathrm{Var}_{\omega}(U)
    }{
        \bigl(\mathbb{E}_{\omega}[U] - \text{Cost}(I,\mathbf C_{\pazocal N}^*)\bigr)^{2}
    }
}
\end{equation}
Here, $\text{Cost}(I,\mathbf{C}_{\pazocal N}^*)$ is the total cost~\eqref{eq:cost-sharing}.
\end{theorem}
\begin{proof}
Based on Definition~\ref{definition:profitable_coinv}, the co-investment is profitable if every player in the coalition receives a non-negative payoff, i.e., $x_{i, \omega}\ge 0, \forall i \in \pazocal{N}$. This is possible when the total coalition payoff, denoted $v_\omega(\pazocal{N})$, is itself positive; that is, when the coalition’s total utility exceeds its total cost, i.e., 
\begin{equation}
    U_\omega > \text{Cost}(I,\mathbf C_{\pazocal N}^*), \quad \forall \omega \in \Omega
    \label{eq:profit_condition}
\end{equation}
In other words, the co-investment is profitable whenever~\eqref{eq:profit_condition} holds. 
We are therefore interested in the probability
\begin{equation}
    \mathbb{P}\bigl(U_\omega > \text{Cost}(I,\mathbf C_{\pazocal N}^*)\bigr)
    \label{eq:probability}
\end{equation}
To bound this probability, we rely on the Paley–Zygmund inequality \citep[Page~8]{kahane1985some}(see also \citep[Section 3]{ghosh2002probability}), applied to the non-negative random variable $U$ (which has finite variance by assumption). 
For any $\theta \in [0,1]$, the inequality yields
\begin{equation}
    \mathbb{P}\bigl(U_{\omega} > \theta\,\mathbb{E}_{\omega}[U]\bigr)
    \;\geq\;
    \frac{(1-\theta)^2\,(\mathbb{E}_{\omega}[U])^2}{(1-\theta)^2\,(\mathbb{E}_{\omega}[U])^2 + \text{Var}_\omega(U)}
    \label{eq:paley_zygmund}
\end{equation}
To obtain the probability in~\eqref{eq:probability}, 
we define 
\begin{equation}
    \theta \defeq \frac{\text{Cost}(I,\mathbf C_{\pazocal N}^*)}{\mathbb{E}_\omega[U]}
    \label{eq:theta}
\end{equation}
We now show that $\theta \in [0,1]$ as follows. The optimization problem in Assumption~\ref{assumption: optimal} always admits the 
\emph{do-nothing} feasible solution
\[
\mathbf C = \mathbf 0,
\qquad 
\mathbf h_i^t = \mathbf 0,\ \forall i\in\pazocal S,\ \forall t\in\pazocal T, \quad \text{component-wise}
\]
which trivially satisfies constraints~\eqref{eq:constraint-sum}--\eqref{eq:constraint-positive}.  
Under this choice, no capacity is installed and no resources are allocated at any time; hence
no user of coalition~$\pazocal N$ receives service.  
Therefore, the expected utilities are all zero, and there is no cost to pay, i.e., \(\text{Cost}(I,0)=0\), then the value of the objective at the do-nothing solution becomes zero.
Since the do-nothing solution is feasible, we have the optimal value of the optimization problem in~\eqref{equation: optimal solution}:
\begin{equation}
\bar{v}(\pazocal N)
\;\ge\;
0
\label{eq:pos_v}
\end{equation}
where $\bar{v}(\pazocal N)$ is defined in~\eqref{equation:v_bar_s} as: 
\begin{equation}
    \bar{v}(\pazocal N) = \mathbb{E}_{\omega}[U] - \text{Cost}(I,\mathbf C_{\pazocal N}^*)
\end{equation}
with $\mathbb{E}_{\omega}[U]= \sum_{t \in \pazocal T} \sum_{i \in \pazocal N} \bar{u}_{i}^{t}(h_{i,\pazocal N}^{t*})$.
Therefore, the coalition value is always positive in expectation~\eqref{eq:pos_v},
which proves ${E}_{\omega}[U] \ge \text{Cost}(I,\mathbf C_{\pazocal N}^*)$. Hence $\theta \in [0,1]$. Substituting this choice of $\theta$~\eqref{eq:theta} into~\eqref{eq:paley_zygmund} gives:
\begin{equation}
    \mathbb{P}\bigl(U_\omega > \text{Cost}(I,\mathbf C_{\pazocal N}^*)\bigr)
    \;\geq\;
    \frac{\bigl(1-\tfrac{\text{Cost}(I,\mathbf C_{\pazocal N}^*)}{\mathbb{E}_\omega[U]}\bigr)^2 \, \bigl(\mathbb{E}_\omega[U]\bigr)^2}
         {\bigl(1-\tfrac{\text{Cost}(I,\mathbf C_{\pazocal N}^*)}{\mathbb{E}_\omega[U]}\bigr)^2 \, \bigl(\mathbb{E}_\omega[U]\bigr)^2 + \mathrm{Var}_\omega(U)}
    \label{eq:PGZ_step}
\end{equation}
Noting that
\[
\bigl(1-\tfrac{\text{Cost}(I,\mathbf C_{\pazocal N}^*)}{\mathbb{E}_\omega[U]}\bigr)^2 \, \bigl(\mathbb{E}_\omega[U]\bigr)^2
=
\bigl(\mathbb{E}_\omega[U] - \text{Cost}(I,\mathbf C_{\pazocal N}^*)\bigr)^2,
\]
we can rewrite~\eqref{eq:PGZ_step} as
\begin{equation}
    \mathbb{P}\bigl(U_\omega > \text{Cost}(I,\mathbf C_{\pazocal N}^*)\bigr)
    \;\geq\;
    \frac{
        \bigl(\mathbb{E}_\omega[U] - \text{Cost}(I,\mathbf C_{\pazocal N}^*)\bigr)^2
    }{
        \bigl(\mathbb{E}_\omega[U] - \text{Cost}(I,\mathbf C_{\pazocal N}^*)\bigr)^2
        + \mathrm{Var}_\omega(U)
    }
    \label{eq:PGZ_step2}
\end{equation}
By dividing the numerator and denominator of~\eqref{eq:PGZ_step2} by 
$\bigl(\mathbb{E}_\omega[U] - \text{Cost}(I,\mathbf C_{\pazocal N}^*)\bigr)^{2}$, we obtain:
\begin{equation}
    \mathbb{P}\bigl(U_\omega > \text{Cost}(I,\mathbf C_{\pazocal N}^*)\bigr)
    \;\geq\;
    \frac{1}{
        1
        +
        \dfrac{\mathrm{Var}_\omega(U)}
              {\bigl(\mathbb{E}_\omega[U] - \text{Cost}(I,\mathbf C_{\pazocal N}^*)\bigr)^{2}}
    }
    \label{eq:PGZ_step3}
\end{equation}
This yields the claimed lower bound:
\begin{equation}
    \nu_{\pazocal{N}}^{\text{LB}}
    \;\defeq\;
    \frac{1}{
        1
        +
        \dfrac{\mathrm{Var}_{\omega}(U)}
              {\bigl(\mathbb{E}_{\omega}[U] - \text{Cost}(I,\mathbf C_{\pazocal N}^*)\bigr)^{2}}
    }
    \label{eq:nu_profitable}
\end{equation}
Hence, by Definition~\ref{definition:profitable_coinv}, the obtained lower bound $\nu_{\pazocal N}^{\text{LB}}$ \eqref{eq:nu_profitable} guarantees profitability of the co-investment, which completes the proof.
\end{proof}}

\added{We now present the following lemma (proved in Appendix~\ref{appendix:lemma_parametrize}), which characterizes the behavior of the lower bound
$\nu_{\pazocal N}^{\mathrm{LB}}$
when the underlying utility functions depend on an exogenous parameter.}

\added{\begin{lemma}
\label{lem:log_derivative}
If the random utility functions
$u_{i}^t(\mathbf h_i^t)$
are parametrized by an exogenous parameter $H$
(we denote such functions by
$u_{i}^t(\mathbf h_i^t)^{[H]}$,
and let the corresponding total utility be~$U^{[H]}$, to emphasize this parametrization),
and if we denote the total cost by $\text{Cost}(I,\mathbf C^{[H]})$, the corresponding expected coalition payoff by
\(
\bar v^{[H]}(\pazocal N)
\;\defeq\;
\mathbb{E}_{\omega}\!\left[U^{[H]}\right]
-
\text{Cost}(I,\mathbf C^{[H]}),
\)
and the associated variance by
\(
\mathrm{Var}_{\omega}\!\left(U^{[H]}\right),
\)
then, if the logarithmic derivative of the variance of the total utility
with respect to $H$ is smaller than or equal to the logarithmic
derivative of the square of the expected coalition payoff,  i.e.,
\begin{equation}
\frac{\mathrm d}{\mathrm d H}
\log \bigl(\mathrm{Var}_{\omega}(U^{[H]})\bigr)
\;\le\;
\frac{\mathrm d}{\mathrm d H}
\log \bigl(\bar v^{[H]}(\pazocal N)^2\bigr),
\label{eq:log_derivative_condition}
\end{equation}
the lower bound on the probability of profitable co-investment,
$\nu_{\pazocal N}^{\mathrm{LB}[H]}$ (defined as the lower bound in Theorem~\ref{thm:prof} under parameter value $H$), is non-decreasing with respect to $H$.
\end{lemma}}

\subsection{Realtionship between stability and profitability}\label{sec:relationship}
We now formalize the relationship between the stability of the grand coalition and the profitability of the co-investment. By definition, if the grand coalition is stable, no player is better off not being part of the grand coalition,  which implies that each player is receiving a non negative payoff. Therefore, stability guarantees profitability for all players. This is captured in the following proposition.

\begin{proposition}
Assume that the realized value of the grand coalition is strictly positive for all realizations, i.e., \( v_{\omega}(\pazocal{N}) > 0, \forall \omega \in \Omega \). Then, if the grand coalition~$\pazocal{N}$ is stable in the stochastic game~$(\pazocal{N}, \pazocal{V})$, the co-investment of the grand coalition is profitable for all players with probability~$1$.
\label{proposition:grand_stable}
\end{proposition}

We now show, in the following Proposition and Corollary, that profitability alone does not guarantee stability: even when all players benefit, the grand coalition may still fail to hold together.

\begin{proposition}
The probability that the co-investment is profitable for all players is larger than or equal to the probability that the grand coalition is stable, i.e.,
\begin{equation}
\mathbb{P}\left( x_{i,\omega} \geq 0,\ \forall i \in \pazocal N \right) \ge
    \mathbb{P}\left( \text{Grand coalition is stable} \right)
\end{equation}
\label{proposition:profitable_in_stable}
\end{proposition}

\begin{corollary}
The lower bound on the probability that the grand coalition is stable, denoted~$\nu^{\text{LB}}$, also serves as a lower bound on the probability that the co-investment is profitable for all players. That is,
\begin{equation}
    \mathbb{P}\left( x_{i,\omega} \geq 0,\ \forall i \in \pazocal N \right) \ge \nu^{\text{LB}}
\end{equation}
\label{corollary:v_lb}
\end{corollary}
The proofs of Proposition~\ref{proposition:grand_stable}, Proposition~\ref{proposition:profitable_in_stable}, and Corollary~\ref{corollary:v_lb} are provided in Appendix~\ref{appendix:proof_prop_poistivepayoff}, Appendix~\ref{appendix:prof_in_stable} and Appendix~\ref{appendix:proof_coro}, respectively.

\subsection{Sharing the payoff among players}
\label{sec:shapley}
After optimizing the capacity and the shares, and maximizing the total payoff of the game, the proposed co-investment scheme distributes the payoff among players.
Stochastic game $(\pazocal N, \pazocal V)$ presents uncertain payoffs~\eqref{eq: generic expression of payoff-realizations_s}. At the end of the investment period~$I$, all payoffs become known and they must be redistributed among players in an efficient way.

To this end, we adopt the nucleolus, a cooperative-game solution concept that provides a unique payoff allocation and accounts for the dissatisfaction of all coalitions~\citep{schmeidler1969nucleolus}. In particular, the nucleolus selects the allocation that lexicographically minimizes coalition excesses, thereby prioritizing the most dissatisfied coalitions. This makes it well suited to our setting, where the payoff must be allocated while preserving the stability of the grand coalition. Moreover, when the core is non-empty, the nucleolus belongs to the core.

For each realization~$\omega \in \Omega$, let $x_{\omega} = (x_{i,\omega})_{i \in \pazocal N}$ denote the nucleolus allocation of the realized game $(\pazocal N, v_{\omega})$. For any coalition $\pazocal S \subseteq \pazocal N$, we define the excess of coalition $\pazocal S$ under allocation $x_{\omega}$ as
\begin{equation}
    e_{\omega}(\pazocal S, x_{\omega})
    \defeq
    v_{\omega}(\pazocal S) - \sum_{i \in \pazocal S} x_{i,\omega},
\label{eq:excess}
\end{equation}
where $v_{\omega}(\cdot)$ is calculated as in~\eqref{eq: generic expression of payoff-realizations_s}.

The nucleolus allocation is defined as the efficient allocation of the grand-coalition payoff that lexicographically minimizes the vector of coalition excesses. Therefore, for every realization~$\omega \in \Omega$, the allocation $x_{\omega}$ satisfies
\begin{equation}
    \sum_{i \in \pazocal N} x_{i,\omega} = v_{\omega}(\pazocal N)
\label{eq:efficiency-nucleolus}
\end{equation}

Observe that $x_{i,\omega}$ is thus a random variable for every player $i \in \pazocal N$, since it depends on the realization~$\omega$ of the stochastic game. At the end of the game, the actual realization~$\omega \in \Omega$ that occurred becomes known, and revenues are redistributed so that the payoff of each player is~$x_{i,\omega}$. Recall that each player~$i$ contributes an initial payment~$p_i$ at the start of the game, such that the total cost is fully covered by the coalition~\eqref{eq:cost-sharing}.
While these payments may differ across players, the scheme ensures that all contributions are balanced through the final reward distribution. This reward, denoted as~$r_{i,\omega}$, represents the actual amount received by player~$i$ in realization~$\omega$ at the end of the game. It reflects both the payoff value and the return on investment:
\begin{equation}
    r_{i,\omega} \defeq x_{i,\omega} + p_i \quad \forall \omega \in \Omega, i \in \pazocal N
    \label{equation:reward}
\end{equation}
An important property of the reward distribution~\eqref{equation:reward} is that it ensures budget balance: the total amount redistributed to the players exactly matches the total collected utility. In other words, the mechanism guarantees that no more is distributed than what has actually been earned. This property is stated in the following Lemma.

\begin{lemma}
The sum of the individual rewards~$r_{i,\omega}$ distributed to players~$i \in \pazocal N$ equals the total utility collected across all time slots in realization~$\omega$:
\begin{equation}
    \sum_{i \in \pazocal N} r_{i,\omega}
    =
    \sum_{i \in \pazocal N} \sum_{t \in \pazocal T}
    u_{i,\omega}^t(
    \deleted{h}\added{\mathbf h}_{i, \pazocal N}^{t*}
    )
\end{equation}
\label{lemma:sum_reward}
\end{lemma}

The proof of this Lemma can be found in Appendix~\ref{appendix:budget-balance}.

\section{Application to \added{Mobile} Edge computing (\added{M}EC)} \label{Section: ana-results}
We now consider a particular investment setting, namely the one of~\added{M}EC, where the \textit{InP} is a NO, such as Comcast or Orange, and SPs are actors receiving requests from users, processing them and replying to them. Examples of SPs are Netflix, Youtube or new companies providing new generation services, such as augmented or virtual reality.

\subsection{Settings}\label{sec:settings}
We aim to capture a realistic and analytically tractable expression for the 
total cost~$\text{Cost}(I,\deleted{C}\added{\mathbf C})$ incurred when 
installing computational capacity~$\deleted{C}\added{\mathbf C}$ in one or multiple edge nodes and maintaining it over an investment 
period~$I$. This cost should reflect two main components: 
(i)~the upfront cost of installing resources, which increases with the 
\deleted{total installed capacity} 
\added{installed resources at each node and for each resource type}, and 
(ii)~the operational cost of maintaining these resources throughout the 
investment period, which increases with both the capacity matrix and the duration of the investment.

A natural formulation that meets these requirements is the following matrix-based and time-dependent expression:
\begin{equation}
Cost(I, \deleted{C}\added{\mathbf C})= \deleted{d \cdot C + d' \cdot I \cdot C}
\added{
\sum_{k=1}^K \sum_{m=1}^M 
\left(
    d_{k,m}
    + \int_{0}^{I} D_{k,m}(t) dt
\right) C_{k,m}
}
\label{equation: cost}
\end{equation}
where 
\deleted{$d$ and $d'$ denote constant installation and maintenance prices.}
\added{
$d_{k,m}$ is the installation cost per unit of capacity of type~$k$ at node~$m$, 
and $D_{k,m}(t)$ is a non-negative and increasing function representing the time-varying unit maintenance cost for resource $k$ at node $m$. The integral $\int_0^I D_{k,m}(t)\, dt$ represents the maintenance expenditure over the investment period.
}


To realistically model the revenue generated from services provided to end users at the edge, the utility function to be adopted must meet the following criteria:  
(i)~monotonically increasing in the allocated resources, as more resources lead to higher revenue,
(ii)~exhibit diminishing returns, as the marginal benefit of additional resources decreases beyond a certain point due to capacity saturation or limited user demand, 
(iii)~load-dependent, meaning that the generated utility scales with the amount of user traffic an SP receives,  
and (4)~bounded, capturing economic limitations (e.g., no infinite revenue). We do not explicit in this work the quality of service of admitted end users, we assume that it is fulfilled using the quantity of resources allocated to the SPs.
The following utility function satisfies all of the above properties \added{ and represents the revenue obtained by SP~$i$ from serving its load using the aggregate of allocated resources across all resource types and nodes}:
\begin{equation}
u_{i,\omega}^t( 
\deleted{h}\added{\mathbf h}_i^t)
=
\beta_i \cdot 
\added{f(}l_{i,\omega}^t\added{)}
\cdot 
\left(1 - 
\exp^{-
\deleted{\xi \cdot h_i^t}
\added{
\sum_{k=1}^K \sum_{m=1}^M \xi_{k,m} \, h_{i,k,m}^t
}}
\right),
\quad 
\forall \omega \in \Omega,\; t\in\pazocal T,\; i\in\pazocal N
\label{equation: utility function}
\end{equation}
where $l_{i,\omega}^t$ is the load of end-user requests received by SP~$i$ in realization~$\omega$, \added{$f(l_{i,\omega}^t)$ is a function of this load, and $\beta_i$, which we call the benefit factor, can be interpreted as the revenue SP~$i$ gets from serving one unit of load. The allocation $h_{i,k,m}^t$ represents the amount of resource of type~$k$ allocated to SP~$i$ at node~$m$ and time slot~$t$, while}\deleted{, parameter} $\xi\added{_{k,m}}$ \deleted{represents} \added{is} the \added{corresponding} diminishing return \deleted{parameter} \added{coefficient associated with resource type~$k$ at node~$m$.} \deleted{, and 
$\beta_i$, which we call the benefit factor, can be interpreted as the revenue SP~$i$ gets from serving one unit of load. }
\added{The function $f(\cdot)$ represents the impact of the load on the utility. It can take different forms depending on how we wish to model congestion. This function can take any shape (with no impact on the theoretical part, Section 3). 
For instance, a realistic way to describe the congestion effect is to adopt a concave function, e.g.,
\begin{equation}
    f(l^t_{i,\omega})=\frac{l^t_{i,\omega}}{1+\kappa \cdot l^t_{i,\omega}}
    \label{eq:concave_f}
\end{equation}
where $\kappa \ge 0$ represents the level of congestion. Larger values of $\kappa$ correspond to higher congestion, as the marginal utility of additional load decreases more rapidly. This shape of function captures situations where the available resources may not be sufficient to accommodate a massive volume of incoming requests with a satisfactory level of service. In such cases, additional load yields smaller marginal gains due to limited capacity.
A simple case is obtained when $\kappa = 0$, for which:
\begin{equation}
    f(l_{i,\omega}^t)=l_{i,\omega}^t
    \label{eq:l_linear}
\end{equation} 
which corresponds to a 
setting without congestion effects, where a higher load directly increases 
the utility. We use this shape of function throughout the evaluation section.}

The load~$l_{i,\omega}^t$ represents the amount of requests collected by SP~$i$ at time slot~$t$ in realization~$\omega$  and is given by:
\begin{equation}
    l_{i,\omega}^t \defeq \lambda_{i,\omega}^t \cdot \Delta
    \label{equation: load}
\end{equation}
where~$\lambda_{i,\omega}^t$ is the request rate and~$\Delta$ is the duration of time slot.

Rate~$\lambda_{i, \omega}^t$ (and thus load~$l_{i, \omega}^t$) may change from realization~$\omega \in \Omega$ to another. It is thus a random variable. Therefore, the stochasticity of the utilities in our numerical results derives from the stochasticity of the user request rate, as shown in~\eqref{equation: utility function} and~\eqref{equation: load}.\footnote{Note that, the theory developed in the previous section holds as well if the sources of stochasticity are others, e.g.,  imperfect knowledge of benefit factors}
The expected amount of requests received by SP~$i$ at time slot~$t$ is defined as follows:
\begin{equation}
        \bar{l}_i^t \defeq \mathbb E_\omega [l_{i,\omega}^t] \overset{\eqref{equation: load}}{=} \bar{\lambda}_{i}^t \cdot \Delta
    \label{equation: expectedload}
\end{equation}
where~$\bar{\lambda}_{i}^t$ denotes the expected request rate, defined as:
\begin{equation}
    \bar{\lambda}_{i}^t \defeq \mathbb{E}_\omega \left[ \lambda_{i,\omega}^t \right]
    \label{equation:lambda_bar}
\end{equation}
To model the load~$l_{i,\omega}^t$~\eqref{equation: load} of each SP, we aim for a formulation that satisfies the following criteria:  
(i) periodic, to reflect the regular temporal cycles observed in network traffic (e.g., day–night patterns in user activity), 
(ii) random, as service demand includes fluctuations caused by unpredictable user behavior, and
(iii) tunable, meaning the model allows control over the influence of randomness versus expected trends, enabling flexible simulation and analysis.

To showcase the generality of the proposed co-investment scheme, we consider two settings, with fundamentally different characterizations of uncertainty. In the first setting (section~\ref{section:first_modeling}), the stochastic process underlying~$l_{i, \omega}^t$ is unknown. What is known is just that it has finite and known upper and lower bounds and that values at different timeslots are independent. 
In the second setting (section~\ref{section: brownian}), we instead allow rates to be unbounded and \deleted{correlated} \added{dependent} from a timeslot to another. \deleted{Although this case does not comply with the hypothesis of Theorem~\ref{thm:main-theorem}, we will show numerically that the proposed co-investment scheme is still profitable for all players with high probability. }
In both cases, the uncertainty captured in~$l_{i,\omega}^t$ corresponds to randomness, in the sense defined by~\citep[Page 285]{gendreau2023network}, which reflects inherent, statistically characterizable although with possible unknown fluctuations in the process. 
\subsubsection{Traffic load with \added{independent} bounded fluctuations}\label{section:first_modeling}
In this setting, the load satisfies the three desired criteria mentioned above~(Section~\ref{sec:settings}) through a bounded stochastic model with values that are independent across time.
The request rate of an SP determines the number of user requests received per time unit. For example, with a typical 5G mean capacity of 1~Gbps~\citep{xu2020understanding} and an average request size of~$r_s = 100$~Kbits, the mean service rate $(s_r)$ corresponds to approximately 10,000 requests per second. In our setup, we set the mean request rate~$(\alpha_r)$ to 20\% of~$s_r$~\citep{peng2015random}. 

Empirical observations across 
communication networks reveal daily periodic patterns driven by temporal cycles~\citep{akgul2010periodicity} \citep{shi2018discovering}.
We define the expected request rate~$\bar{\lambda}_{i}^t$~\eqref{equation:lambda_bar} of SP~$i$ at time slot~$t$ using the model proposed by large network operators~\citep{vela2016traffic} to capture periodic fluctuations:
\begin{equation}
    \bar{\lambda}_{i}^t = a_{i,0} + \sum_{v=1}^V a_{i,v} \cdot \sin\left(2v\pi \cdot \frac{t - t_{i,v}}{T}\right)
    \label{equation:request_rate}
\end{equation}
where~$a_{i,v}$ and~$t_{i,v}$ represent the amplitude and phase shift of the~$v$th sinusoidal component, respectively. 
These coefficients are chosen such that~$\bar{\lambda}_i^t$ fluctuates around~$\alpha_r$, e.g., $\bar{\lambda}_i^t \in [0.1\alpha_r,\ 1.5\alpha_r]$.

We consider the actual request rate~$\lambda_{i,\omega}^t$ of SP~$i$ at time slot~$t$ in realization~$\omega$ to fluctuate around~$\bar{\lambda}_{i}^t$, i.e.,
\begin{equation}
    \lambda_{i,\omega}^t \in \left[(1 - \sigma)\bar{\lambda}_{i}^t,\; (1 + \sigma)\bar{\lambda}_{i}^t\right]
    \label{equation: bounded request rate}
\end{equation}
where~$\sigma \in [0, 1]$ quantifies the level of uncertainty. 
 
The amount of requests received by SP~$i$ at time slot~$t$ in realization~$\omega$, given by~\eqref{equation: load}, is directly tied to~$\lambda_{i,\omega}^t$ and inherits the same bounded fluctuation~\eqref{equation: bounded request rate}. It satisfies:
\begin{equation}
    l_{i,\omega}^t \in \left[(1 - \sigma)\bar{l}_i^t,\; (1 + \sigma)\bar{l}_i^t\right], \quad \forall \omega \in \Omega,\; t \in \pazocal T,\; i \in \pazocal N
    \label{eq:samebounded}
\end{equation}
where~$\bar{l}_i^t$ is the expected amount of requests~\eqref{equation: expectedload}.

\subsubsection{Traffic load with \deleted{correlated} \added{dependent} fluctuations}\label{section: brownian}
Another formulation of \( l_{i, \omega}^t \) is based on an alternative expression of \( \lambda_{i, \omega}^t \), where request rates follow an unbounded and temporally \deleted{correlated} \added{dependent} stochastic process. This model continues to satisfy the three desired criteria~(Section~\ref{sec:settings}). 
We define the request rate as follows:
\begin{equation}
\lambda_{i,\omega}^t = d_i^t \cdot \left[ (1 - \alpha) \cdot \mathbb{E}_\omega \left[ S_{i, \omega}^t \right] + \alpha \cdot S_{i, \omega}^t \right]
\label{equation: request_rate_general}
\end{equation}
where~$S_{i,\omega}^t$ is a non-negative stochastic process (to ensure that the request rate~$\lambda_{i,\omega}^t$ remains positive), and~$\alpha \in [0, 1]$ serves as a tunable parameter that quantifies the level of uncertainty introduced in the model. The term~$d_i^t$ captures the periodic trend, assumed to be known, and follows a sinusoidal form as:
\begin{equation}
d_i^t = a_{i,0} + \sum_{v=1}^V a_{i,v} \cdot \sin\left(2v\pi\frac{t - t_{i,v}}{T}\right)
\label{equation: sinusoidal}
\end{equation}
where~$a_{i,v}$ and~$t_{i,v}$ represent the amplitude and phase shift of the $v$th components in the sinusoidal function.

Internet traffic is known to exhibit two fundamental characteristics: self similarity and long range dependence~\citep{karagiannis2004long, lee2005stochastic} (``A large number of studies have shown that the network traffic possesses self-similarity and long-term dependence properties''~\citep{bi2021hybrid}).
Self-similarity refers to the persistence of traffic patterns across different time scales, while long-range dependence implies strong temporal \deleted{correlations} \added{dependence} that decays slowly over time. It is important to clarify that long range dependence does not mean a direct, memory like dependency on individual past events. Rather, it reflects persistent patterns and \deleted{correlations} \added{dependence} over long timescales. A common cause is the presence of bursty, high-variance applications such as video streaming, social media, or large file transfers.
For example, when a video goes viral, it can continue to generate significant traffic days or even months after its initial release, as users repeatedly share and engage with it over time.
Another source of long-range dependence is aggregated user behavior. When millions of users follow regular usage patterns such as daily, weekly, or seasonal routines, these behaviors create persistent statistical patterns in the traffic. These patterns appear as long term \deleted{correlations} \added{dependence} in measurements. To capture these elements, we model stochastic process~$S_{i, \omega}^t$ contained in~\eqref{equation: request_rate_general} as:
\begin{equation}
    S_{i, \omega}^t \defeq \max(0, f_{i, \omega}^t)
    \label{equation: stochastic_term_f}
\end{equation}
where~\( f_{i, \omega}^t \) is a fractional Brownian motion (fBm) with parameter~$H$~\citep{mandelbrot1968fractional}, \citep{mishura2008stochastic}. 
This choice is motivated by the fact that
fbm effectively captures the above characteristics of internet traffic, making it a natural tool for modeling it~\citep{dieker2004simulation}. 
The behavior of fBm is governed by the Hurst parameter~\(H \in (0,1)\), which controls the \deleted{correlation} \added{dependence} structure of the process. When \(H > 0.5\), the process exhibits long-range dependence and positive \deleted{correlation} \added{dependence} between increments. When \(H = 0.5\), fBm reduces to classical Brownian motion with independent increments. When \(H < 0.5\), the process exhibits anti-persistence, meaning that an increase in the process is more likely to be followed by a decrease, and vice versa, indicating a negative \deleted{correlation} \added{dependence} between increments.
The $\max(\cdot)$ in~\eqref{equation: stochastic_term_f} ensures that \(S_{i,\omega}^t\) is always non-negative and thus remains consistent with its interpretation as a request rate component, while preserving the essential temporal structure captured by fBm.

The following Lemma provides a closed-form expression for~$\bar{l}^t_i$~\eqref{equation: expectedload}, which is essential for deriving 
the amount of resources to be installed and their dynamic allocation among SPs~(Assumption~\ref{assumption: optimal}).

\begin{lemma}
The expected amount of requests~\eqref{equation: expectedload} received by SP~$i$ at time slot~$t$ is given by:

\begin{equation}
\bar{l}_{i}^t = d_i^t \cdot \frac{t^H}{\sqrt{2\pi}} \cdot \Delta
\label{equation: expected_load_explicit}
\end{equation}
\label{lemma: mean_s}
\end{lemma}

The proof of Lemma~\ref{lemma: mean_s} can be found in Appendix~\ref{appendix: mean}.

\subsection{Characterization of the resource sharing} \label{section: KKT}
In our model, we consider an \textit{InP} and several SPs.\deleted{ The following proposition allows computing resource sharing}
\added{We first derive a general characterization of the resource-sharing mechanism}, according to Assumption~\ref{assumption: optimal}, \added{which applies to resource types~$k$ and nodes~$m$. We then show that, under a simplified setting, resource sharing can be computed} in closed form (see Proposition~\ref{prop:convexity-of-problem}). \deleted{Its} \added{The} proof\added{s} (in Appendices~\ref{appendix:multi} and \ref{appendix: kkt}) rel\added{y}\deleted{ies} on the convexity of \added{the optimization problem}~\eqref{equation: optimal solution}\added{--}\eqref{eq:constraint-positive}.
The notation below is compliant with \tablename~\ref{tab:notations} and \tablename~\ref{tab:params}. 

\added{We first define the expected transformed load as
\begin{equation}
    \bar f_i^t \;\defeq\; \mathbb E_\omega \big[\, f(l_{i,\omega}^t) \,\big]
\end{equation}
and we then state the following two propositions:} 
\added{
\begin{proposition}
\label{prop:multi}
Consider the setting with $K$ resource types and $M$ nodes, where the
coalition installs a capacity matrix 
$\mathbf C \in \mathbb R_+^{K\times M}$ and where each SP receives an
allocation matrix $\mathbf h_i^t \in \mathbb R_+^{K\times M}$ at every time
slot $t\in\pazocal T$.
From Proposition~\ref{prop:convexity-of-problem},
problem~\eqref{equation: optimal solution}--\eqref{eq:constraint-positive} admits a unique optimal solution
$(\mathbf C_{\pazocal S}^*,\, \{\mathbf h_{i,\pazocal S}^{t*}\}_{t\in\pazocal T})$.
This solution is fully characterized by the Karush--Kuhn--Tucker (KKT)
conditions associated with the multidimensional Lagrangian.
\end{proposition}
For simplicity, we now focus on the special case of a single resource type and a single node (i.e., $K=1$ and $M=1$). In this case, problem~\eqref{equation: optimal solution}--\eqref{eq:constraint-positive} admits closed-form expressions for the installed capacity $C$, which becomes a scalar, and for the allocation, which becomes the set of vectors $\{\vec{h}\}_{t \in \pazocal T} = \{\{h_{i}^{t}\}_{i \in \pazocal{S}}\}_{t \in \pazocal T}$. For the sake of clarity, we set $d = d_{1,1}$, $D(t) = D_{1,1}(t)$, $\xi = \xi_{1,1}$, $C = C_{1,1}$, and $h_i^t = h_{i,1,1}^t$.}

\begin{proposition}
Given a coalition $\pazocal S$, let $\pazocal S' \subseteq \pazocal S \setminus \{\textit{InP}\}$ denote the subset of SPs with strictly positive load, i.e., those $i$ for which $\deleted{\bar{l}}\added{\bar{f}}_i^t > 0$ for all $t \in \pazocal T$. The resources installed and their allocation are as follows:
\begin{itemize}
    \item Case 1:  \\
    If $\pazocal S' \neq \emptyset$, the total installed capacity is given by:
    \begin{equation}
        C_{\pazocal S}^* = \frac{|\pazocal S'|}{\xi} \ln \left( \frac{\sum_{t \in \pazocal{T}} \left( \prod_{i \in \pazocal{S'} } \xi \cdot \beta_i \cdot \deleted{\bar{l}}\added{\bar{f}}_i^t \right)^{\frac{1}{|\pazocal S'|}}}{d + \deleted{d' \cdot I}\added{\int_0^I D(t)\, dt}} \right)
    \end{equation}
    and the share of the capacity~$C_{\pazocal S}^*$ allocated to each SP~$i \in \pazocal S'$  is:
    \begin{equation}
        h_{i, \pazocal S}^{t*} = \frac{C_{\pazocal S}^*}{|\pazocal S'|} + \frac{1}{\xi} \ln \left( \frac{\beta_i \cdot \deleted{\bar{l}}\added{\bar{f}}_i^t}{\left( \prod_{j \in \pazocal{S'}} \beta_j \cdot \deleted{\bar{l}}\added{\bar{f}}_j^t \right)^{\frac{1}{|\pazocal S'|}}} \right)
        \label{eqaution:share_h_active}
    \end{equation}
    \item Case 2:  \\
    If $\pazocal S' = \emptyset$ (i.e., all SPs have zero load at all $t \in \pazocal{T}$), then no capacity is installed and no resources are allocated:
    \begin{equation}
        C_{\pazocal{S}}^* = 0, \quad h_{i, \pazocal{S}}^{t*} = 0,\; \forall\, i \in \pazocal{S},\; t \in \pazocal{T}
        \label{equation:share_zero}
    \end{equation}
\end{itemize}
\label{Proposition: kkt}
\end{proposition}

\added{From Proposition~\ref{Proposition: kkt}, the optimal $C_{\pazocal N}^*$ and its shares depend only on the total cost term $d \cdot C + C \int_0^I D(t)\, dt$, i.e., on the scalar \(
d + \int_0^I D(t)\,dt
\), regardless of the specific shape of $D(t)$. In other words, decisions are influenced solely by the total cost over the entire investment period, whatever $D(t)$. Consequently, recomputing all results of Section~\ref{Section: ana-results} with different shapes of $D(t)$ always leads to identical outcomes. In particular, any two functions $D(t)$ that yield the same integral $\int_0^I D(t)\,dt$ lead to identical investment and sharing decisions. For this reason, we assume a constant unit maintenance cost $D(t)=d'$ in the numerical results.}

We now show the following proposition \added{(proved in Appendix~\ref{appendix:h_increase_with_b_l})}:
\begin{proposition}\label{Proposition:h_increase}
From Proposition~\ref{Proposition: kkt}, the resources~\( h_{i, \pazocal{S}}^{t*}, \forall t \in \pazocal T\) allocated to SP~$i$ are strictly increasing with respect to its expected load \( \deleted{\bar{l}}\added{\bar{f}}_i^t \) and its revenue per unit of load~\( \beta_i \), for all cases where $\deleted{\bar{l}}\added{\bar{f}}_i^t > 0$ and $\beta_i > 0$.
\end{proposition}
\deleted{The proof is provided in Appendix~\ref{appendix:h_increase_with_b_l}.}
\added{Note that Proposition \ref{Proposition:h_increase} states that SPs with higher traffic or more profitable service $\beta_i$ receive more resources.}

\added{For the sake of clarity and without loss of generality, our numerical results focus on one type of resource, namely virtual CPU (K = 1) and a single edge node (M = 1), so that the installed capacity $C_{\pazocal S}^*$ and the allocations $h_{i,\pazocal S}^{t*}$ are scalars obtained from Proposition~\ref{Proposition: kkt}. This simplification allows us to isolate and highlight the co-investment behavior without the added complexity of multi-node or multi-resource management.}


We now evaluate our model. We test a thorough set of scenarios with different levels of uncertainty, different numbers of SPs, and different profiles of SPs. Our settings are thus hypothetical but comprehensive. Moreover, great care is given in setting parameters according to real-world scenarios as detailed in \tablename~\ref{tab:params}.

\begin{table*}
    \centering
    \caption{Settings}
    \label{tab:params}
    \resizebox{\textwidth}{!}{ 
    \begin{tabular}{p{0.44\textwidth} p{0.44\textwidth} p{0.44\textwidth}}
    \toprule
    Functions & Symbols & Values \\
    \midrule
    Investment period  & $I=|\pazocal{T}| \cdot \Delta$ & 5 years \\
    Price per vcore
    \eqref{equation: cost}
    & $d$ & 10.94 $\$/\text{vcore}$~\citep{azureStackEdgePricing} \\
    Price per hour for one virtual CPU (vcore)
    \eqref{equation: cost}
    & $d'$ & $0.0225 \, \$/\text{hour}/\text{vcore}$~\citep{azureStackEdgePricing} \\
    Benefit factor
    \eqref{equation: utility function} 
    & $\beta_i=\beta$ & $6 \times 10^{-6} \, \$/\text{req}$ \citep{AWSLambdaPricing} \\
    Diminishing return 
    \eqref{equation: utility function} 
    & $\xi$ & 0.03 \\
    Time slot length 
    \eqref{equation: load}
    & $\Delta$ & 1 hour \\
    Number of time slots per day (periodicity of load)
    \eqref{equation:request_rate}
    & $T$ & 24 time slots \\
    Parameter controlling the level of uncertainty in the \deleted{uncorrelated} \added{independent} load~\eqref{equation: bounded request rate} 
    & $\sigma$ & $[0, 1]$ \\
    Parameter controlling the level of uncertainty in the \deleted{correlated} \added{dependent} load~\eqref{equation: request_rate_general} 
    & $\alpha$ & $[0, 1]$ \\
    Hurst parameter \eqref{equation: expected_load_explicit} & $H$ & 0.7 \\
    \bottomrule
    \end{tabular}
    }
\end{table*}

\subsection{Evaluation of the traffic load with \added{independent} bounded fluctuations}

We start by evaluating the first traffic model introduced in Section~\ref{section:first_modeling}.
\subsubsection{Illustration of traffic load realizations}
\figurename~\ref{fig:stochastic_icc} illustrates the traffic load corresponding to one SP over a five-day period, where the x-axis represents time in days and the y-axis denotes the load in requests per second. The red solid line represents the expected load, which follows a clear periodic pattern, reflecting cyclic daily variations in traffic demand. The dashed lines correspond to multiple stochastic load realizations, which exhibit significant variability around the expected load. In some instances, the stochastic realizations reach the double of the expected load, while in others, they drop close to zero, highlighting the unpredictable nature of traffic fluctuations. Additionally, the stochastic realizations are uncorrelated, meaning that each trajectory evolves independently without following a common trend beyond the underlying periodicity. This variability suggests the presence of a compensation effect, where high and low deviations cancel each other out over time, leading to an overall balance despite short-term fluctuations. 
\begin{figure}
    \centering
    \includegraphics[width=0.8\linewidth]{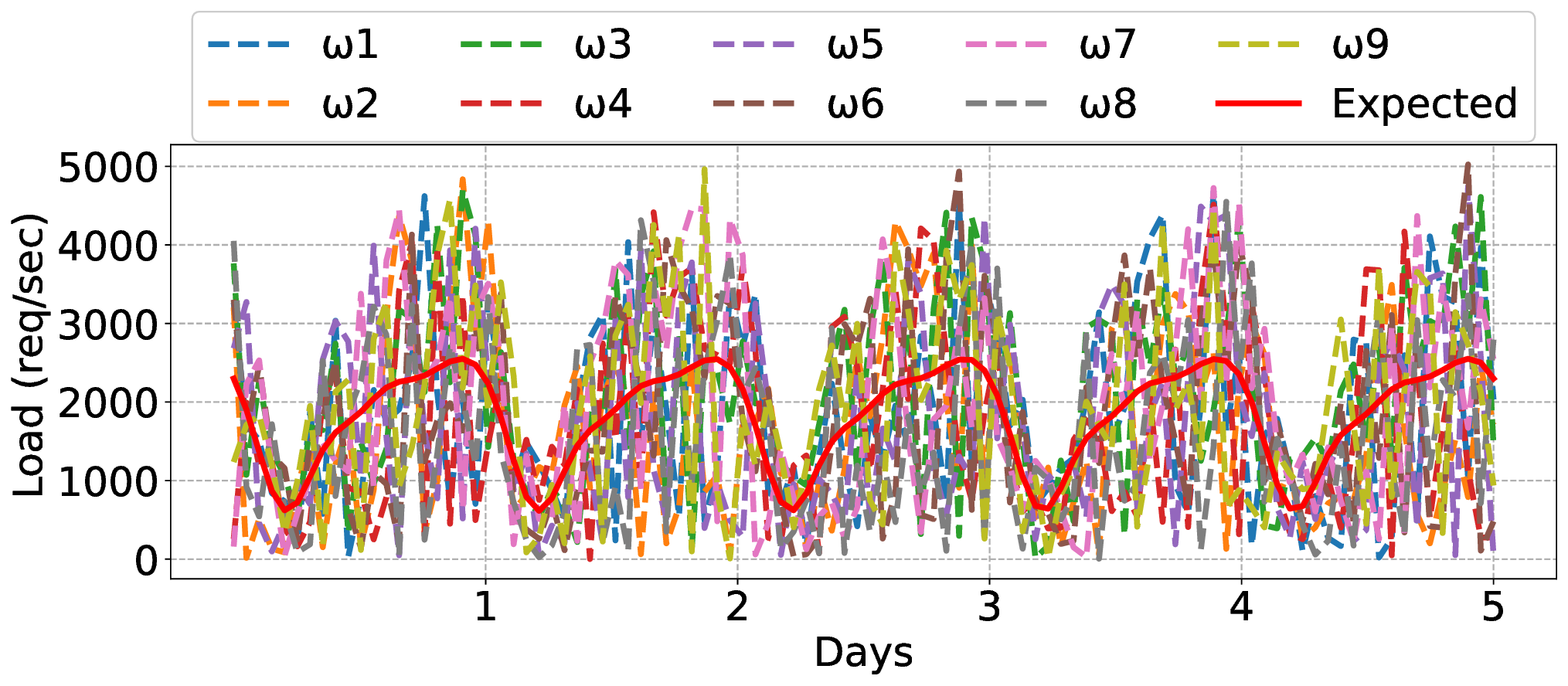}
    \caption{Traffic load \added{of SP1} with bounded fluctuations over five days within the investment period}
    \label{fig:stochastic_icc}
\end{figure}
\subsubsection{Dynamic resource sharing} 
We first consider a coalition~$\pazocal N$ including one \textit{InP} and two SPs. \figurename~\ref{fig:load_and_share} illustrates the expected traffic load~(\ref{fig:expecte-load}) and the corresponding dynamic resource allocation~(\ref{fig:share}) over a one-day period within the investment period. 
As shown in \figurename~\ref{fig:expecte-load}, the traffic load patterns of SP1 and SP2 reflect the distinct characteristics of the areas they serve. 
SP1 exhibits a low traffic load during the early morning hours, followed by a steady increase throughout the day, peaking in the evening. This pattern is typical of residential areas, where usage ramps up as people return home and engage in online activities like streaming and gaming. 
In contrast, SP2 shows lower traffic in the early morning, followed by a rapid rise during the daytime, peaking in the afternoon and declining after business hours. This trend resembles business or commercial usage patterns, where activity is driven by daytime operations.
To accommodate these variations, the resource allocation strategy dynamically adjusts to them, as shown in \figurename~\ref{fig:share}. During high-demand periods for SP1, a greater share of resources is allocated to it, whereas SP2 receives a larger share when SP1's demand decreases. This dynamic approach ensures optimal utilization of available resources, improving efficiency and performance for both SPs.
\begin{figure}[h!]
    \centering
    \begin{subfigure}{0.49\linewidth}
        \centering
        \includegraphics[width=\linewidth]{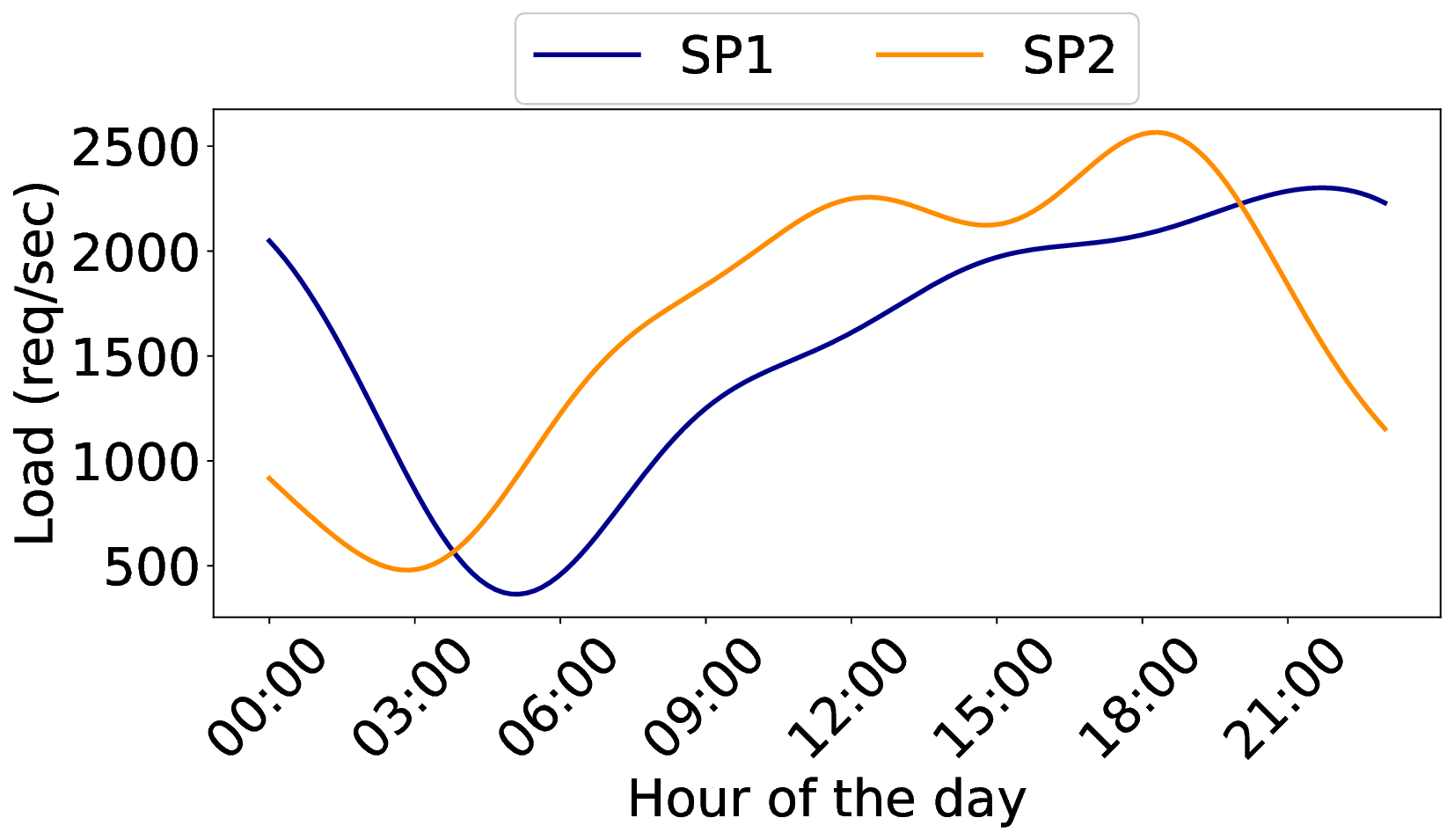}
        \caption{Expected \added{traffic} load \added{of SP1 and SP2} over one day \added{, highlighting residential and business usage patterns}}
        \label{fig:expecte-load}
    \end{subfigure}
    \hfill
    \begin{subfigure}{0.49\linewidth}
        \centering
        \includegraphics[width=\linewidth]{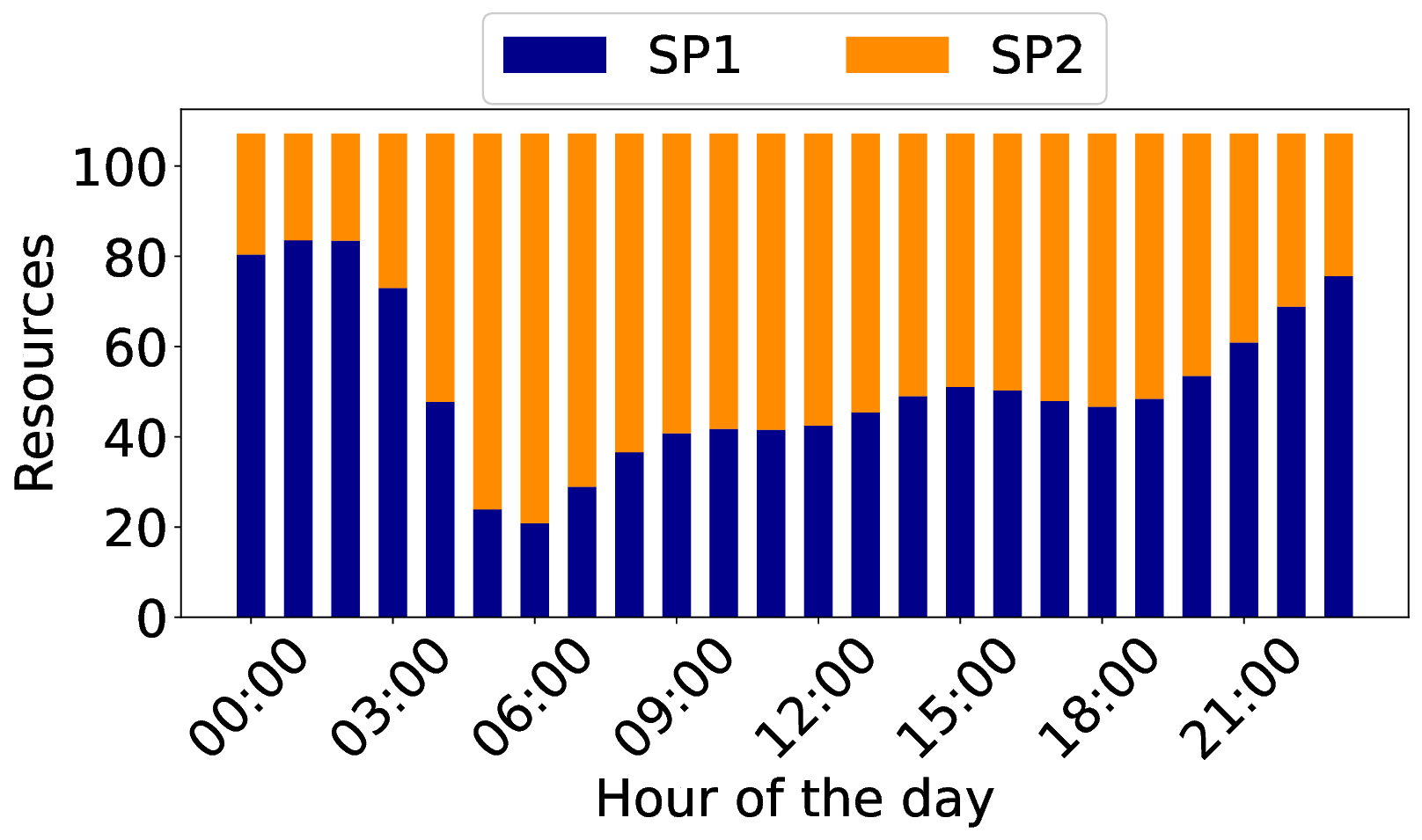}
        \caption{\added{Dynamic} \deleted{S}\added{s}hare of \added{infrastructure} resources \added{ over one day, showing how resource shares adapt to demand}}
        \label{fig:share}
    \end{subfigure}
    \caption{Expected traffic load and the corresponding resource allocation \added{in a coalition with one \textit{InP} and two SPs} for a typical day during the investment period}
    \label{fig:load_and_share}
\end{figure}

\subsubsection{Impact of coalition size and player heterogeneity on stability}
We analyze different coalitions~$\pazocal N$, each including one \textit{InP} along with 2, 3, 4, or 5 SPs. Among the SPs, SP1, SP2, and SP3 have comparable traffic loads, SP4 has high traffic load, and SP5 has low traffic load.  
\figurename~\ref{fig:lower_bound} presents the lower bound $\nu^{\text{LB}}$ on the probability that the grand coalition is stable (Theorem~\ref{thm:main-theorem}), across different levels of uncertainty~$\sigma$.
In the extreme case, $\sigma=1$ means that the uncertainty in the load is equal to 100\% of its expected value, indicating the maximum range of unknown requests, i.e., $[0, 2 \bar{l}_i^t]$~\eqref{eq:samebounded}. 
We observe that coalitions composed of SPs with comparable traffic loads, such as SP1, SP2, SP3, and even SP4, achieve a lower bound equal to 1, indicating stable cooperation across all uncertainty levels~$\sigma$. Expanding the coalition always increases the total coalition payoff. When the added SPs have comparable or higher traffic levels, this increase is sufficient to ensure that all players remain in the grand coalition, resulting in a lower bound of 1. However, when a low traffic load SP, like SP5, is included, as in coalition $\{1,2,3,5\}$, the lower bound on stability decreases significantly, and drops to zero beyond moderate values of~$\sigma$ (e.g., when $\sigma=0.5$). This suggests that the coalition becomes less attractive for at least some players. Although some stability may still be achievable under low uncertainty, the disparity in traffic patterns reduces the incentive of players to form and stay in the grand coalition and thus co-invest.
This effect is most pronounced in the full coalition~$\{1,2,3,4,5\}$, where the lower bound remains zero across all uncertainty levels~$\sigma$, except when $\sigma = 0.001$. Despite the potential for a higher coalition payoff in larger coalitions, the inclusion of highly heterogeneous SPs, especially those with low traffic load, leads to instability. When individual profiles are not aligned due to unequal contribution or risk exposure, players are unlikely to form the coalition.
\begin{figure}[h!]
    \centering
    \includegraphics[width=0.6\linewidth]{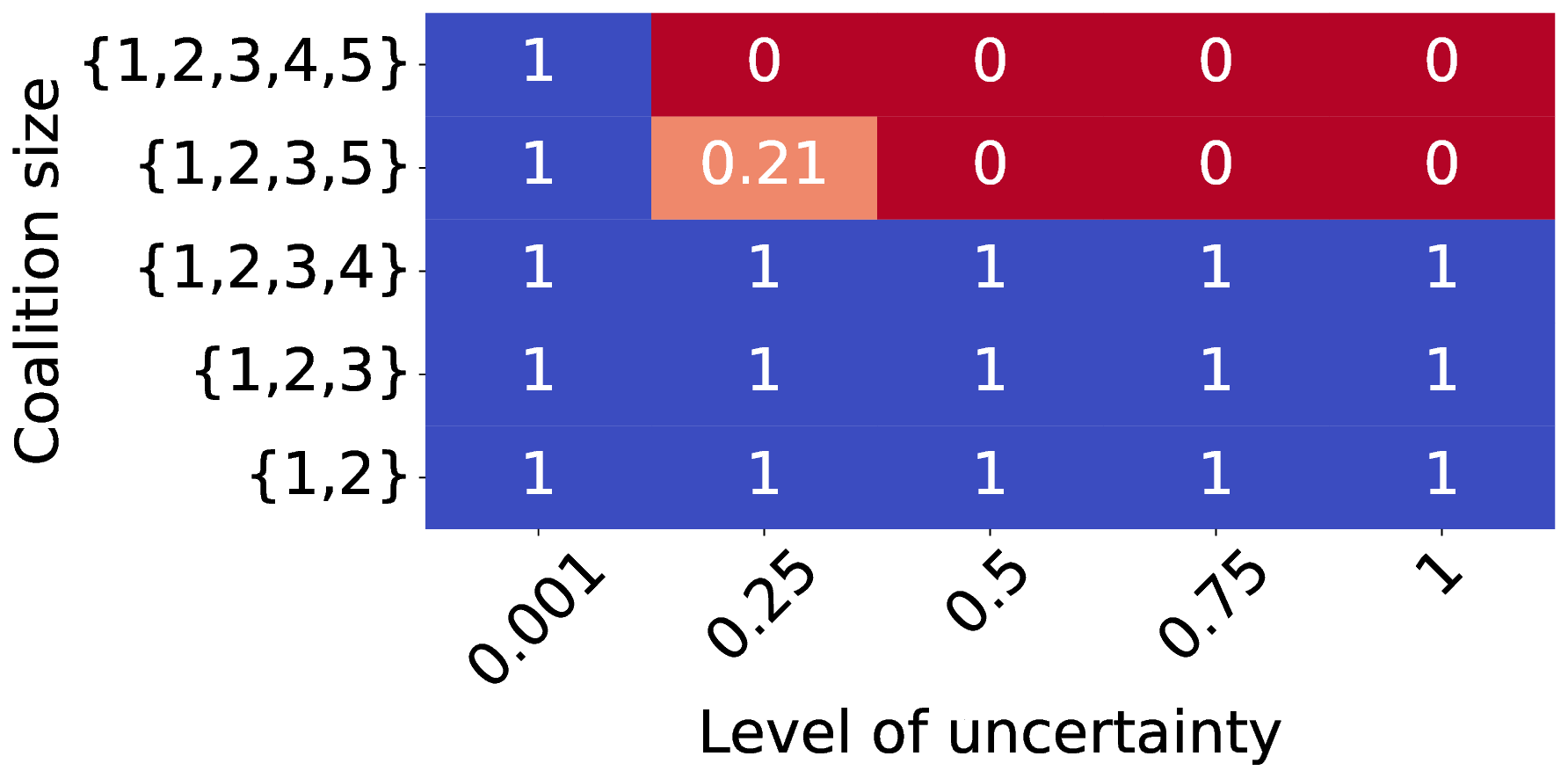}
    \caption{Lower bound on the probability that the grand coalition is stable under
different levels of uncertainty~$\sigma$ and varying coalition sizes}
    \label{fig:lower_bound}
\end{figure}

We now consider the coalition~$\pazocal N$ formed by one \textit{InP} and the set~$\{1,2,3,4,5\}$ of SPs. \tablename~\ref{tab:com_lower} shows the individual lower-bound probabilities $\mathbb{P}_i^{\text{LB}}$, $\forall i \in \pazocal{N}$~\eqref{equation: lb}, which represent the probability that each player has an incentive to form and stay in the grand coalition~$\pazocal N$ across different values of~$\sigma$. It also shows~$\nu^{\text{LB}}$~(Theorem~\ref{thm:main-theorem}), the lower-bound probability for the stability of the grand coalition~$\pazocal{N}$. We observe from the table that SP5, the smallest player with the lowest traffic load, is highly motivated to join the coalition, as it benefits from partnering with bigger players regardless of uncertainty. In contrast, SP4, the biggest player with the highest traffic load, is more hesitant and has an incentive to remain in the grand coalition only when uncertainty is negligible~$(\sigma=0.001)$, highlighting its risk-averse nature. SP1, SP2, and SP3, which have comparable medium traffic loads, exhibit similar behaviors: their lower-bound probability of stability drops rapidly, reaching zero shortly after $\sigma = 0.25$. 

\added{From~\eqref{equation: lower bound}, the stability probability lower bound increases when a player has a smaller range of utility fluctuations, as defined in~\eqref{eq:delta-u}. Since SP5 is a low-load SP, its revenue is relatively small, and therefore its utility $\text{range}(u_i^t)$ remains small (relative to the ranges of the others) even when the uncertainty level $\sigma$ is high. A smaller range $\text{range}(u_i^t)$  reduces the denominator inside the exponential term of~\eqref{equation: lower bound}, which makes the expression $1 - 2\exp(\cdot)$ closer to one. As a result, SP5 achieves a higher individual stability probability than high-load SPs, whose larger revenue variations lead to larger ranges and hence smaller probability values.}

Note that $\mathbb{P}_{\text{InP}}^{\text{LB}} = 1$ for all values of~$\sigma$. This indicates that the \textit{InP} always benefits from the co-investment, since the coalition cannot form without it (as established in Proposition~\ref{proposition:pointless}).

\begin{table*}[h!]
    \centering
    \caption{Lower-bound probabilities for individual players and the grand coalition~$\pazocal{N}$ across different levels of uncertainty~$\sigma$}
    \label{tab:com_lower}
    \resizebox{0.6\textwidth}{!}{
    \begin{tabular}{ccccccc}
        \toprule
        $\sigma$ & \multicolumn{6}{c}{Lower bound probability of stability} \\
        \cmidrule(lr){2-7}
        & \(\mathbb{P}^{\text{LB}}_{\text{SP1}}\) & \(\mathbb{P}^{\text{LB}}_{\text{SP2}}\) & \(\mathbb{P}^{\text{LB}}_{\text{SP3}}\) & \(\mathbb{P}^{\text{LB}}_{\text{SP4}}\) & \(\mathbb{P}^{\text{LB}}_{\text{SP5}}\) & \( \nu^{\text{LB}} \) \\
        \midrule
        0.001 & 1    & 1    & 1    & 1    & 1    & 1 \\
        0.25  & 0.31 & 0.20 & 0.30 & 0    & 1    & 0 \\
        0.5   & 0    & 0    & 0    & 0    & 0.99 & 0 \\
        \bottomrule
    \end{tabular}}
\end{table*}

\added{\subsubsection{Profitability of the co-investment under independent and bounded fluctuations}
As shown above (see Fig.~\ref{fig:lower_bound}, Table~\ref{tab:com_lower}), both the individual probabilities and the lower bound on the probability of stability of the grand coalition decrease as uncertainty increases when the SPs forming the coalition are not of similar size (e.g., SP4 and SP5). However, co-investment opportunities still exist in this case, even when the stability lower bound $\nu^{\text{LB}}$ vanishes, because the lower bound on the probability of profitable co-investment, denoted $\nu_{\pazocal N}^{\text{LB}}$, remains high and is found to be close to~$1$. Table~\ref{tab:lower_icc_profit} reports this lower bound together with the empirical probabilities, computed over $100$ realizations, that each player (\textit{InP} and SP1--SP5) achieves a non-negative payoff. These empirical probabilities are equal to~$1$ for all players, even at the maximum level of uncertainty $\sigma = 1$. This empirical evidence confirms the tightness of the lower bound derived in Theorem~\ref{thm:prof}.}
\begin{table*}[t]
\centering
\caption{\added{Lower-bound and empirical probabilities of profitable co-investment for individual players and for the coalition~$\pazocal N$ at the maximum level of uncertainty, $\sigma = 1$}}
\label{tab:lower_icc_profit}
\resizebox{0.82\textwidth}{!}{
\begin{tabular}{c|c|ccccccc}
\toprule
$\sigma$ 
& \multicolumn{1}{c|}{Lower bound} 
& \multicolumn{7}{c}{Empirical probability} \\
\cmidrule(lr){2-2}
\cmidrule(lr){3-9}
& $\nu_{\pazocal N}^{\text{LB}}$ 
& $\tilde{\mathbb{P}}_{\textit{InP}}$ 
& $\tilde{\mathbb{P}}_{\text{SP1}}$ 
& $\tilde{\mathbb{P}}_{\text{SP2}}$ 
& $\tilde{\mathbb{P}}_{\text{SP3}}$ 
& $\tilde{\mathbb{P}}_{\text{SP4}}$ 
& $\tilde{\mathbb{P}}_{\text{SP5}}$ 
& $\tilde{\mathbb{P}}_{\pazocal N}$ \\
\midrule
1 
& 1 
& 1 
& 1 
& 1 
& 1 
& 1 
& 1 
& 1 \\
\bottomrule
\end{tabular}}
\end{table*}

\subsection{Evaluation of the traffic load with \deleted{correlated} \added{dependent} fluctuations}

We now evaluate the second traffic model introduced in Section~\ref{section: brownian}.

\subsubsection{Illustration of stochastic realizations}
\figurename~\ref{fig:fbm_load} illustrates the stochastic process~$S_{i, \omega}^t$~\eqref{equation: stochastic_term_f} corresponding to one SP in an \added{M}EC environment over a five-day period. The expected load (red line) follows a gradual increase, reflecting the assumption that \added{M}EC is a newly introduced technology. Initially, there is little to no traffic, but as the system is adopted, more applications offload workloads to edge nodes, leading to a steady rise in demand. However, the stochastic realizations ($\omega1$ to $\omega9$) exhibit significant divergence from the expected trend, demonstrating highly unpredictable traffic behavior. Some realizations remain close to zero throughout the period, indicating scenarios where edge adoption is minimal or delayed. Others experience sharp increases, suggesting rapid deployment and high usage of edge resources, while some trajectories exhibit fluctuations, either rising and falling repeatedly or showing a delayed start followed by a sudden increase or decline.
\begin{figure}[h!]
    \centering
    \includegraphics[width=0.8\linewidth]{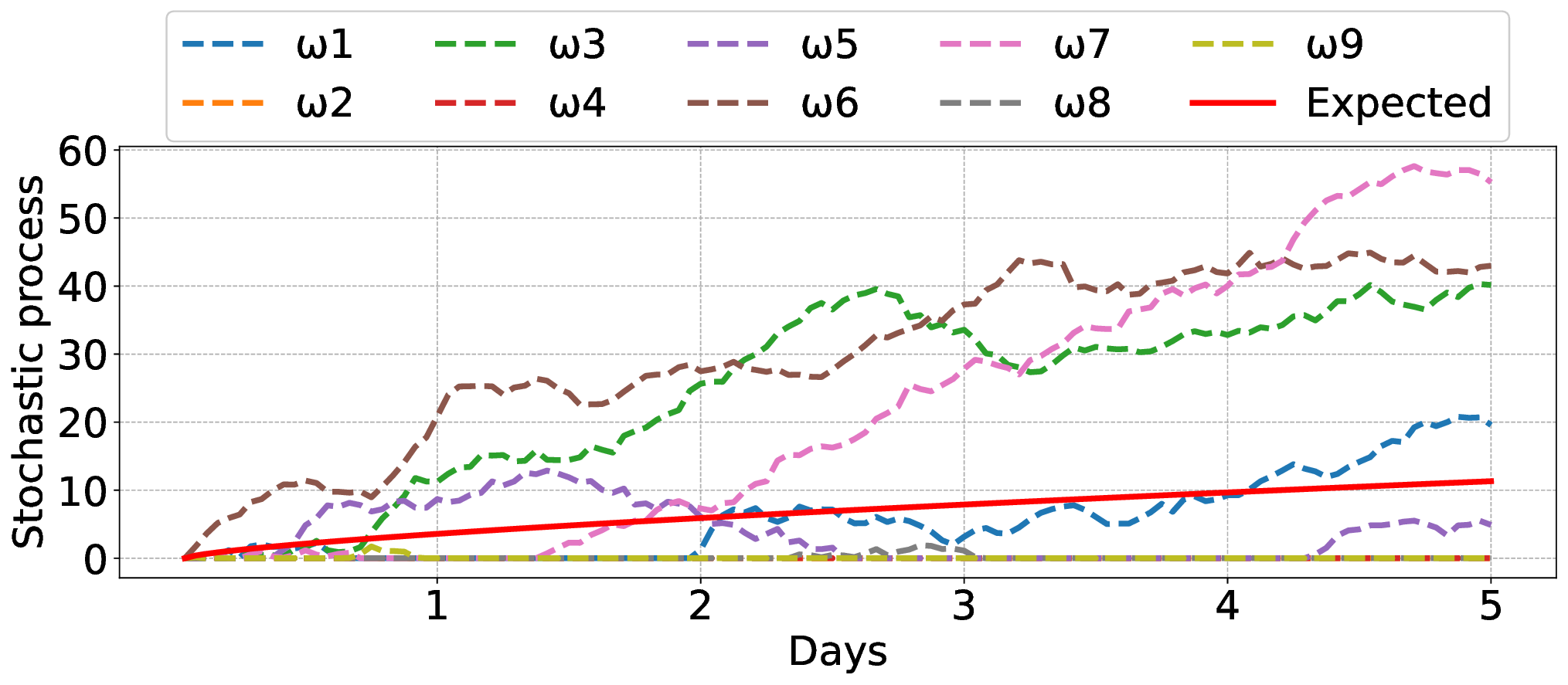}
    \caption{Stochastic realizations used to model traffic load variations \added{of SP1} over five days within the investment period}
    \label{fig:fbm_load}
\end{figure}

\subsubsection{Outcome of the stochastic game}
We provide in \figurename~\ref{fig:all_distributions_new} a detailed analysis of revenue, payment and payoff redistributions using the nucleolus (Section~\ref{sec:shapley}) across different players (\textit{InP}, SP1, SP2) under varying levels of uncertainty represented by~$\alpha$. In the case of~$\alpha = 1$, the traffic load depends entirely on the stochastic realization, with no contribution from its expected value~(according to~\eqref{equation: request_rate_general}). The figure shows outcome redistributions across multiple realizations~$\omega$.
To better understand the impact of uncertainty, we present in \figurename~\ref{fig:revenue}, \figurename \ref{fig:payment} and \figurename \ref{fig:payoff} box plots of these redistributions. The length of the box represents the variability of the distribution.

Starting with revenue redistribution (\figurename~\ref{fig:revenue})~\eqref{equation:reward}, we observe that at low values of~$\alpha$, the \deleted{violins} \added{box plots} are short, with a compact interquartile range (IQR), represented by this box. The IQR represents the middle 50\% of the data, ranging from the first quartile (25th percentile) to the third quartile (75th percentile), while the white line inside the box indicates the median revenue.
As $\alpha$ increases, the IQR increases, indicating that revenue becomes dispersed and unpredictable.
Additionally, the elongation of the upper tail further indicates that while higher uncertainty may provide opportunities for greater revenues, the lower tail suggests the possibility of lower revenue, raising the risk of extreme outcomes. This effect is particularly evident for SP1 and SP2, as they are affected by fluctuating demand and service usage. The \textit{InP}'s revenues remain stable and do not experience significant drops. It is important to note that this revenue is not intrinsic, it is a revenue that arises from the redistribution of the total revenue generated by the SPs across all players, including the \textit{InP} itself. Interestingly, negative revenues also appear under high uncertainty, meaning that for some realizations, players consume more resources than the value they generate.

A similar pattern is observed in the payment redistribution (\figurename~\ref{fig:payment}), where at low $\alpha$ values (e.g., $\alpha = 0.25$), payments remain tightly clustered, as shown by the small IQR. However, as uncertainty grows, the IQR expands,  meaning payments are no longer concentrated around a single expected value but are spread across a broader range. The increased length of the box plots at high $\alpha$ values (e.g., $\alpha = 1$) indicates that payments fluctuate significantly under uncertain conditions, leading to a higher financial risk and variability for the players. This suggests that in some cases, players may pay significantly more or less than expected, reinforcing their exposure to financial uncertainty. 
Notably, we observe negative payments, particularly for the \textit{InP} and occasionally for SP1 and SP2. Payments are calculated based on payoffs and revenues redistribution, not just raw cost. When a player has a high payoff because they are essential to the coalition or contribute significantly, they may receive money from others to reflect their role. For the \textit{InP}, this often occurs because it acts as a veto player~(Proposition~\ref{proposition:pointless}). 
For SPs, it may happen because they generate the revenue through user services, making them key contributors to the coalition’s payoff.

These dynamics come together in the payoff redistribution. We show in \figurename~\ref{fig:payoff} that at low $\alpha$, all players enjoy stable and positive payoffs, meaning they recover more than they pay. But under higher uncertainty (e.g., $\alpha = 1$), the spread of the payoff redistributions increases, and the payoffs of SP1 and SP2 become negative in some realizations. In these cases, not only does variability increase, but some players experience negative payoffs, meaning they end up paying more than they earn overall. This highlights a key risk of operating under high uncertainty: while some scenarios may lead to high profits, others can result in financial losses. 

\begin{figure}[h!]
    \centering
    \begin{subfigure}{0.49\textwidth}
        \centering
        \includegraphics[width=\linewidth]{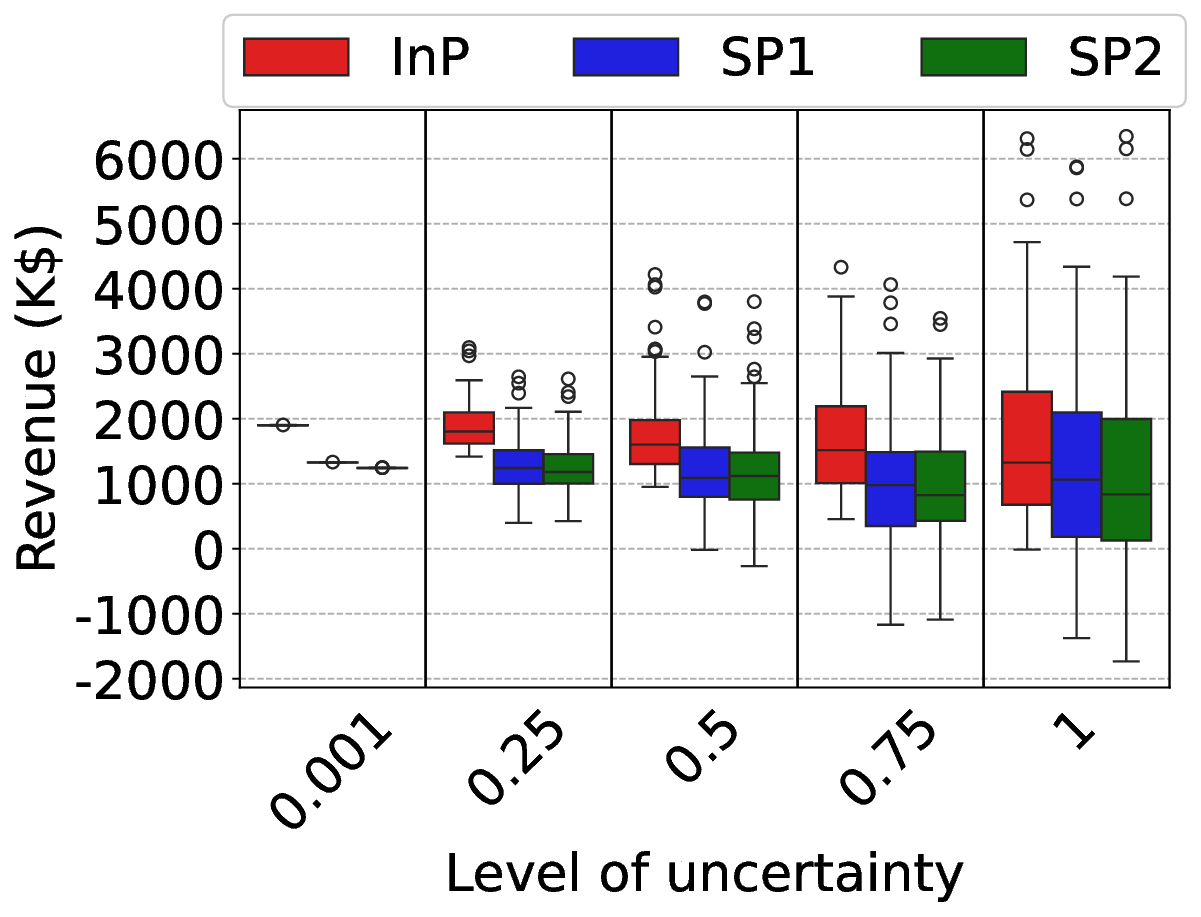}
        \caption{Revenue redistribution}
                \label{fig:revenue}
    \end{subfigure}
    \hfill
    \begin{subfigure}{0.49\textwidth}
        \centering
        \includegraphics[width=\linewidth]{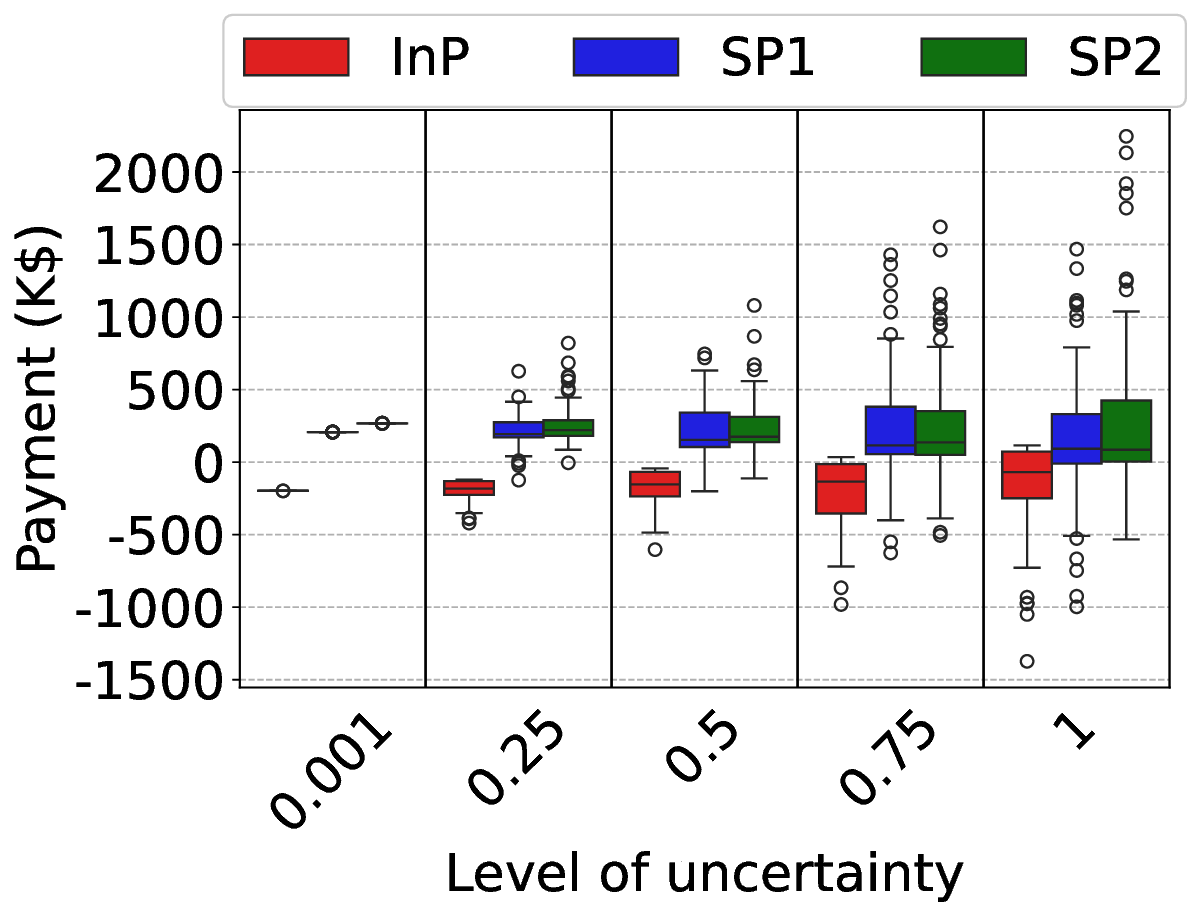}
        \caption{Payment redistribution}
        \label{fig:payment}
    \end{subfigure}
    \hfill
    \begin{subfigure}{0.49\textwidth}
        \centering
        \includegraphics[width=\linewidth]{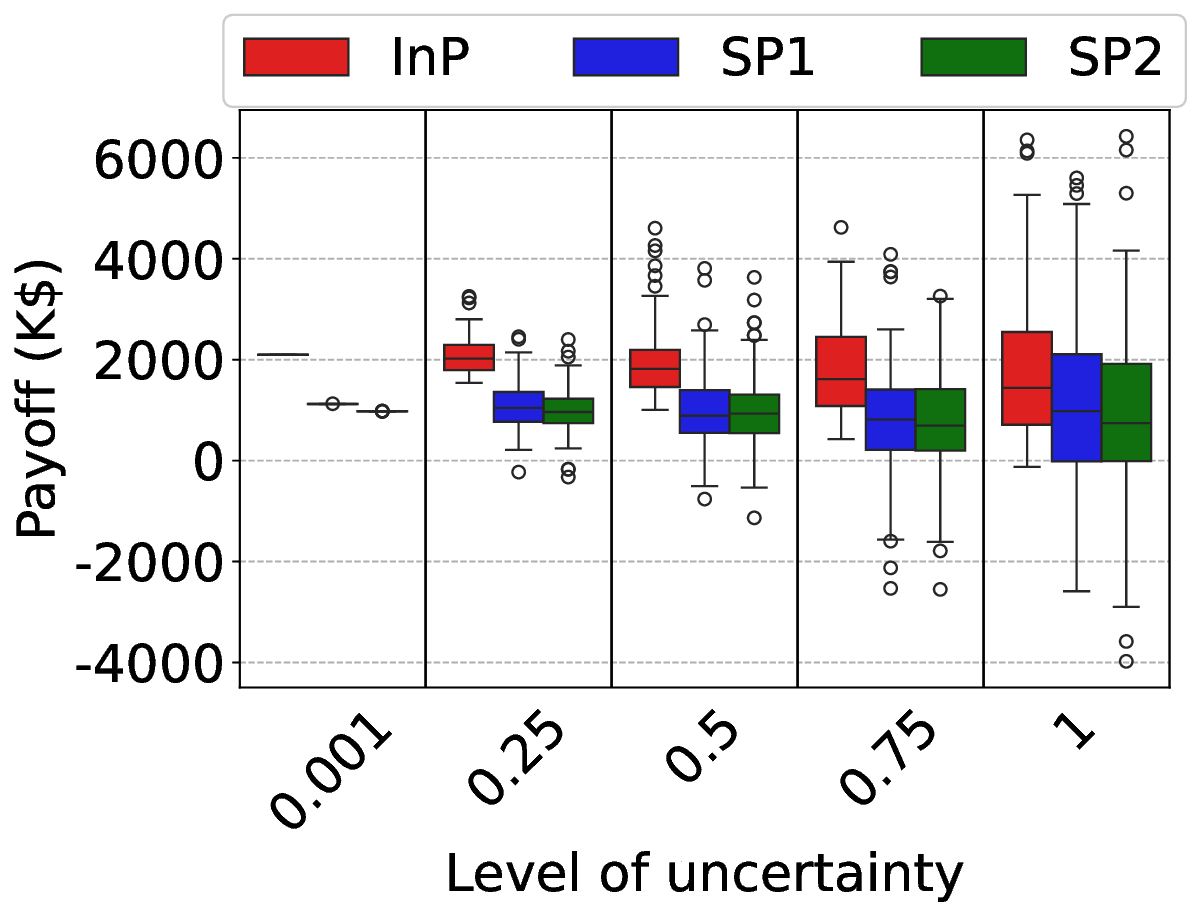}
        \caption{Payoff redistribution}
                \label{fig:payoff}
    \end{subfigure}
\caption{
\added{Box}\deleted{Violin} plots of revenue, payment, and payoff redistributions for \textit{InP}, SP1, and SP2 under varying levels of uncertainty~$\alpha$, using 100 generated realizations. 
\added{Each box represents the interquartile range, the median is shown by the horizontal line, the whiskers indicate the variability across realizations, and individual points beyond the whiskers correspond to values that fall outside the main dispersion range}
}
    \label{fig:all_distributions_new}
\end{figure}
\added{\subsubsection{Stability of the grand coalition under strong dependence}
We initially attempted to derive analytically a lower bound on the probability that the grand coalition is stable using Chebyshev’s inequality (Theorem~\ref{thm:main-theorem}~(1)). However, the resulting bound was overly conservative and therefore not informative. This behavior can be explained by the presence of highly dependent traffic modeled via fractional Brownian motion (Section~\ref{section: brownian}), which induces large variance and, in turn, significantly lowers the resulting probability bound. This motivates the use of the lower bound on co-investment profitability given in Theorem~\ref{thm:prof}, which applies to all traffic fluctuation models and all degrees of dependence, including strong long-range dependence, as discussed in the following section.}

\subsubsection{\deleted{p}\added{P}rofitability of the co-investment under uncertainty}
\deleted{We initially attempted analytically to estimate a lower bound on the probability that the grand coalition is stable using Chebyshev's inequality, but the result was overly conservative and therefore not useful. Instead, we estimate the probability of profitable co-investment empirically, across multiple realizations~$\omega$, to assess its profitability according to Definition~\ref{definition:profitable_coinv}.}
We show in \tablename~\ref{tab:comparison} how \added{the lower bound ($\nu_{\pazocal N}^{\text{LB}}$) on the coalition’s profitability and the empirical} probabilit\deleted{y}\added{ies} of profitable co-investment for individual players \added{ ($\tilde{\mathbb{P}}_i$) and the coalition (\(\mathbb{\tilde{P}}_\pazocal N\))}
vary with different $\alpha$ values. 

\added{The lower bound applies to the coalition as a whole and provides a conservative theoretical guarantee on its profitability probability.}
The probability~\(\mathbb{\tilde{P}}\) represents the empirical Monte Carlo probability, where profitability is determined by directly evaluating the probability that revenue exceeds payment. This makes it the most economically meaningful and realistic measure of positive payoffs. 
We observe that \added{both} the \added{lower bound on the probability and the empirical} probability of profitable co-investment decreases as uncertainty increases, with the SPs becoming increasingly exposed to risk. At low uncertainty levels (e.g, $\alpha=0.001$ and $\alpha=0.25$), all players, including the coalition as a whole, are almost guaranteed to make profit. However, as uncertainty grows, the gap widens: SPs see their chances of profitability fall significantly, while the \textit{InP} has probability 1 throughout. This reflects the \textit{InP}'s role as a veto player, since the infrastructure cannot be deployed without it, it faces less risk when joining the coalition. 
\added{Moreover, the lower bound is tight, since the lower-bound values are close to the empirical probability of profitable co-investment for the coalition as a whole.}

\begin{table*}[h!]
\centering
\caption{Lower-bound and empirical probabilities of profitable co-investment for individual players and for the coalition~$\pazocal N$, across different levels of uncertainty~$\alpha$}
\label{tab:comparison}
\resizebox{0.7\textwidth}{!}{
\begin{tabular}{c|c|cccc}
\toprule
$\alpha$ 
& \multicolumn{1}{c|}{\added{Lower bound}} 
& \multicolumn{4}{c}{Empirical probability} \\
\cmidrule(lr){2-2}
\cmidrule(lr){3-6}
& $\nu_{\pazocal N}^{\text{LB}}$
& $\tilde{\mathbb{P}}_{\textit{InP}}$ 
& $\tilde{\mathbb{P}}_{\text{SP1}}$ 
& $\tilde{\mathbb{P}}_{\text{SP2}}$ 
& $\tilde{\mathbb{P}}_{\pazocal N}$ \\
\midrule
0.001 & 1        & 1    & 1    & 1    & 1    \\
0.25  & 0.95     & 1    & 0.96    & 1 & 0.96 \\
0.50  & 0.78     & 1    & 0.92 & 0.87 & 0.80\\
0.75  & 0.60     & 1    & 0.84 & 0.77 & 0.64 \\
1.00  & 0.50     & 1    & 0.79 & 0.69 & 0.54 \\
\bottomrule
\end{tabular}}
\end{table*}

\subsubsection{Effect of the investment period on the outcome of the game}
\figurename~\ref{fig:investment_outcomes_new} presents the impact of increasing the investment period $I$ (from 5 to 10 years) on revenue, payment and payoff redistributions for each player when $\alpha = 1$.
As shown in \figurename~\ref{fig:investment_revenue}, the revenue redistribution per player increases with longer investment periods, reflecting the natural increase in cumulative earnings over time. However, the more elongated \deleted{violin} \added{box} plots indicate that uncertainty in revenue grows. This means that while financial gains increase, the risk of fluctuations in earnings also rises.
The payment redistribution per player, observed in \figurename~\ref{fig:investment_payment}, shows that longer investment periods require larger financial contributions from all players. We observe that SP1 and SP2 experience the highest increase in payment, whereas the \textit{InP} sees a gradual rise. This difference is likely because the \textit{InP} plays a veto role, whereas the SPs' contributions depend more on service demand and market fluctuations.
\figurename~\ref{fig:investment_payoff} shows the payoff redistribution per player. It reveals that all players benefit from longer investment periods, with the \textit{InP} achieving the highest payoff. SP1 and SP2 also see significant growth in their payoffs.
Overall, our findings suggest that extending the investment period is generally beneficial, leading to higher cumulative revenues, greater payoffs, and increased payments. However, it also amplifies uncertainty, particularly for SPs.

\begin{figure}[h!]
    \centering
    \begin{subfigure}[b]{0.49\textwidth}
        \centering
        \includegraphics[width=\textwidth]{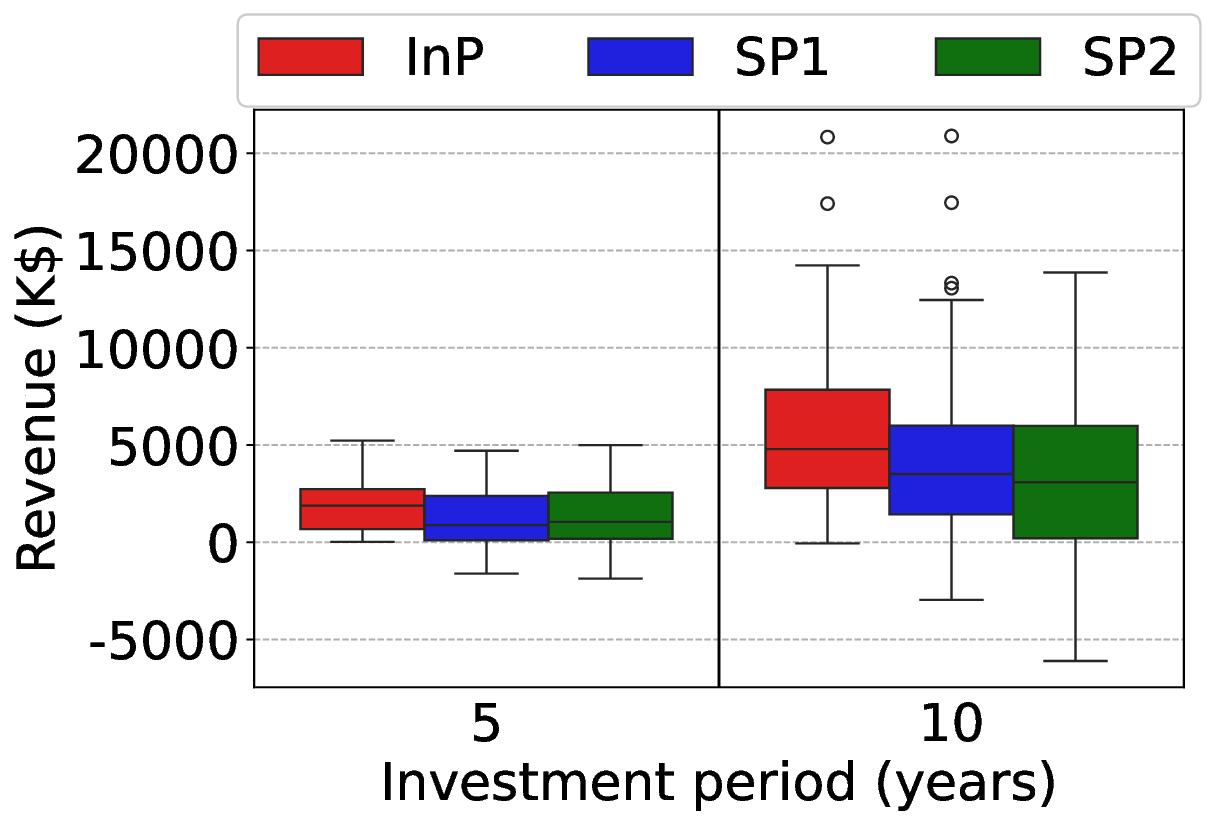}
        \caption{Revenue redistribution}
        \label{fig:investment_revenue}
    \end{subfigure}
    \hfill
    \begin{subfigure}[b]{0.49\textwidth}
        \centering
        \includegraphics[width=\textwidth]{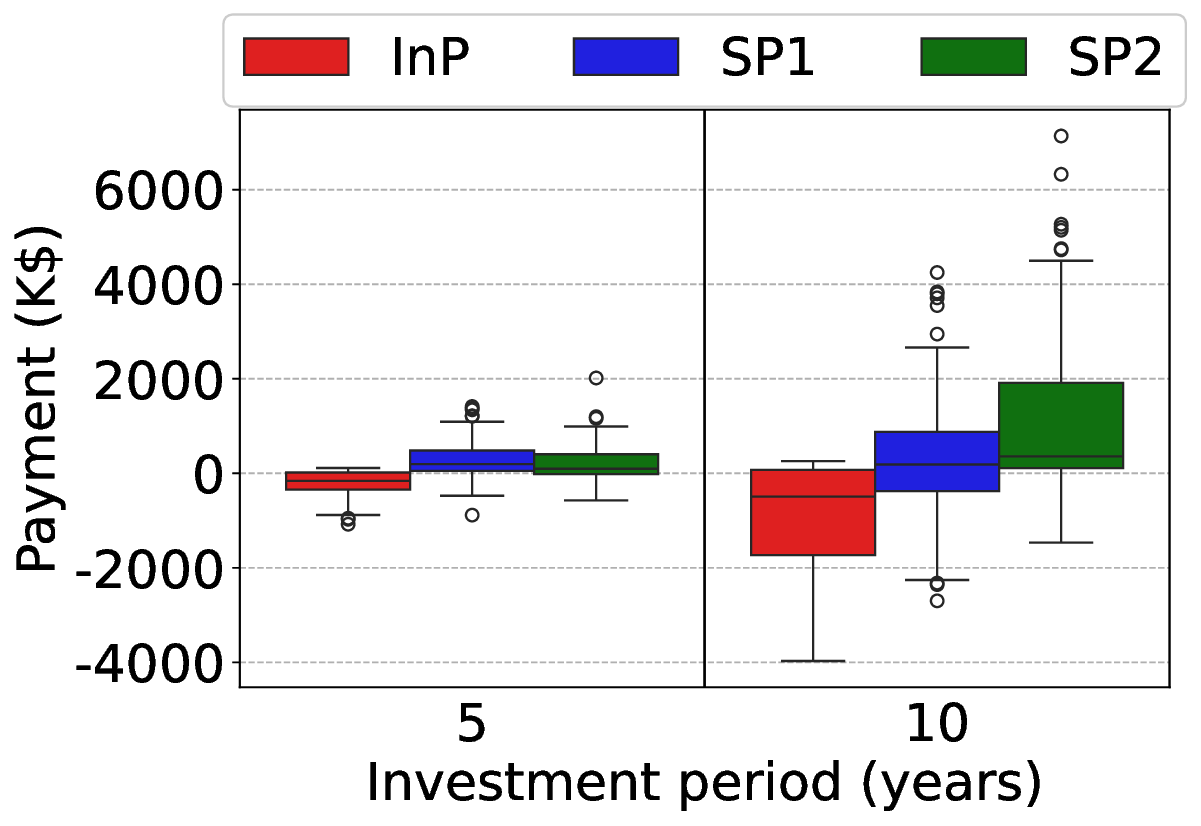}
        \caption{Payment redistribution}
        \label{fig:investment_payment}
    \end{subfigure}
    \hfill
    \begin{subfigure}[b]{0.49\textwidth}
        \centering
        \includegraphics[width=\textwidth]{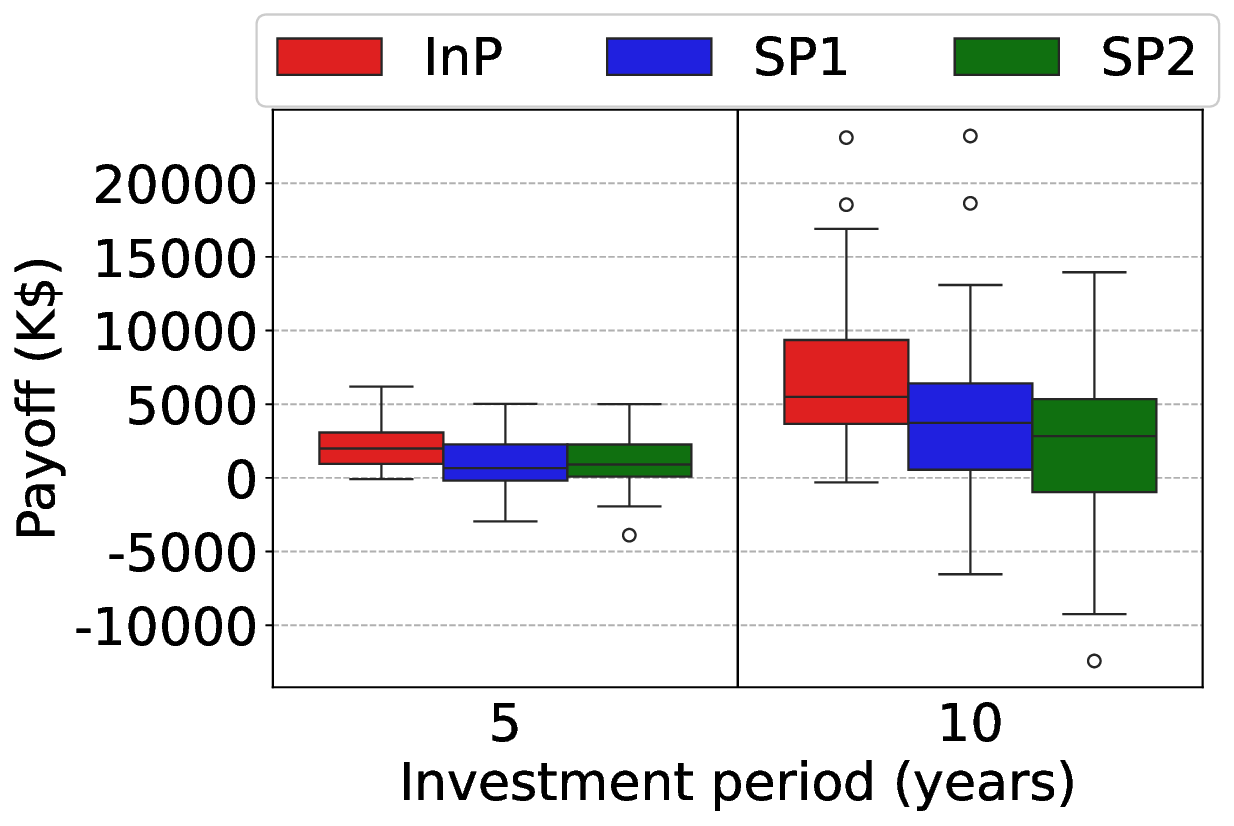}
        \caption{Payoff redistribution}
        \label{fig:investment_payoff}
    \end{subfigure}
    \caption{Effect of the investment period length on revenue, payments, and payoff redistribution, \added{shown using box plots, where the box indicates the interquartile range, the line inside the box is the median, the whiskers represent the variability across realizations, and individual points beyond the whiskers correspond to values that fall outside the main dispersion range.}}

    \label{fig:investment_outcomes_new}
\end{figure}
\subsubsection{Impact of investment period on co-investment profitability}
We show in \figurename~\ref{fig:lower_invest} how the \added{lower bound on} probability that the co-investment is profitable (Fig.~\ref{figure:9a}) \added{, together with the corresponding empirical probability (Fig.~\ref{figure:9b})} evolves with different investment periods $I$ and uncertainty levels~$\alpha$. The results indicate that as the investment period increases, profitability improves significantly.
For low uncertainty level ($\alpha \le 0.25$), 
\added{both the lower bound and the empirical probability remain high, i.e., }
profitability remains guaranteed\added{,} for both investment periods (5 and 10 years). However, as uncertainty increases, shorter investment periods become less profitable than longer ones. 
The increase in profitability with longer investment periods is primarily due to the accumulation of revenues over time, which allows the players to better absorb market fluctuations and operational costs. A longer time horizon enables co-investors to amortize their fixed investment costs over a greater number of operational years, reducing the impact of upfront capital expenditures on overall financial performance. Moreover, demand growth over time contributes to increased revenue, further improving investment viability. Importantly, uncertainty does not increase proportionally with the investment period, meaning that while revenues grow, the relative level of risk remains manageable.
\added{Notably, the empirical probability in Fig.~\ref{figure:9b} closely follows the lower bound in Fig.~\ref{figure:9a}, indicating that the bound is tight across the considered uncertainty levels.}

\begin{figure}[h!]
    \centering
    \begin{subfigure}[b]{0.49\textwidth}
        \centering
        \includegraphics[width=\textwidth]{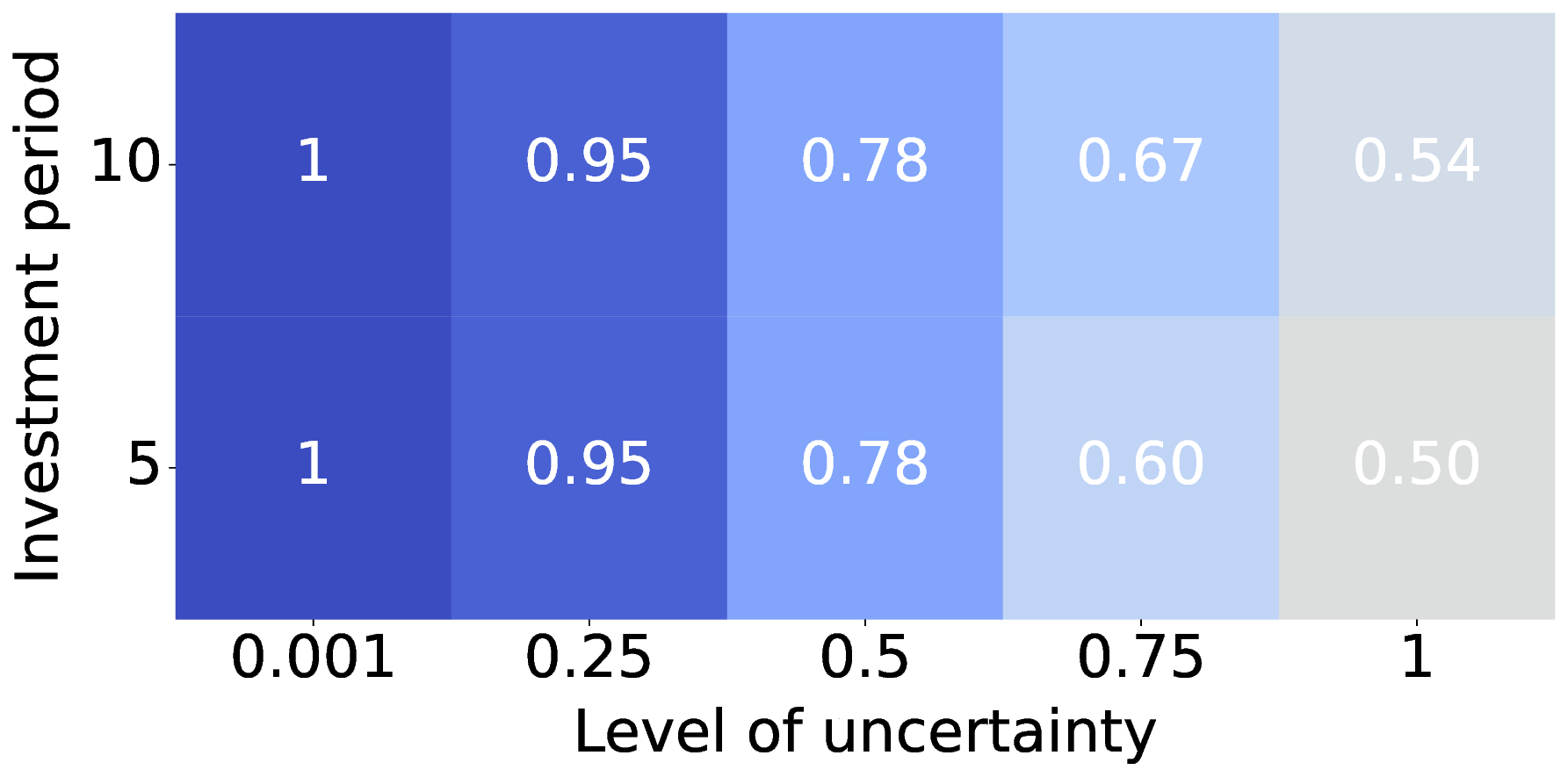}
        \caption{\added{Lower bound on the probability that co-investment is profitable}}
        \label{figure:9a}
    \end{subfigure}
    \hfill
    \begin{subfigure}[b]{0.49\textwidth}
        \centering
        \includegraphics[width=\textwidth]{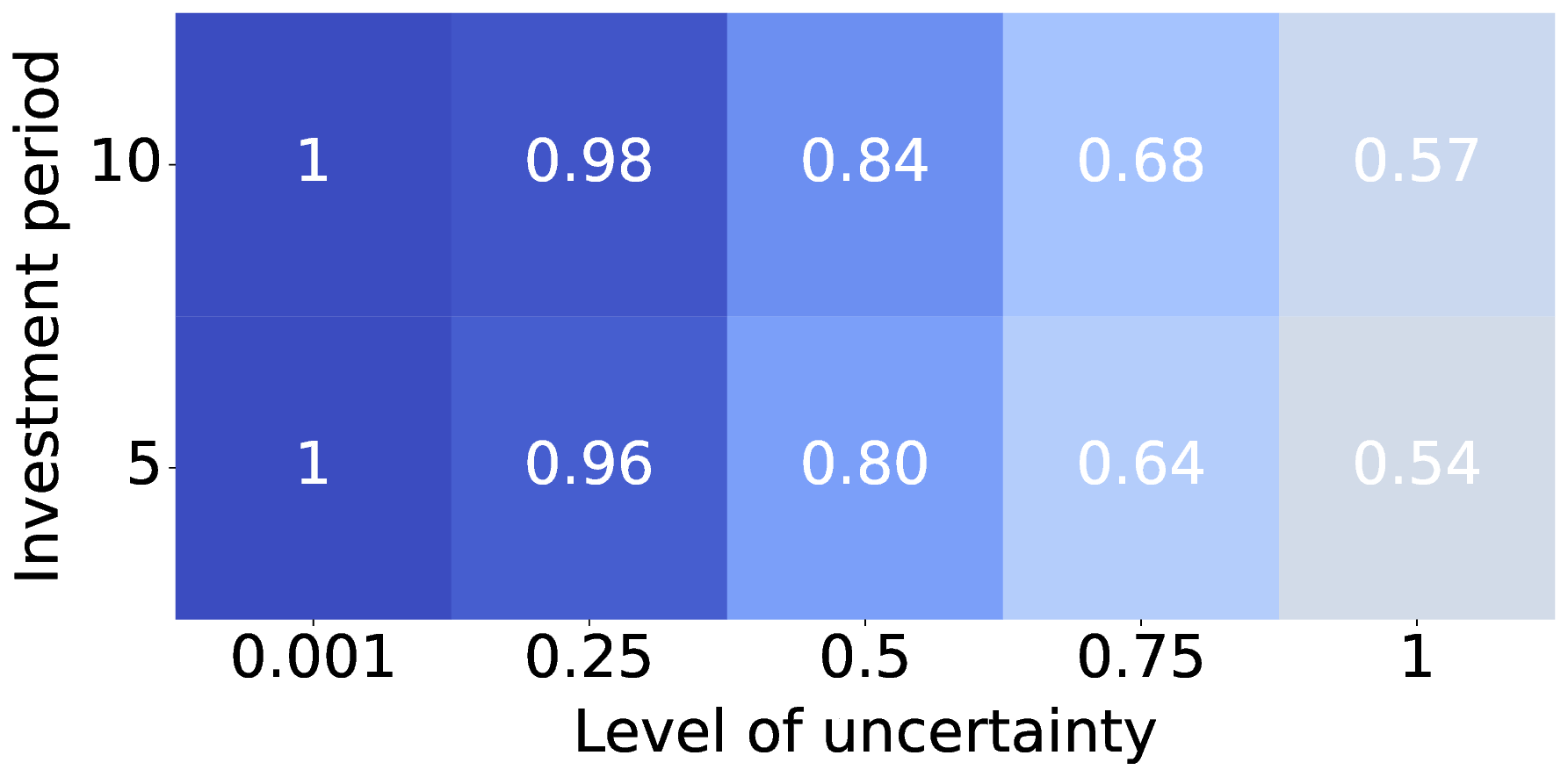}
        \caption{Probability that co-investment is profitable \deleted{for different investment periods  across different levels of uncertainty~$\alpha$}}
        \label{figure:9b}
    \end{subfigure}
    \caption{\added{Lower bound on the probability that co-investment is profitable, together with the empirical probability of profitability obtained from 100 realizations, across different investment periods (5 and 10 years) and different levels of uncertainty $\alpha \in [0,1]$}}
    \label{fig:lower_invest}
\end{figure}

\subsubsection{Impact of investment period and uncertainty on payback year}
At the start of the game, all players make an initial investment to purchase a package that includes computational resources and maintenance. A key financial question is when do players recover their investment? This is determined by the payback year, the point at which the total payoff surpasses the initial investment costs for a given investment period \(I\). 
To track this, we proceed as follows: the realized coalition payoff~$v_{\omega}(\pazocal{N})$~\eqref{eq: generic expression of payoff-realizations_s} is typically calculated over the full investment period~$I$. We now compute this payoff over different periods $Z \leq I$, in order to identify when the payoff becomes positive.
From this moment onward, all additional earnings represent pure profit.
\figurename~\ref{fig:payback} presents this analysis using box plots. Each box plot summarizes the variability of outcomes: the box captures the IQR, representing the middle 50\% of observations, while whiskers extend to the most extreme values within 1.5 times the IQR; points beyond the whiskers are outliers caused by  extreme randomness.
\figurename~\ref{fig:payback5} and \figurename~\ref{fig:payback10} illustrate the distribution of payback years under varying investment periods (\(I = 5\) and \(I = 10\) years) and uncertainty levels (\( \alpha = 0.25, 0.5, 0.75, 1\)).
The red line in each figure represents the expected payback year, which serves as a reference point for understanding deviations due to stochastic effects.  
However, when uncertainty is introduced, the actual payback year can deviate from this expected value. 
For a 5-year investment period, we observe in \figurename~\ref{fig:payback5} that at low~$\alpha$ values (e.g., $\alpha = 0.25$), the payback year is highly concentrated around its expected value. This means that investment recovery time is predictable. However, as $\alpha$ value increases (e.g., $\alpha = 0.75$), the IQR increases and outliers become frequent, reflecting greater variability in recovery time.
In these cases, the payback year may occur earlier or later than expected, depending on the realized stochastic fluctuations.
Similarly, for a longer investment period (\(I = 10\) years), we show in \figurename~\ref{fig:payback10} that extending the investment period leads to an earlier payback year in relative terms. This makes longer investment horizons more attractive. This effect is most evident at low $\alpha$ values, where payoff accumulation follows the expected trend.
Additionally, as $\alpha$ value increases, the payback year becomes broader, reinforcing the idea that greater uncertainty makes it more difficult to predict the exact payback year.
\begin{figure}[h!]
    \centering
    \begin{subfigure}[b]{0.49\textwidth}
        \centering
        \includegraphics[width=\textwidth]{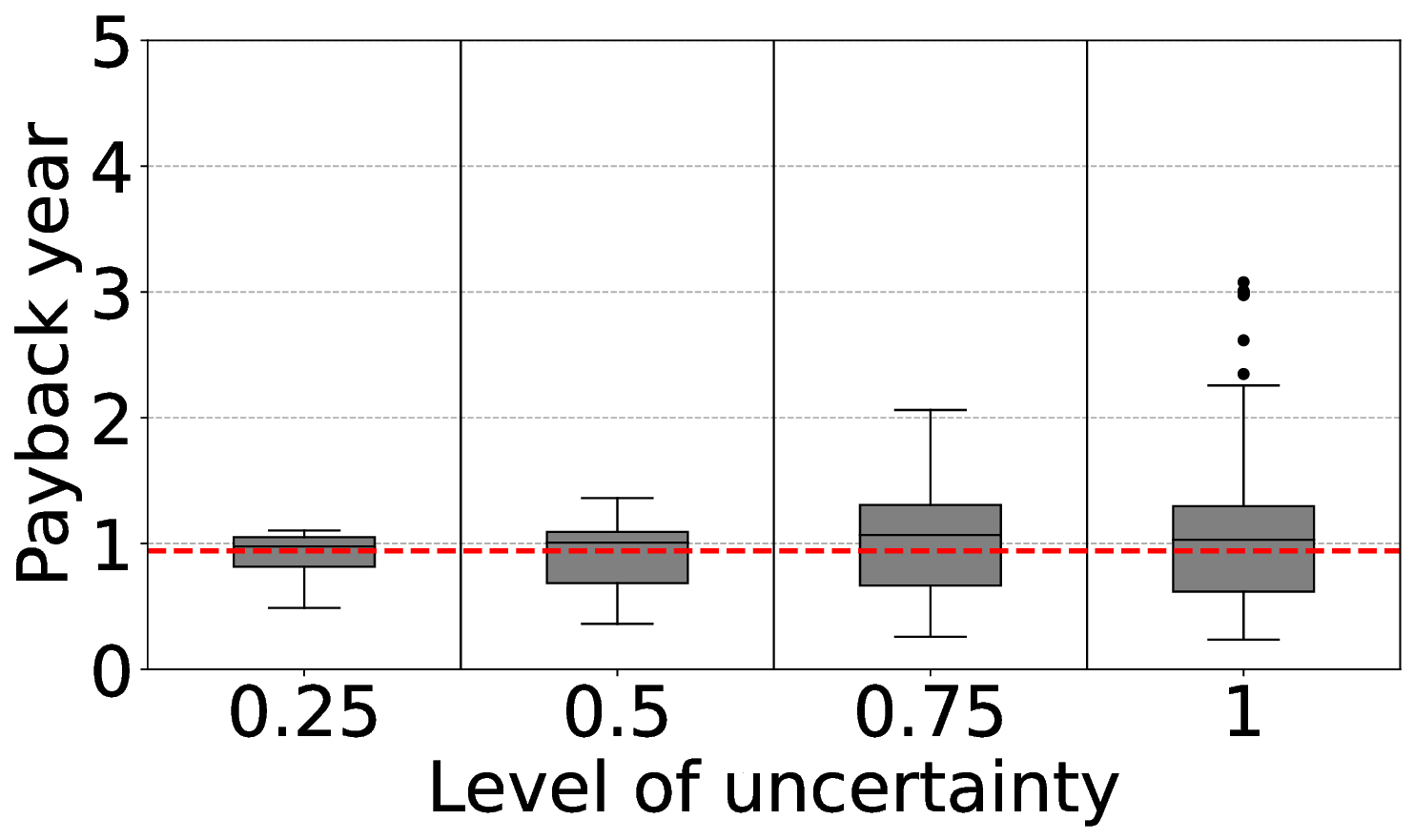}
        \caption{\( I = 5 \) years}
        \label{fig:payback5}
    \end{subfigure}
    \hfill
    \begin{subfigure}[b]{0.49\textwidth}
        \centering
        \includegraphics[width=\textwidth]{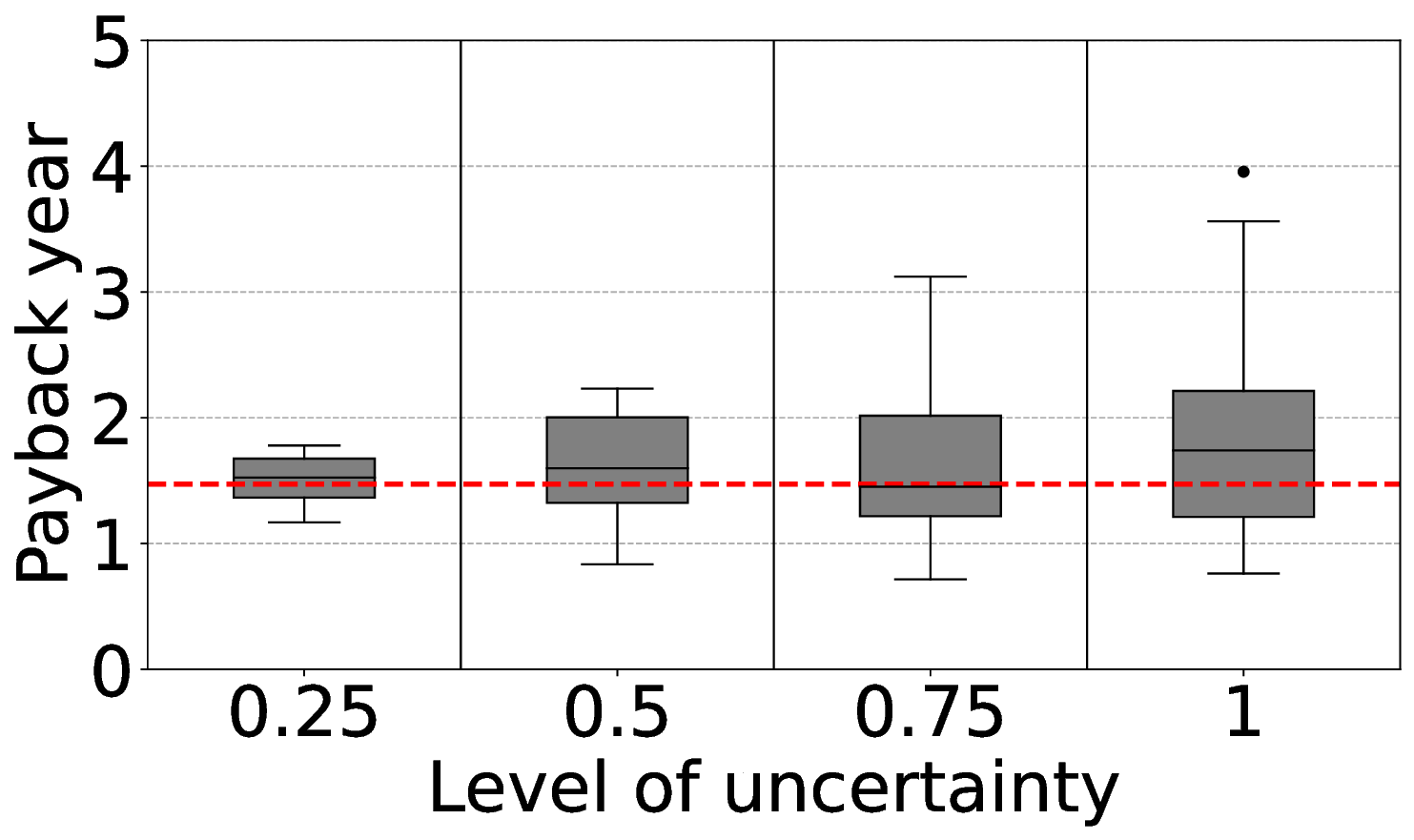}
        \caption{\( I = 10 \) years}
        \label{fig:payback10}
    \end{subfigure}
    \caption{Distribution of the payback year for different investment periods and levels of uncertainty~$\alpha$, presented using box plots. The red line represents the expected payback year}
    \label{fig:payback}
\end{figure}

\subsubsection{Impact of the Hurst parameter on co-investment profitability}
The Hurst parameter, denoted as \( H \), is a key characteristic of fractional Brownian motion and is used in our model to represent the degree of long-range dependence in traffic load fluctuations~(see Section~\ref{section: brownian}). 

\added{Recall that Lemma~\ref{lem:log_derivative} states that the lower bound on the
probability of profitable co-investment provides a profitability guarantee that
is non-decreasing in \(H\) whenever the logarithmic derivative of the variance of
the total utility is smaller than the logarithmic derivative of the squared
expected coalition payoff.
Figure~\ref{fig:11c} directly evaluates this condition by plotting the two
logarithmic derivatives with respect to~\(H\), at the maximum uncertainty level~\(\alpha = 1\).
The figure shows that the condition of Lemma~\ref{lem:log_derivative} is satisfied
for all considered values of~\(H\), and therefore guarantees that the lower bound~\(\nu_{\pazocal N}^{\mathrm{LB}[H]}\) is non-decreasing in \(H\). Moreover, the gap between the two derivatives is largest at \(H = 0.3\) and decreases as \(H\) increases. A large gap implies that uncertainty grows much faster, in relative terms, than the expected payoff, which weakens the resulting profitability guarantee and decreases the lower bound.}

\added{The monotonic behavior of the lower bound is clearly reflected in
Figure~\ref{figure:11a}, which reports
the lower bound on the probability of profitable co-investment across different
levels of uncertainty \(\alpha \in [0,1]\) and across different Hurst parameters
\(H \in \{0.3, 0.5, 0.7\}\), and is further confirmed by the empirical
probabilities shown in Figure~\ref{figure:11b}.}
\deleted{\figurename~\ref{fig:lower_hurst} illustrates its impact on the probability that the co-investment is profitable under different levels of uncertainty ($\alpha$ values).} 
We show that, when \( H = 0.7 \), where traffic patterns exhibit strong long-term dependence and persistency, \added{both the lower bound and the empirical probability of} profitability remain\deleted{s} high even as uncertainty increases. At \(\alpha = 0.75 \), the \added{lower bound} \deleted{probability} remains above 0.5, and even at \( \alpha = 1 \), there is still a chance of profitability. This suggests that when traffic fluctuations are positively \deleted{correlated} \added{dependent} over time, payoff is enhanced, reducing the impact of short-term fluctuations that could negatively impact investment profitability.
On the other hand, when \( H = 0.5 \), traffic follows an \deleted{uncorrelated} \added{independent} pattern with no long-range dependence. In this case, \added{both the lower bound and the empirical probability of} profitability follow a similar trend as in the case of positive \deleted{correlation} \added{dependence} but decline more rapidly as $\alpha$ increases. 
However, the situation changes for \( H = 0.3 \), where traffic exhibits negative \deleted{correlation} \added{dependence}, an increase in demand is often followed by a drop, and vice versa. We show that the lower bound decreases to about 0.55 at \(\alpha=0.25\) and drops sharply for larger \(\alpha\). This suggests that in bursty and highly fluctuating traffic conditions, payoff is significantly compromised, making co-investment more risky.
\added{Notably, Fig.~\ref{figure:11b} closely follows Fig.~\ref{figure:11a}, guaranteeing the bound’s tightness.}

\begin{figure}[h!]
    \centering
    \begin{subfigure}[b]{0.28\textwidth}
        \centering
        \includegraphics[width=\textwidth]{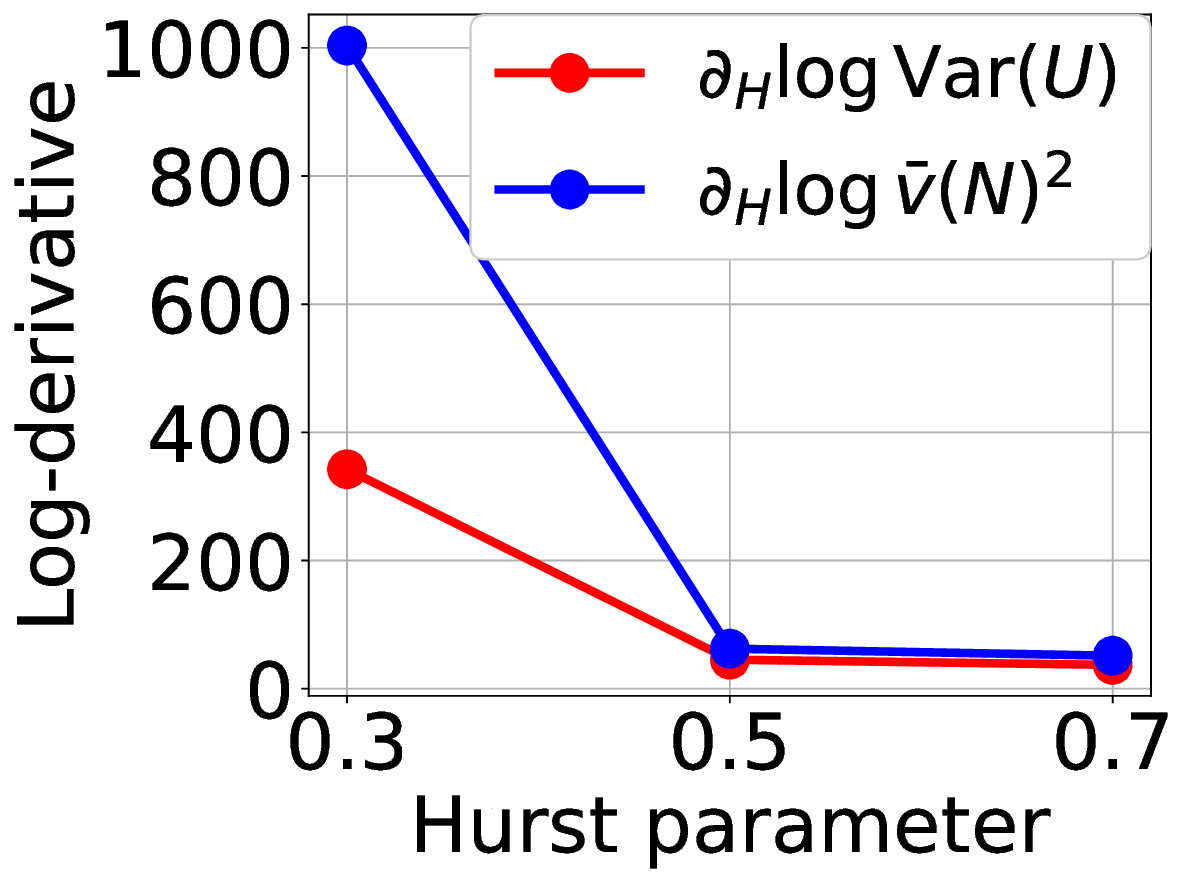}
        \caption{Derivatives with respect to $H$ for $\alpha = 1$}
        \label{fig:11c}
    \end{subfigure}
    \hfill
    \begin{subfigure}[b]{0.35\textwidth}
        \centering
        \includegraphics[width=\textwidth]{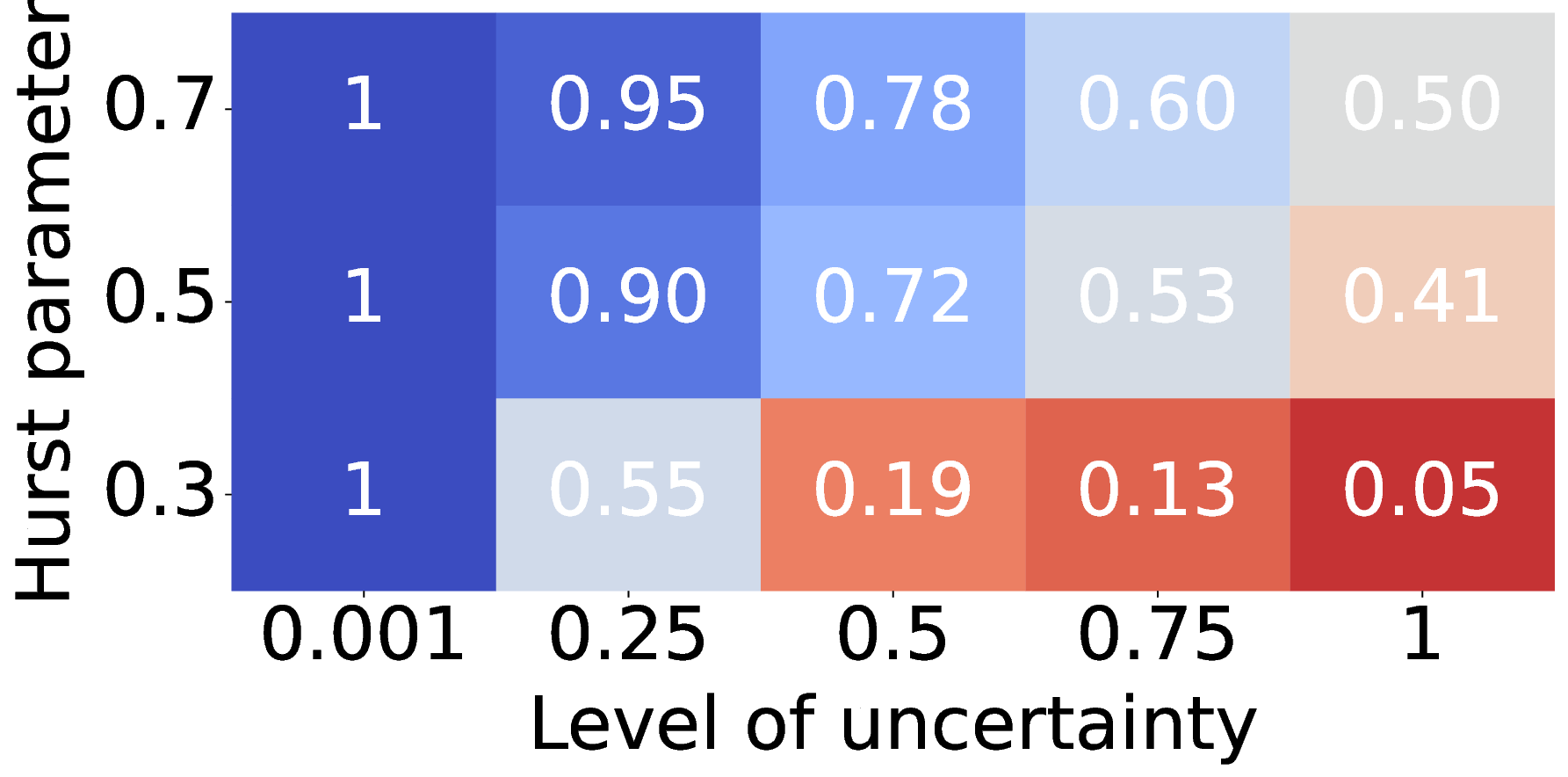}
        \caption{\added{Lower bound on the probability of profitable co-investment}}
        \label{figure:11a}
    \end{subfigure}
    \hfill
    \begin{subfigure}[b]{0.35\textwidth}
        \centering
        \includegraphics[width=\textwidth]{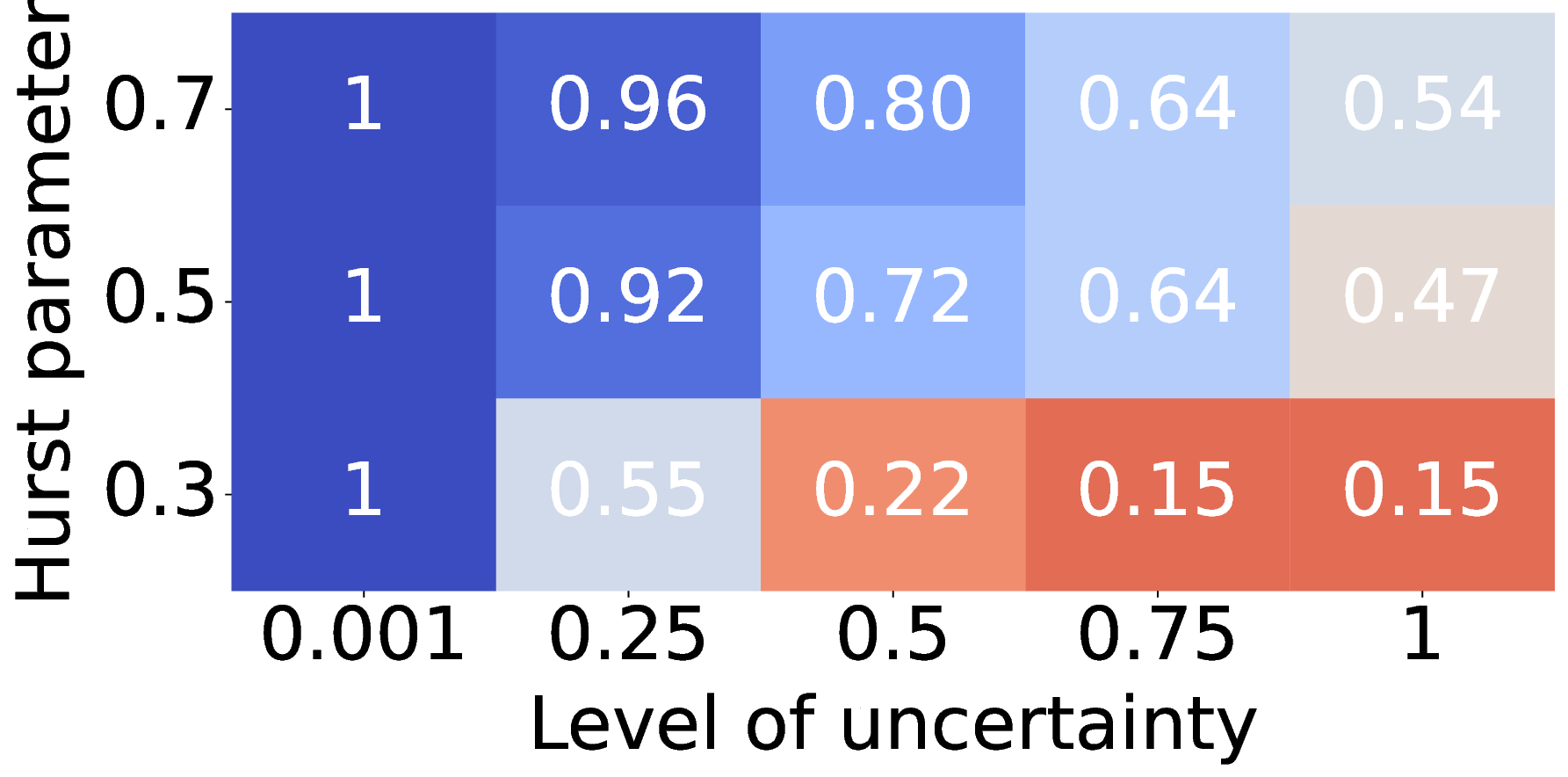}
        \caption{Probability that co-investment is profitable \deleted{under different Hurst parameter values and levels of uncertainty~$\alpha$}}
        \label{figure:11b}
    \end{subfigure}
    \caption{\added{Logarithmic derivatives of the variance of total utility and of the squared total expected payoff with respect to the Hurst parameter $H$ evaluated at $\alpha = 1$, together with the lower bound on the probability that co-investment is profitable and the empirical probability of profitability obtained from 100 realizations, across different Hurst parameter values (0.3, 0.5, and 0.7) and levels of uncertainty $\alpha \in [0,1]$}}
    \label{fig:lower_hurst}
\end{figure}

\section{Conclusion} \label{section: conclusion}
We proposed a co-investment scheme between an \textit{InP} and multiple SPs to support the deployment and maintenance of \added{M}EC infrastructure under uncertain traffic loads and hence revenues. We modeled the problem using stochastic coalitional game theory, introducing a framework to assess whether all players benefit from joining the grand coalition. 
\deleted{In the case of uncorrelated traffic, w}
\added{W}e \added{derived} a theoretical lower bound on the probability that the grand coalition is stable. Our results show that this lower bound is high when the SPs' \deleted{load} \added{revenue}s are of similar magnitude. \deleted{and when the investment period is sufficiently long to smooth out uncertainty. For \deleted{correlated} \added{dependent} traffic, modeled using fractional Brownian motion, we compute the empirical probability of profitability through simulations. The results indicate a high probability of profitability even under a high level of uncertainty.}
\added{In addition, we established a lower bound on the probability of profitable co-investment. Unlike stability, this profitability guarantee remains robust under strong temporal dependence, high levels of uncertainty, and pronounced asymmetries among SPs, and it increases with the investment period. We validate these findings through numerical experiments covering different numbers of SPs, multiple uncertainty levels, and diverse traffic models, including independent demand fluctuations, dependent traffic modeled via fractional Brownian motion, and real traffic traces.}
The \textit{InP}, as a veto player, faces minimal risk due to its important role in infrastructure deployment. Overall, the proposed \added{co-investment scheme} \deleted{strategy} enables \deleted{fair, and efficient infrastructure sharing,} \added{the deployment and maintenance of new large-scale infrastructure by allowing players to share costs and revenues and to achieve payoff maximization,} even in the presence of significant demand uncertainty.

\deleted{Our model can be extended to a dynamic framework. In a dynamic setting, players may join or leave the grand coalition over time, based on individual incentives or evolving traffic load. Such an extension would require modeling coalition formation as a multi-stage game, where payoff redistribution mechanisms ensure stability at each stage. This constitutes our future work perspective. }

\added{Several directions could extend the proposed scheme.
First, an important research avenue is to analyse how regulatory instruments, such as entry and exit rules, or access pricing, affect the long-term viability of the co-investment scheme.
A second extension is to adapt the model to multi-phase, evolving infrastructure projects, where investment, demand, and costs evolve across successive epochs. In such cases, players may repeatedly adjust their participation and investment as new technologies emerge or as demand conditions change. This too relates to dynamic participation, in which players may enter, remain in, or leave the co-investment at predetermined epochs. Their decisions would be guided by the evolution of demand, technology, and coalition stability and profitability, while investment decisions would need to be updated accordingly.
Another promising direction is to introduce risk-aware infrastructure dimensioning. The present scheme assumes players decide the capacity to be deployed so as to maximize expected payoffs. In practice, actors may exhibit risk aversion when making such decisions, especially under volatile demand. Incorporating risk measures, such as Conditional Value at Risk (CVaR), would allow players to better protect against adverse future outcomes and may affect (positively, we believe) coalition stability.
Finally, integrating machine-learning techniques for improved demand forecasting could enhance the predictive power of the model. Data-driven prediction methods, including reinforcement learning, can serve as inputs to the uncertainty parameters of the co-investment scheme.}

\bmhead{Acknowledgements}
This work was supported by the ANR under the France 2030 program, grant NF-NAI: ANR-22-PEFT-0003.

\section*{Declarations}
\textbf{Conflict of interest}
The authors declare that they have no conflicts of interest.

\begin{appendices}
\section{Proof of Proposition~\ref{proposition:pointless}}\label{appendix:pointless-propo}
The \textit{InP} is responsible for installing the infrastructure. Without the \textit{InP}, there is no infrastructure, and thus no revenues can be collected, due to Assumption~\ref{ass:nonnegative}.
Similarly, if no SP joins a coalition, no revenues are collected from end-users, according to Assumption~\ref{ass:inp-zero}, and thus such a coalition would be pointless.

\section{Proof of Proposition~\ref{prop:allocation}}\label{appendix:proof_allocation}
If the \textit{InP} does not belong to coalition~$\pazocal S$, then no resources can be installed, as only the \textit{InP} is responsible for infrastructure deployment (see Proposition~\ref{proposition:pointless}). In this case, the total installed capacity is zero and no player receives any allocation, which directly leads to~\eqref{eq:InPnot}.

If the \textit{InP} is included in~$\pazocal S$, then the coalition can install resources~$\deleted{C}\added{\mathbf C}_{\pazocal S}^*$ and apply the optimal dynamic allocation~$\{\deleted{h}\added{\mathbf h}_{i, \pazocal S}^{t*}\}$ computed from the solution of the optimization problem~\eqref{equation: optimal solution}. In this case, the installed capacity and the shares allocated to players in~$\pazocal S$ are equal to those determined by the global optimization.

\section{Proof of Proposition~\ref{prop:convexity-of-problem}}
\label{appendix:prop_convex}
We prove that the optimization problem~\eqref{equation: optimal solution}--\eqref{eq:constraint-positive} is convex \added{(in the sense of a concave maximization over a convex set)} by showing that:
(i) the objective function is concave, and 
(ii) the feasible set defined by the constraints is convex.

From~\eqref{equation: expected payoff}, the objective function in~\eqref{equation: optimal solution} consists of two parts: the sum of utility functions and a cost function \added{and can be written as}:
\begin{equation}
    \label{eq:objective-appendix}
    \sum_{t \in \pazocal T} \sum_{i \in \pazocal N}
    \bar{u}_i^t(\deleted{h}\added{\mathbf h}_i^t)
    - \text{Cost}(I, \deleted{C}\added{\mathbf C})
\end{equation}

Each utility function~$\bar{u}_i^t(\cdot)$ is concave \added{in $\mathbf h_i^t$} by assumption. Therefore, their sum is concave as well \added{in the collection of variables $\{\mathbf h_i^t\}_{i\in\pazocal S,\,t\in\pazocal T}$} \citep[Section~3.2]{Boyd2004}.
\deleted{The cost function~$\text{Cost}(I, C)$ is convex in~$C$. To express it in terms of the decision variables~$\{h_i^t\}$, we use the affine mapping:
\begin{equation}
C = \frac{1}{|\pazocal T|} \sum_{t \in \pazocal T} \sum_{i \in \pazocal N} h_i^t    
\end{equation}
Since the composition of a convex, non-decreasing function with an affine function preserves convexity~\citep[Page~2]{roberts1974convex}, 
it follows that~$\text{Cost}(I, C)$ is convex in~$\{h_i^t\}$. Therefore, the objective function~\eqref{equation: optimal solution}, which is a concave term (sum of utilities) minus a convex term (cost), is concave overall (thus, (i) holds). The constraints~\eqref{eq:constraint-sum}-\eqref{eq:constraint-positive} are affine, so the feasible set is convex~\citep[Section~4.2.1]{Boyd2004} (thus (ii) holds). Thus, the entire optimization problem is convex. Therefore, the Karush-Kuhn-Tucker (KKT) conditions can be applied to determine the optimal solution~\citep[Section~5.5.3]{Boyd2004}.}

\added{
By assumption, the cost function
\[
\text{Cost}(I,\cdot): \mathbb R_+^{K\times M} \to \mathbb R,
\quad
\mathbf C \mapsto \text{Cost}(I,\mathbf C)
\]
is convex in $\mathbf C$. Hence, the function $-\text{Cost}(I,\mathbf C)$ is concave in $\mathbf C$. Therefore, the objective function~\eqref{eq:objective-appendix} is the sum of two concave functions: the concave sum of expected utilities and the concave function $-\text{Cost}(I,\mathbf C)$. It follows that the objective in~\eqref{equation: optimal solution} is concave in the joint decision variables $(\mathbf C,\{\mathbf h_i^t\}_{i\in\pazocal S,\,t\in\pazocal T})$. This establishes~(i).}

\added{Since all constraints~\eqref{eq:constraint-sum}--\eqref{eq:constraint-positive} are affine functions of the decision variables $(\mathbf C,\{\mathbf h_i^t\})$, the feasible set is convex~\citep[Section 4.2.1]{Boyd2004}. This establishes~(ii).}

\added{
We have shown that the objective function of Problem~\eqref{equation: optimal solution}--\eqref{eq:constraint-positive} is concave and that its feasible set is convex. Therefore, this problem is a concave maximization over a convex set, and is thus a convex optimization problem in the sense of convex analysis. Thus, the Karush--Kuhn--Tucker (KKT) conditions are both necessary and sufficient for optimality~\citep[Section~5.5.3]{Boyd2004}. Hence, the optimal solution of~\eqref{equation: optimal solution}--\eqref{eq:constraint-positive} can be characterized via its KKT conditions.}

\section{Proof of Theorem~\ref{thm:convexity-of-game}}
\label{appendix:th_nominal}
Thanks to~Proposition~\ref{prop:allocation} and Assumption~\ref{ass:nonnegative} and~\eqref{eq:cost-zero}, \eqref{equation:v_bar_s} becomes:
\begin{equation}
\Bar{v}(\pazocal S) = 
\begin{cases}
\displaystyle
\sum\limits_{t \in \pazocal T} \sum\limits_{i \in \pazocal S} 
\bar{u}_i^t(
\deleted{h}
\added{\mathbf h}_{i, \pazocal S}^{t*}
)
- \text{Cost}(I,
\deleted{C}\added{\mathbf C}^*_{\pazocal S}
), & \text{if  \textit{InP} } \in \pazocal S \\[8pt]
0, & \text{otherwise}
\end{cases}
\label{equation: Fc}
\end{equation}
We first show that the nominal value function~$\bar{v}(\cdot)$ is monotonically non-decreasing, i.e., for all coalitions $\pazocal R \subseteq \pazocal S \subseteq \pazocal N$, it holds that:
\begin{equation}
    \bar{v}(\pazocal R) \leq \bar{v}(\pazocal S)
    \label{equation:v_monot}
\end{equation}
To prove~\eqref{equation:v_monot}, we distinguish three cases based on the presence of the \textit{InP} in the coalitions:

\begin{itemize}
    \item Suppose $\textit{InP} \notin \pazocal S$. Since $\pazocal R \subseteq \pazocal S$, it follows that $\textit{InP} \notin \pazocal R$ as well. Then, via~\eqref{equation: Fc}, we have $\bar{v}(\pazocal R) = \bar{v}(\pazocal S) = 0$, so~\eqref{equation:v_monot} holds.
    
    \item Suppose $\textit{InP} \notin \pazocal R$ but $\textit{InP} \in \pazocal S$. Then, again via~\eqref{equation: Fc}, we have $\bar{v}(\pazocal R) = 0$, while $\bar{v}(\pazocal S) \ge 0$. The inequality in~\eqref{equation:v_monot} is trivially satisfied.
    
    \item Suppose $\textit{InP} \in \pazocal R \subseteq \pazocal S$. In this case, define a suboptimal solution for coalition~$\pazocal S$ by using the optimal capacity~
    \deleted{$C$}\added{matrix $\mathbf C$}$_{\pazocal R}^*$
    and resource allocation \added{matrices}~
    $\{\deleted{h}\added{\mathbf h}_{i, \pazocal R}^{t*}\}_{t \in \pazocal T}$
    computed for~$\pazocal R$, via~\eqref{equation: optimal solution}–\eqref{eq:constraint-positive}, for all players in~$\pazocal R$, and assigning \added{the} zero resource \added{matrices} to all players in~$\pazocal S \setminus \pazocal R$. Let the value of this suboptimal solution be denoted by~$\bar{v}_{\text{sub}}(\pazocal S)$. By construction, this gives:
    \begin{equation}
        \bar{v}_{\text{sub}}(\pazocal S) = \bar{v}(\pazocal R)    
    \end{equation}
    Since~$\bar{v}(\pazocal S)$ is the optimal value over all feasible allocations for coalition~$\pazocal S$, and the allocation used to define~$\bar{v}_{\text{sub}}(\pazocal S)$ is one such feasible solution, we have:
    \begin{equation}
        \bar{v}(\pazocal S) \ge \bar{v}_{\text{sub}}(\pazocal S) = \bar{v}(\pazocal R)
    \end{equation}
\end{itemize}

Therefore, the value function~$\bar{v}(\cdot)$ is monotonically non-decreasing (\eqref{equation:v_monot} holds).

We now prove that there is an allocation that satisfies the conditions of the core, namely efficient,
\(
\sum_{i\in \pazocal N} x_i=\bar v(\pazocal N)
\)
and stable,
\(
\sum_{i\in \pazocal S} x_i \ge \bar v(\pazocal S), \qquad \forall \pazocal S \subseteq \pazocal N.
\)

The idea is to construct an allocation in which the whole payoff of the grand coalition is assigned to the \textit{InP}, while all the other players receive zero. This allocation is introduced only to prove that the core is non-empty; it is not intended to be the payoff-sharing rule used in practice.

Consider the allocation
\(
\mathbf{x}=(x_i)_{i\in \pazocal N}\in\mathbb{R}^{|\pazocal N|}
\)
defined by
\[
x_{\textit{InP}}=\bar v(\pazocal N), \qquad x_i=0,\quad \forall i\in \pazocal N\setminus\{\textit{InP}\}.
\]
This allocation is efficient because
\(
\sum_{i\in \pazocal N} x_i=\bar v(\pazocal N).
\)

Let \(\pazocal S\subseteq \pazocal N\) and we proceed in two cases depending on the presence of the \textit{InP} in the coalition.
\begin{itemize}
    \item If \(\textit{InP}\notin \pazocal S\), then by~\eqref{equation: Fc},
\(
\bar v(\pazocal S)=0.
\)
Therefore,
\(
\sum_{i\in \pazocal S} x_i = 0 = \bar v(\pazocal S).
\)
In other words, coalitions that exclude the \textit{InP} do not lose anything under this allocation, since they would obtain zero even outside the grand coalition.

    \item If \(\textit{InP}\in \pazocal S\), then
\(
\sum_{i\in \pazocal S} x_i = x_{\textit{InP}}=\bar v(\pazocal N).
\)
Since \(\pazocal S\subseteq \pazocal N\) and \(\bar v(\cdot)\) is monotone non-decreasing by~\eqref{equation:v_monot}, we have
\(
\bar v(\pazocal S)\le \bar v(\pazocal N).
\)
Hence, even under this radical allocation, no coalition containing the \textit{InP} can obtain more by deviating from the grand coalition. Therefore,
\(
\sum_{i\in \pazocal S} x_i=\bar v(\pazocal N)\ge \bar v(\pazocal S).
\)
\end{itemize}

Therefore,
\(
\sum_{i\in \pazocal S} x_i \ge \bar v(\pazocal S), \qquad \forall \pazocal S\subseteq \pazocal N.
\)

Thus, \(\mathbf{x}\) belongs to the core of the nominal game, i.e.,
\(
\mathbf{x}\in \mathrm{core}(\pazocal N,\bar v).
\)
Hence,
\(
\mathrm{core}(\pazocal N,\bar v)\neq \emptyset.
\)
This concludes the proof.

\section{Proof of Lemma~\ref{lemma: tightness-bound}}
\label{appendix:delta_hut}
The stability value~$\sigma(\pazocal{N}, \bar{v})$ is given by the following optimization problem~[Page 7]\citep{doan2014robust}:  
Given a parameter $\epsilon \geq 0$
\begin{equation}
\sigma(\pazocal{N}, \bar{v}) = \max_{\mathbf{x}, \epsilon} -\epsilon
\label{equation:max_sigma}
\end{equation}
subject to:
\begin{align}
\sum_{i \in \pazocal{N}} x_i &= \bar{v}(\pazocal{N}) \label{eq:efficiency_constraint} \\
\sum_{i \in \pazocal{S}} x_i - \bar{v}(\pazocal{S}) &\geq -\epsilon \quad \forall \pazocal{S} \subsetneq \pazocal{N}, \pazocal{S} \neq \emptyset \label{eq:core_constraint}
\end{align}
Here, $\mathbf{x} = (x_1, x_2, \dots, x_n)$ represents a payoff vector, where $x_i$ is the allocation assigned to player~$i$.

Constraint~\eqref{eq:core_constraint} can be rewritten as:
\begin{equation}
-\epsilon \leq \min_{\substack{\pazocal{S} \subsetneq \pazocal{N} \\ \pazocal{S} \neq \emptyset}} \left( \sum_{i \in \pazocal{S}} x_i - \bar{v}(\pazocal{S}) \right)
\label{eq:core_rewr}
\end{equation}
Since the objective is to maximize $-\epsilon$~(\eqref{equation:max_sigma}), the optimal value is obtained when the inequality in~\eqref{eq:core_rewr} holds with equality,
which leads to the following reformulation of~\eqref{equation:max_sigma}-\eqref{eq:core_constraint}:
\begin{equation}
\sigma(\pazocal{N}, \bar{v}) = \max_{\mathbf{x}} \min_{\substack{\pazocal{S} \subsetneq \pazocal{N} \\ \pazocal{S} \neq \emptyset}} \left( \sum_{i \in \pazocal{S}} x_i - \bar{v}(\pazocal{S}) \right)
\label{equation: sigma}
\end{equation}
subject to~\eqref{eq:efficiency_constraint}, where we maximize $\sigma(\pazocal N, \bar{v})$ by choosing the best vector $\mathbf{x}$, ensuring that the minimum over all coalitions $\pazocal S$ is as large as possible.

Now, we define:
\begin{equation}
\hat{\sigma}(\pazocal{N}, \bar{v}) \defeq \min_{\substack{\pazocal{S} \subsetneq \pazocal{N} \\ \pazocal{S} \neq \emptyset}} \left( \sum_{i \in \pazocal{S}} \bar{x}_i - \bar{v}(\pazocal{S}) \right)
\label{eq:sigma_hut}
\end{equation}
where~$\bar{x}_i$ is the expected payoff of player $i$.
It follows directly from the definition~\eqref{eq:sigma_hut} that:
\begin{equation}
\hat{\sigma}(\pazocal{N}, \bar{v}) \leq \sigma(\pazocal{N}, \bar{v})
\label{eq:sigma_bound}
\end{equation}

We then define:
\begin{equation}
y_{\hat{\sigma}}(\pazocal{S}) \defeq
\begin{cases}
\frac{\bar{v}(\pazocal{S}) + \hat{\sigma}(\pazocal{N}, \bar{v})}{\bar{v}(\pazocal{N})}, & \text{if } \pazocal{S} \subsetneq \pazocal{N} \\
1, & \text{if } \pazocal{S} = \pazocal{N}
\end{cases}
\end{equation}

and
\begin{equation}
\hat{\delta} \defeq \min\left( \frac{\bar{v}(\pazocal{N})}{|\pazocal{N}|}, \frac{\hat{\sigma}(\pazocal{N}, \bar{v})}{\displaystyle\max_{\substack{\pazocal{S} \subsetneq \pazocal{N} \\ \pazocal{S} \neq \emptyset}} ( |\pazocal{S}| + (|\pazocal{N}| - 2|\pazocal{S}|) \cdot y_{\hat{\sigma}}(\pazocal{S}) )} \right)
\end{equation}

Similarly, from~\citep[Theorem 2]{doan2014robust}, $y_{\sigma}(\pazocal{S})$ is defined as:
\begin{equation} 
y_{\sigma}(\pazocal{S}) \defeq
\begin{cases}
\frac{\bar{v}(\pazocal{S}) + \sigma(\pazocal{N}, \bar{v})}{\bar{v}(\pazocal{N})}, & \text{if } \pazocal{S} \subsetneq \pazocal{N} \\
1, & \text{if } \pazocal{S} = \pazocal{N}
\end{cases}
\label{eq:y_doan}
\end{equation}
and
\begin{equation}
\delta \defeq \min\left( \frac{\bar{v}(\pazocal{N})}{|\pazocal{N}|}, \frac{\sigma(\pazocal{N}, \bar{v})}{\displaystyle\max_{\substack{\pazocal{S} \subsetneq \pazocal{N} \\ \pazocal{S} \neq \emptyset}} ( |\pazocal{S}| + (|\pazocal{N}| - 2|\pazocal{S}|) \cdot y_{\sigma}(\pazocal{S}) )} \right)
\label{eq:delta_doan}
\end{equation}

From the definition of $\delta$ in~\eqref{eq:delta_doan}, and by directly substituting the expression for $y_{\sigma}(\pazocal{S})$ from~\eqref{eq:y_doan}, we define the function:
\begin{equation}
f(\sigma(\pazocal{N}, \bar{v})) = \frac{\sigma(\pazocal{N}, \bar{v})}{\displaystyle\max_{\substack{\pazocal{S} \subsetneq \pazocal{N} \\ \pazocal{S} \neq \emptyset}} \left( |\pazocal{S}| + (|\pazocal{N}| - 2|\pazocal{S}|) \cdot \left( \frac{\bar{v}(\pazocal{S}) + \sigma(\pazocal{N}, \bar{v})}{\bar{v}(\pazocal{N})} \right) \right)}
\label{eq:f_sigma}
\end{equation}
Dividing both the numerator and denominator in~\eqref{eq:f_sigma} by $\sigma(\pazocal{N}, \bar{v}) > 0$, we obtain:
\begin{equation}
    f(\sigma(\pazocal{N}, \bar{v})) = \frac{1}{\displaystyle\max_{\substack{\pazocal{S} \subsetneq \pazocal{N} \\ \pazocal{S} \neq \emptyset}} \left( \frac{|\pazocal{S}|}{\sigma(\pazocal{N}, \bar{v})} + \frac{(|\pazocal{N}| - 2|\pazocal{S}|)}{\bar{v}(\pazocal{N})} \cdot \left( \frac{\bar{v}(\pazocal{S})}{\sigma(\pazocal{N}, \bar{v})} + 1 \right) \right)}
    \label{eq:f_divide_sig}
\end{equation}

It is clear that the denominator in~\eqref{eq:f_divide_sig} decreases as \(\sigma(\pazocal N, \bar{v})\) increases, implying~\eqref{eq:f_sigma} is increasing. Therefore, since \(\hat{\sigma} \leq \sigma\)~\eqref{eq:sigma_bound}, we conclude:
\begin{equation}
\hat{\delta} \leq \delta
\end{equation}
Thus, the lemma is proved.

\section{\added{Proof of Lemma~\ref{lem:log_derivative}}}\label{appendix:lemma_parametrize}
\added{
From Theorem~\ref{thm:prof}, the lower bound on the probability of
profitable co-investment is given by:
\begin{equation}
\nu_{\pazocal N}^{\mathrm{LB}[H]}
=
\frac{1}{
1
+
\dfrac{
\mathrm{Var}_{\omega}(U^{[H]})
}{
\bar v^{[H]}(\pazocal N)^2
}
},
\label{eq:nu_LB_H}
\end{equation}
where
\(
\bar v^{[H]}(\pazocal N)
=
\mathbb{E}_{\omega}[U^{[H]}]
-
\text{Cost}(I,\mathbf C^{[H]}),
\qquad
\bar v^{[H]}(\pazocal N) > 0
\)}

\added{
Define the ratio
\begin{equation}
R(H)
\;\defeq\;
\frac{
\mathrm{Var}_{\omega}(U^{[H]})
}{
\bar v^{[H]}(\pazocal N)^2
}
\label{eq:R_H}
\end{equation}
The expression~\eqref{eq:nu_LB_H} shows that
$\nu_{\pazocal N}^{\mathrm{LB}[H]}$
is a decreasing function of $R(H)$.
Therefore, to prove that
$\nu_{\pazocal N}^{\mathrm{LB}[H]}$
is non-decreasing in $H$,
it is sufficient to show that $R(H)$ is non-increasing in $H$.}

\added{Taking the logarithmic derivative of $R(H)$ with respect to $H$ yields
\begin{align}
\frac{\mathrm d}{\mathrm d H} \log R(H)
&=
\frac{\mathrm d}{\mathrm d H}
\log\!\left(
\frac{
\mathrm{Var}_{\omega}(U^{[H]})
}{
\bar v^{[H]}(\pazocal N)^2
}
\right)
\nonumber\\[0.4em]
&=
\frac{\mathrm d}{\mathrm d H}
\left(
\log \bigl(\mathrm{Var}_{\omega}(U^{[H]})\bigr)
-
\log \bigl(\bar v^{[H]}(\pazocal N)^2\bigr)
\right)
\nonumber\\[0.4em]
&=
\frac{\mathrm d}{\mathrm d H}
\log \bigl(\mathrm{Var}_{\omega}(U^{[H]})\bigr)
-
\frac{\mathrm d}{\mathrm d H}
\log \bigl(\bar v^{[H]}(\pazocal N)^2\bigr)
\nonumber\\[0.4em]
&=
\frac{
\frac{\mathrm d}{\mathrm d H}
\mathrm{Var}_{\omega}(U^{[H]})
}{
\mathrm{Var}_{\omega}(U^{[H]})
}
-
\frac{
\frac{\mathrm d}{\mathrm d H}
\bigl(\bar v^{[H]}(\pazocal N)^2\bigr)
}{
\bar v^{[H]}(\pazocal N)^2
}
\label{eq:log_R}
\end{align}
By assumption~\eqref{eq:log_derivative_condition},
the right-hand side of~\eqref{eq:log_R} is non-positive.
Hence,
\begin{equation}
    \frac{\mathrm d}{\mathrm d H} \log R(H) \le 0
\end{equation}
which implies that $R(H)$ is non-increasing in $H$.
Since $\nu_{\pazocal N}^{\mathrm{LB}[H]}$
is a decreasing function of $R(H)$,
it follows that the lower bound on the probability of profitable co-investment~$\nu_{\pazocal N}^{\mathrm{LB}[H]}$
is non-decreasing in $H$.
This completes the proof.}

\section{Proof of Proposition~\ref{proposition:grand_stable}}\label{appendix:proof_prop_poistivepayoff}
If the grand coalition~$\pazocal N$ is stable in the stochastic game~$(\pazocal N, \pazocal V)$, by Definition~\ref{def:rcore}, this means that there exists an allocation policy~$\vec{y} = (y_i)_{i \in \pazocal N}$ such that:
\begin{equation}
    \begin{aligned}
        \sum_{i \in \pazocal N} y_i = 1 \quad, \qquad \sum_{i \in \pazocal S} v_\omega(\pazocal N) \cdot y_i \geq v_\omega(\pazocal S), \quad \forall \omega \in \Omega, \pazocal S \subseteq \pazocal N
    \end{aligned}
\end{equation}

We define the individual payoff of player~$i$ in realization~$\omega$ as:
\begin{equation}
    x_{i,\omega} \defeq y_i \cdot v_\omega(\pazocal N) , \quad \forall i \in \pazocal N
    \label{equation:ind_payoff}
\end{equation}

Since the robust core conditions apply to all coalitions, including singleton coalitions \( \{i\} \subseteq \pazocal{N} \), it follows that:
\begin{equation}
    y_i \geq \frac{v_\omega(\{i\})}{v_\omega(\pazocal N)}, \quad \forall i \in \pazocal N
    \label{equation:y_payoff}
\end{equation}
In our case, for any player~$i \in \pazocal N$, the singleton payoff~$v_\omega(\{i\})$ is zero; this holds for both the \textit{InP} and the SPs, as established in Proposition~\ref{proposition:pointless}. Moreover, since~$v_\omega(\pazocal N) > 0$ (by assumption), it follows from~\eqref{equation:y_payoff} that~$y_i \geq 0$ for all~$i \in \pazocal N$.
Therefore, we obtain:
\begin{equation}
x_{i,\omega} \defeq y_i \cdot v_\omega(\pazocal N) \geq 0, \quad \forall \omega \in \Omega, i \in \pazocal N
\end{equation}

Thus, the co-investment is profitable for all players in every realization~$\omega \in \Omega$ (see Definition~\ref{definition:profitable_coinv}), i.e.,
\begin{equation}
    \mathbb{P}\left(\ x_{i,\omega} \geq 0, \forall i \in \pazocal N \right) = 1
\end{equation}
which proves the Proposition.
\section{Proof of Proposition~\ref{proposition:profitable_in_stable}}\label{appendix:prof_in_stable}
Let~$\Omega_{\mathrm{stable}} \subseteq \Omega$ be the set of realizations in which the grand coalition is stable, i.e., where there exists a payoff allocation policy~$\vec{y}$ satisfying the robust core conditions (Definition~\ref{def:rcore}).

By Proposition~\ref{proposition:grand_stable}, if~$\omega \in \Omega_{\mathrm{stable}}$, then~$x_{i,\omega} \geq 0$~\eqref{equation:ind_payoff} for all~$i \in \pazocal N$. Thus:
\begin{equation}
    \Omega_{\mathrm{stable}} \subseteq \Omega_{\mathrm{profitable}} \defequiv \left\{ \omega \in \Omega \mid \forall i \in \pazocal N,\ x_{i,\omega} \geq 0 \right\}
\end{equation}
Taking the probability of both sets gives:
\begin{equation}
    \mathbb{P}(\Omega_{\mathrm{stable}}) \leq \mathbb{P}(\Omega_{\mathrm{profitable}})
\end{equation}
which proves the claim.
\section{Proof of Corollary~\ref{corollary:v_lb}}\label{appendix:proof_coro}
By definition, $\nu^{\text{LB}}$ is a lower bound on the probability that the grand coalition is stable, i.e.,
\begin{equation}
    \mathbb{P}\left( \text{Grand coalition is stable} \right) \ge \nu^{\text{LB}}
    \label{equation:lower_stable}
\end{equation}
From Proposition~\ref{proposition:profitable_in_stable}, we have:
\begin{equation}
    \mathbb{P}\left( x_{i,\omega} \geq 0,\ \forall i \in \pazocal N \right) \ge \mathbb{P}\left( \text{Grand coalition is stable} \right) 
    \label{equation:stable_prof}
\end{equation}
Combining~\eqref{equation:lower_stable} and \eqref{equation:stable_prof}  yields:
\begin{equation}
    \mathbb{P}\left( x_{i,\omega} \geq 0,\ \forall i \in \pazocal N \right) \ge \nu^{\text{LB}}
\end{equation}
which proves the corollary.

\section{Proof of Lemma~\ref{lemma:sum_reward}}
\label{appendix:budget-balance}
We start from the definition of the reward for each player~$i \in \pazocal N$~\eqref{equation:reward}:
\begin{equation}
    r_{i,\omega} = x_{i,\omega} + p_i
\end{equation}

Summing over all players~$i \in \pazocal N$, we get:
\begin{equation}
    \sum_{i \in \pazocal N} r_{i,\omega} = \sum_{i \in \pazocal N} x_{i,\omega} + \sum_{i \in \pazocal N} p_i
    \label{equation:summing}
\end{equation}

By the definition in (Section~\ref{sec:shapley}), the total payoff distributed to players satisfies:
\begin{equation}
    \sum_{i \in \pazocal N} x_{i,\omega}
    =
    v_{\omega}(\pazocal N)
    \overset{\eqref{eq: generic expression of payoff-realizations_s}, \eqref{eq: generic expression of payoff-realizations-complete}}{=}
    \sum_{i \in \pazocal N} \sum_{t \in \pazocal T}
    u_{i,\omega}^t(
    \deleted{h}\added{\mathbf h}_{i, \pazocal N}^{t*}
    )
    -
    \text{Cost}(I,
    \deleted{C}\added{\mathbf C}_{\pazocal N}^*
    ),
    \quad \forall \omega \in \Omega
\end{equation}

From~\eqref{eq:cost-sharing}, we also have:
\begin{equation}
    \sum_{i \in \pazocal N} p_i
    =
    \text{Cost}(I,
    \deleted{C}\added{\mathbf C}_{\pazocal N}^*
    )
    \label{equation:payent}
\end{equation}

Substituting~\eqref{equation:payent} into~\eqref{equation:summing}:
\begin{equation}
    \sum_{i \in \pazocal N} r_{i,\omega}
    =
    \left(
    \sum_{i \in \pazocal N} \sum_{t \in \pazocal T}
    u_{i,\omega}^t(
    \deleted{h}\added{\mathbf h}_{i, \pazocal N}^{t*}
    )
    -
    \text{Cost}(I,
    \deleted{C}\added{\mathbf C}_{\pazocal N}^*
    )
    \right)
    +
    \text{Cost}(I,
    \deleted{C}\added{\mathbf C}_{\pazocal N}^*
    )
\end{equation}

Therefore:
\begin{equation}
    \sum_{i \in \pazocal N} r_{i,\omega}
    =
    \sum_{i \in \pazocal N} \sum_{t \in \pazocal T}
    u_{i,\omega}^t(
    \deleted{h}\added{\mathbf h}_{i, \pazocal N}^{t*}
    )
\end{equation}
This completes the proof.

\section{Proof of Lemma~\ref{lemma: mean_s}}\label{appendix: mean}
The expected amount of requests~\eqref{equation: expectedload} received by SP~$i$ at time slot~$t$ is given by:
\begin{equation}
    \bar{l}_{i}^t \defeq \bar{\lambda}_{i}^t \cdot \Delta
    \label{equation:load_exp}
\end{equation}
where 
\begin{equation}
\begin{aligned}
    \bar{\lambda}_{i}^t  \overset{\eqref{equation:lambda_bar}}{\defeq}\mathbb{E}_{\omega}[\lambda_{i, \omega}^t] \overset{\eqref{equation: request_rate_general}}{=} d_i^t \cdot \left[ (1 - \alpha) \cdot \mathbb{E}_\omega \left[ S_{i, \omega}^t \right] + \alpha \cdot \mathbb{E}_\omega \left[ S_{i, \omega}^t \right] \right] 
    = d_i^t \cdot \mathbb{E}_\omega \left[ S_{i, \omega}^t \right] 
\end{aligned}
    \label{equation:lambda_exp}
\end{equation}
We now compute the expected value of the stochastic process~$S_{i, \omega}^t$, which is defined in~\eqref{equation: stochastic_term_f} as~$S_{i, \omega}^t \defeq \max(0, f_{i, \omega}^t)$.

The term $f_{i, \omega}^t$ is a fractional Brownian motion, as described in (Section~\ref{section: brownian}). The mean of~ $f_{i, \omega}^t$ is zero, and its standard deviation is given by~$\sigma_i^t = t^H$~\citep[Page 6]{dieker2004simulation}.

The expected value of $S_{i, \omega}^t$ is:
\begin{equation}
    \mathbb{E}[S_{i, \omega}^t] = \int_0^\infty x \cdot \textit{pdf}_{f_{i, \omega}^t}(x) \, dx
    \label{eq: integral_expected_s}
\end{equation}
where~\( x \) is the possible values that~\( f_{i, \omega}^t \) can take. Since~\( S_{i, \omega}^t~\defeq~\max(0, f_{i, \omega}^t) \), only the non-negative values of \( f_{i, \omega}^t \) contribute to the expectation. Therefore, the integral is taken over \( x \geq 0 \), which accounts for all non-negative realizations of \( f_{i, \omega}^t \).

Since $f_{i, \omega}^t$ is Gaussian at each time slot \( t \)~\citep[Section 1]{enriquez2004simple}, 
the probability density function~$\textit{pdf}_{f_{i, \omega}^t}$ is given by~\citep[Section 3.3]{bertsekas2008introduction}:

\begin{equation}
\textit{pdf}_{f_{i, \omega}^t}(x) = \frac{1}{\sqrt{2\pi} \sigma_i^t} \exp\left(-\frac{x^2}{2 (\sigma_i^t)^2}\right)
\label{eq: pdf}
\end{equation}
Substituting~\eqref{eq: pdf} into~\eqref{eq: integral_expected_s}:
\begin{equation}
    \begin{aligned}
\mathbb{E}[S_{i,\omega}^t]
&=
\int_0^\infty
   x \;
   \frac{1}{\sqrt{2\pi}\,\sigma_i^t}
   \exp\!\Bigl(-\tfrac{x^2}{2\,(\sigma_i^t)^2}\Bigr)\,dx
\\[6pt]
&=
\frac{1}{\sqrt{2\pi}\,\sigma_i^t}
   \int_0^\infty
     x\;
     \exp\!\Bigl(-\tfrac{x^2}{2\,(\sigma_i^t)^2}\Bigr)\,dx
\\[6pt]
& \text{Let } 
   \alpha = \frac{1}{2\,(\sigma_i^t)^2}
   \quad
   \text{Then:}
\\[2pt]
&=
\frac{1}{\sqrt{2\pi}\,\sigma_i^t}
   \int_0^\infty
     x\;\exp\!\bigl(-\alpha\,x^2\bigr)\,dx
\\[6pt]
& \text{Here, make the change of variable } t = x^2
\\[6pt]
& \therefore 
  \frac{1}{\sqrt{2\pi}\,\sigma_i^t}
    \int_0^\infty
      x\,e^{-\alpha x^2}\,dx
  \;=\;
  \frac{1}{\sqrt{2\pi}\,\sigma_i^t}
    \times
    \frac{1}{2\,\alpha}
\\
&=\;
  \frac{1}{\sqrt{2\pi}\,\sigma_i^t}
    \times
    \frac{1}{2 \times \tfrac{1}{2\,(\sigma_i^t)^2}}
  \;=\;
  \frac{\sigma_i^t}{\sqrt{2\pi}}
\end{aligned}
\label{equation:expected_stochastic}
\end{equation}
Substituting~\eqref{equation:expected_stochastic} into~\eqref{equation:lambda_exp}, and subsequently~\eqref{equation:lambda_exp} into~\eqref{equation:load_exp}, yields the result stated in the Lemma.

\section{\added{Dataset Description}}\label{appendix:dataset_desc}
\added{To further enhance the realism of our simulations and address concerns related to real-world validation, we incorporate a publicly available dataset released by Telecom Italia as part of the Big Data Challenge 2014. The dataset is open-access and can be obtained from the Harvard Dataverse repository.\footnote{\url{https://doi.org/10.7910/DVN/EGZHFV}} It provides Call Detail Records (CDRs) collected from the cellular network infrastructure in the city of Milan, Italy, during the months of November and December 2013. Each record represents the mobile activity over 10-minute intervals for a specific geographical grid cell within the urban area.}

\added{The dataset includes three types of user interactions: SMS (sent and received), voice calls (incoming and outgoing), and Internet connections. Each row contains a timestamp (in milliseconds), a spatial identifier (square ID), and the corresponding load components (SMS-in, SMS-out, call-in, call-out, and Internet usage). These values reflect how user activity is distributed across time and space. To estimate the request load relevant for MEC scenarios, we focus on Internet usage.}

\added{The city is partitioned into a grid of spatial cells (identified by square IDs), and user activity is recorded every 10 minutes. In our study, we aggregate these records into hourly intervals and compute the load for each square. The resulting spatial cells are then partitioned among SPs. Although the original dataset covers only two months of mobile activity, it exhibits consistent diurnal and weekly patterns that are representative of user behavior. To obtain a traffic trace that spans the full investment period (e.g., one year), we repeat each 2-month cycle and multiply each repetition by a random factor drawn from a uniform distribution U(0,2). This allows us to generate a long multi-month sequence while preserving the temporal structure of the original dataset (daily peaks and weekly patterns).}

\added{To compute the expected traffic load, which is required to determine the optimal investment capacity and shares (Assumption \ref{assumption: optimal}), we proceed as follows. Each SP is associated with a fixed number of squares. We assume that each square corresponds to one realization of the traffic process; thus, the expected load at each time slot (e.g., hourly) is computed as the empirical average across the corresponding realizations.}

\section{\added{Proof of Proposition~\ref{prop:multi}}}
\label{appendix:multi}
\added{We now discuss how the Karush–Kuhn–Tucker (KKT) conditions can be used to characterize the optimal solution in the case where the
coalition installs a capacity matrix
$\mathbf C = \{C_{k,m}\}_{1\le k\le K,\,1\le m\le M}$ and each SP receives,
at each time slot, an allocation matrix
$\mathbf h_i^t = \{h_{i,k,m}^t\}_{1\le k\le K,\,1\le m\le M}$.
In this setting, the expected utility takes the form
\begin{equation}
    \bar u_i^t(\mathbf h_i^t)
    =
    \beta_i \bar f_i^t
    \left(
        1 - 
        \exp^
            - \sum_{k=1}^K \sum_{m=1}^M \xi_{k,m} h_{i,k,m}^t
    \right),
    \label{eq:multi-u-appendix}
\end{equation}
where the term inside the exponential is a weighted sum of all allocated
resources. To simplify the notation, we define the aggregate effective allocation
\begin{equation}
    X_i^t
    \defeq
    \sum_{k=1}^K \sum_{m=1}^M \xi_{k,m} h_{i,k,m}^t
    \label{eq:def-X-appendix}
\end{equation}
Starting from the optimization problem~\eqref{equation: optimal solution}--\eqref{eq:constraint-positive}, and substituting the cost function~\eqref{equation: cost} and the utility function~\eqref{eq:multi-u-appendix}, we obtain:
\begin{equation}
    \begin{aligned}
    \max_{\mathbf C, \, \{\mathbf H^t\}_{t\in\pazocal T}}
    \quad
    &
    \sum_{t \in \pazocal{T}} 
    \sum_{i \in \pazocal S} 
    \beta_i \bar f_i^t 
    \left(
        1 - 
        e^{- X_i^t}
    \right)
    \\
    &\hspace{1.5em}
    -
    \sum_{k=1}^K \sum_{m=1}^M
    \left(
        d_{k,m}
        + \int_{0}^{I} D_{k,m}(t)\,dt
    \right) C_{k,m}
    \\
    \text{s.t.} \quad 
    &
    \sum_{i \in \pazocal S} h_{i,k,m}^t \leq C_{k,m},
    \quad 
    \forall k,m,\; \forall t \in \pazocal{T} 
    \\
    &
    h_{i,k,m}^t \geq 0, 
    \quad 
    \forall i \in \pazocal S,\, \forall k,m,\; \forall t \in \pazocal{T} 
    \\
    &
    C_{k,m} > 0,
    \quad \forall k,m.
    \end{aligned}
    \label{eq:optimization_problem_multi}
\end{equation}
Assumption~\ref{ass:inp-zero} implies that
$\mathbf h_{\textit{InP},\pazocal S}^{t*} = \mathbf 0$, $\forall t$.
Let $\pazocal S' \subseteq \pazocal S \setminus \{\textit{InP}\}$ denote the
subset of SPs with strictly positive expected loads.\\
We introduce Lagrange multipliers $\lambda_{k,m,t} \ge 0$ for the capacity
constraints
$\sum_{i\in\pazocal S'} h_{i,k,m}^t \le C_{k,m}$, for all $(k,m,t)$.
The Lagrangian is
\begin{equation}
    \begin{aligned}
    \pazocal{L} = 
    &\sum_{t \in \pazocal{T}} 
    \bigg[
        \sum_{i \in \pazocal S'} 
        \beta_i \bar f_i^t 
        \left(1 - e^{- X_i^t} \right)
        -
        \sum_{k=1}^K \sum_{m=1}^M
        \lambda_{k,m,t} 
        \left( 
            \sum_{i \in \pazocal S'} h_{i,k,m}^t - C_{k,m}
        \right)
    \bigg] 
    \\
    &\hspace{2em}
    -
    \sum_{k=1}^K \sum_{m=1}^M
    \left(
        d_{k,m} 
        +
         \int_{0}^{I} D_{k,m}(t)\,dt
    \right)C_{k,m}
    \end{aligned}
    \label{eq:lagrangian_multi}
\end{equation}
We now take the partial derivatives of~\eqref{eq:lagrangian_multi} with
respect to $h_{i,k,m}^t$ and $C_{k,m}$.
\begin{itemize}
    \item Derivative with respect to $h_{i,k,m}^t$: From~\eqref{eq:def-X-appendix}, we have
$\partial X_i^t / \partial h_{i,k,m}^t = \xi_{k,m}$, and thus
\[
\frac{\partial}{\partial h_{i,k,m}^t}
\left( 
    1 - e^{- X_i^t}
\right)
=
\xi_{k,m} e^{- X_i^t}
\]
Therefore,
\begin{equation}
    \frac{\partial \pazocal{L}}{\partial h_{i,k,m}^t} 
    =
    \beta_i \bar f_i^t \cdot \xi_{k,m} e^{- X_i^t} 
    - \lambda_{k,m,t}
    \label{eq:partial_h_multi}
\end{equation}
Setting~\eqref{eq:partial_h_multi} to zero:
\begin{equation}
    \beta_i \bar f_i^t \cdot \xi_{k,m} e^{- X_i^t} 
    =
    \lambda_{k,m,t},
    \quad \forall i \in \pazocal S',\, \forall k,m,\; \forall t \in \pazocal T
    \label{eq:optimality_h_multi}
\end{equation}
In particular, whenever $h_{i,k,m}^t > 0$ and $h_{i,k',m'}^t > 0$ for some
$(k,m)$ and $(k',m')$, we obtain the proportionality relation
\begin{equation}
    \frac{\lambda_{k,m,t}}{\xi_{k,m}}
    =
    \frac{\lambda_{k',m',t}}{\xi_{k',m'}},
    \quad \forall t \in \pazocal T,
    \label{eq:proportionality_multi}
\end{equation}
that is, all active resource components used by SP~$i$ at time $t$ share
the same normalized multiplier.
\item Derivative with respect to $C_{k,m}$:
Taking the derivative of~\eqref{eq:lagrangian_multi} with respect to
$C_{k,m}$ yields
\begin{equation}
    \frac{\partial \pazocal{L}}{\partial C_{k,m}}
    =
    \sum_{t\in\pazocal T} \lambda_{k,m,t}
    -
    \left(
        d_{k,m}
        +
        \int_{0}^{I} D_{k,m}(t)\,dt
    \right)
    \label{eq:partial_C_multi}
\end{equation}
Setting~\eqref{eq:partial_C_multi} to zero for optimality gives
\begin{equation}
    \sum_{t \in \pazocal{T}} \lambda_{k,m,t}
    =
    d_{k,m}
    +
    \int_{0}^{I} D_{k,m}(t)\,dt,
    \quad \forall k,m
    \label{eq:lambda_sum_multi}
\end{equation}
\end{itemize}
Equations~\eqref{eq:optimality_h_multi}--\eqref{eq:lambda_sum_multi},
together with the primal feasibility conditions
(i.e., $\sum_{i\in\pazocal S'} h_{i,k,m}^t \le C_{k,m}$,
$h_{i,k,m}^t \ge 0$, and $C_{k,m} > 0$)
and the complementary slackness conditions
(i.e., $\lambda_{k,m,t}
\left(\sum_{i\in\pazocal S'} h_{i,k,m}^t - C_{k,m}\right) = 0$ for all
$k,m,t$),
form the KKT conditions for the multi-resource, multi-node case.\\
In the multidimensional setting, however, the
conditions~\eqref{eq:optimality_h_multi}--\eqref{eq:lambda_sum_multi} couple
all entries of the capacity matrix $\mathbf C$ and of the allocation
matrices $\mathbf h_i^t$ across time, resource types and nodes. As a
consequence, no closed-form expression analogous to
\eqref{eq:C_star}--\eqref{eq:h_star_final} is available in general.}

\added{Nevertheless, Problem~\eqref{equation: optimal solution}--\eqref{eq:constraint-positive}
remains convex under the assumptions of
Proposition~\ref{prop:convexity-of-problem}, and the KKT conditions above is
both necessary and sufficient for optimality. Hence, the optimal capacity
matrix $\mathbf C_{\pazocal S}^*$ and the allocations
$\mathbf h_{i,\pazocal S}^{t*}$ for all $i\in\pazocal S$, $t\in\pazocal T$
exist, are unique, and can be efficiently computed using standard convex
optimization solvers.
However, in the scalar case ($K=1$, $M=1$), the KKT conditions yield explicit formulas for $C_{\pazocal S}^*$ and $h_{i,\pazocal S}^{t*}$, as shown in Proposition~\ref{Proposition: kkt} and its proof in Appendix~\ref{appendix: kkt}. }

\section{Proof of Proposition~\ref{Proposition: kkt}} \label{appendix: kkt}
\subsection{Case 1:}\label{section:activeload}
Starting from the optimization problem~\eqref{equation: optimal solution}–\eqref{eq:constraint-positive}, and substituting each term with its expression from~\eqref{equation: expected payoff},\eqref{equation: expected_utility},\eqref{equation: utility function}, and~\eqref{equation: cost}, we obtain the following optimization problem:
\begin{equation}
    \begin{aligned}
    \max_{\{{\vec{h}^t}\}_{t \in \pazocal T},\, C \geq 0} \quad & \sum_{t \in \pazocal{T}} \sum_{i \in \pazocal S} \beta_i \deleted{\bar{l}}\added{\bar{f}}_i^t \left(1 - e^{-\xi h_i^t} \right) - \left(d + \added{\int_{0}^{I} D(t)\,dt}\deleted{d' \cdot I} \right) \cdot C \\
    \text{s.t.} \quad & \sum_{i \in \pazocal S} h_i^t \leq C, \quad \forall t \in \pazocal{T} \\
    & h_i^t \geq 0, \quad \forall i \in \pazocal S,\, t \in \pazocal{T} \\
    & C > 0
    \end{aligned}
    \label{eq:optimization_problem}
\end{equation}
According to Assumption~\ref{ass:inp-zero}, the \textit{InP} does not interact directly with end users; therefore, its load is $ \deleted{\bar{l}}\added{\bar{f}}_{\text{InP}}i^t = 0$, and no resources are allocated to the \textit{InP}, i.e.,~$h_{\textit{InP}, \pazocal{S}}^{t*} = 0, \quad \forall t \in \pazocal T$~\eqref{equation:no_resources_inp}. 

Let $\pazocal S' \subseteq \pazocal S \setminus \{\textit{InP}\}$ denote the subset of SPs with strictly positive load, i.e., $\deleted{\bar{l}}\added{\bar{f}}_i^t > 0$ for all $t \in \pazocal T$.
The capacity and allocation depend only on this active subset.

We define the Lagrangian function as follows:

\begin{equation}
    \begin{aligned}
    \pazocal{L} = 
    &\sum_{t \in \pazocal{T}} \left[ \sum_{i \in \pazocal S'} \beta_i \deleted{\bar{l}}\added{\bar{f}}_i^t \left(1 - e^{-\xi h_i^t} \right) - \lambda_t \left( \sum_{i \in \pazocal S'} h_i^t - C \right) \right] 
    - \left(d + \added{\int_{0}^{I} D(t)\,dt}\deleted{d' \cdot I} \right) \cdot C
    \end{aligned}
    \label{eq:lagrangian}
\end{equation}
Here, $\lambda_t$ is the Lagrange multiplier associated with the capacity constraint at each time slot $t \in \pazocal{T}$.

We take the partial derivatives of \eqref{eq:lagrangian} with respect to \( h_i^t \) and \( C \).
\begin{itemize}
    \item Partial derivative with respect to \( h_i^t \): 
    \begin{equation}
    \frac{\partial \pazocal{L}}{\partial h_i^t} = \beta_i \deleted{\bar{l}}\added{\bar{f}}_i^t \cdot \xi e^{-\xi h_i^t} - \lambda_t
    \label{eq:partial_h}
\end{equation}

Setting the derivative in \eqref{eq:partial_h} to zero for optimality:
\begin{equation}
    \beta_i \deleted{\bar{l}}\added{\bar{f}}_i^t \cdot \xi e^{-\xi h_i^t} = \lambda_t
    \label{eq:optimality_h}
\end{equation}

Solving for \( h_i^t \):
\begin{equation}
    h_{i, \pazocal S}^{t*} = \frac{1}{\xi} \ln \left( \frac{\xi \beta_i \deleted{\bar{l}}\added{\bar{f}}_i^t}{\lambda_t} \right)
    \label{eq:h_star}
\end{equation}

    \item Partial derivative with respect to \( C \):
\begin{equation}
    \frac{\partial \pazocal{L}}{\partial C} = \sum_{t \in \pazocal{T}} \lambda_t - \left(d + \added{\int_{0}^{I} D(t)\,dt}\deleted{d' \cdot I} \right)
    \label{eq:partial_C}
\end{equation}

Setting the derivative in \eqref{eq:partial_C} to zero for optimality:
\begin{equation}
    \sum_{t \in \pazocal{T}} \lambda_t = d + \added{\int_{0}^{I} D(t)\,dt}\deleted{d' \cdot I}
    \label{eq:lambda_sum}
\end{equation}
\end{itemize}

We substitute the optimal \( h_{i, \pazocal S}^{t*} \) from \eqref{eq:h_star} into the constraint in \eqref{eq:optimization_problem}:
\begin{equation}
    \sum_{i \in \pazocal{S'} } h_{i, \pazocal S}^{t*} = \sum_{i \in \pazocal{S'}} \frac{1}{\xi} \ln \left( \frac{\xi \beta_i \deleted{\bar{l}}\added{\bar{f}}_i^t}{\lambda_t} \right) = C, \quad \forall t \in \pazocal{T}
    \label{eq:constraint_substitution}
\end{equation}

This gives:
\begin{equation}
    \frac{1}{\xi} \left( \sum_{i \in \pazocal{S'}} \ln (\xi \beta_i \deleted{\bar{l}}\added{\bar{f}}_i^t) - |\pazocal S'| \ln \lambda_t \right) = C, \quad \forall t
    \label{eq:solve_lambda_log}
\end{equation}
where $|\pazocal S'|$ is the number of SPs in the coalition~$\pazocal S'$.

Then, solving~\eqref{eq:solve_lambda_log} for \( \lambda_t \):
\begin{equation}
    \ln \lambda_t = \frac{1}{|\pazocal S'|} \sum_{i \in \pazocal{S'}} \ln (\xi \beta_i \deleted{\bar{l}}\added{\bar{f}}_i^t) - \frac{\xi C}{|\pazocal S'|}
\quad \Rightarrow \quad
\lambda_t = \left( \prod_{i \in \pazocal{S'} } \xi \beta_i \deleted{\bar{l}}\added{\bar{f}}_i^t \right)^{\frac{1}{|\pazocal S'|}} \cdot e^{-\frac{\xi C}{|\pazocal S'|}}, \quad \forall t
\label{eq:lambda_expression}
\end{equation}

Substituting \eqref{eq:lambda_expression} into \eqref{eq:lambda_sum}:
\begin{equation}
    \sum_{t \in \pazocal{T}} \left( \prod_{i \in \pazocal{S'} } \xi \beta_i \deleted{\bar{l}}\added{\bar{f}}_i^t \right)^{\frac{1}{|\pazocal S'|}} \cdot e^{-\frac{\xi C}{|\pazocal S'|}} = d + \added{\int_{0}^{I} D(t)\,dt}\deleted{d' \cdot I}
    \label{eq:lambda_substitution}
\end{equation}

Solving \eqref{eq:lambda_substitution} for \( C \):
\begin{equation}
    e^{-\frac{\xi C}{|\pazocal S'|}} = \frac{d + \added{\int_{0}^{I} D(t)\,dt}\deleted{d' \cdot I}}{\sum_{t \in \pazocal{T}} \left( \prod_{i \in \pazocal{S'}} \xi \beta_i \deleted{\bar{l}}\added{\bar{f}}_i^t \right)^{\frac{1}{|\pazocal S'|}}}
\quad \Rightarrow \quad
C_{\pazocal S}^* = \frac{|\pazocal S'|}{\xi} \ln \left( \frac{\sum_{t \in \pazocal{T}} \left( \prod_{i \in \pazocal{S'}} \xi \beta_i \deleted{\bar{l}}\added{\bar{f}}_i^t \right)^{\frac{1}{|\pazocal S'|}}}{d + \added{\int_{0}^{I} D(t)\,dt}\deleted{d' \cdot I}} \right)
\label{eq:C_star}
\end{equation}

Substitute \eqref{eq:lambda_expression} and \eqref{eq:C_star} into \eqref{eq:h_star}:
\begin{align}
    h_{i, \pazocal S}^{t*} &= \frac{1}{\xi} \ln \left( \frac{\xi \beta_i \deleted{\bar{l}}\added{\bar{f}}_i^t}{\lambda_t} \right) 
    = \frac{1}{\xi} \ln \left( \frac{\xi \beta_i \deleted{\bar{l}}\added{\bar{f}}_i^t}{\left( \prod_{j \in \pazocal{S'}} \xi \beta_j \deleted{\bar{l}}\added{\bar{f}}_j^t \right)^{\frac{1}{|\pazocal S'|}} \cdot e^{-\frac{\xi C_{\pazocal S}^*}{|\pazocal S'|}}} \right) \label{eq:h_star_substitution} \\
    &= \frac{C_{\pazocal S}^*}{|\pazocal S'|} + \frac{1}{\xi} \ln \left( \frac{\beta_i \deleted{\bar{l}}\added{\bar{f}}_i^t}{\left( \prod_{j \in \pazocal{S'}} \beta_j \deleted{\bar{l}}\added{\bar{f}}_j^t \right)^{\frac{1}{|\pazocal S'|}}} \right)
    \label{eq:h_star_final}
\end{align}

\subsubsection{Special case: Only one SP has non-zero load}\label{appendix:special}
If $|\pazocal S'| = 1$, let the single active SP be denoted by $i$. In this case, the optimization problem simplifies, since only one SP contributes to utility and receives allocation.
The optimization problem becomes:
\begin{equation}
    \begin{aligned}
    \max_{\{{h_i^t}\}_{t \in \pazocal T},\, C \geq 0} \quad & \sum_{t \in \pazocal{T}} \beta_i \deleted{\bar{l}}\added{\bar{f}}_i^t \left(1 - e^{-\xi h_i^t} \right) - \left(d + \added{\int_{0}^{I} D(t)\,dt}\deleted{d' \cdot I} \right) \cdot C \\
    \text{s.t.} \quad & h_i^t \leq C, \quad \forall t \in \pazocal{T} \\
    & h_i^t \geq 0, \quad \forall t \in \pazocal{T} \\
    & C > 0
    \end{aligned}
    \label{eq:optimization_single}
\end{equation}

The Lagrangian becomes:
\begin{equation}
    \pazocal{L} = 
    \sum_{t \in \pazocal{T}} \left[ \beta_i \deleted{\bar{l}}\added{\bar{f}}_i^t \left(1 - e^{-\xi h_i^t} \right) - \lambda_t \left( h_i^t - C \right) \right] 
    - \left(d + \added{\int_{0}^{I} D(t)\,dt}\deleted{d' \cdot I} \right) \cdot C
    \label{eq:lagrangian_single}
\end{equation}

We take the partial derivatives of \eqref{eq:lagrangian_single} with respect to \( h_i^t \) and \( C \).

\begin{itemize}
    \item Partial derivative with respect to \( h_i^t \): 
    \begin{equation}
    \frac{\partial \pazocal{L}}{\partial h_i^t} = \beta_i \deleted{\bar{l}}\added{\bar{f}}_i^t \cdot \xi e^{-\xi h_i^t} - \lambda_t
    \label{eq:partial_h_single}
    \end{equation}

    Setting the derivative in \eqref{eq:partial_h_single} to zero for optimality:
    \begin{equation}
        \beta_i \deleted{\bar{l}}\added{\bar{f}}_i^t \cdot \xi e^{-\xi h_i^t} = \lambda_t
        \label{eq:optimality_h_single}
    \end{equation}

    Solving for \( h_i^t \):
    \begin{equation}
        h_{i, \pazocal S}^{t*} = \frac{1}{\xi} \ln \left( \frac{\xi \beta_i \deleted{\bar{l}}\added{\bar{f}}_i^t}{\lambda_t} \right)
        \label{eq:h_star_single}
    \end{equation}

    \item Partial derivative with respect to \( C \):
    \begin{equation}
        \frac{\partial \pazocal{L}}{\partial C} = \sum_{t \in \pazocal{T}} \lambda_t - \left(d + \added{\int_{0}^{I} D(t)\,dt}\deleted{d' \cdot I} \right)
        \label{eq:partial_C_single}
    \end{equation}

    Setting the derivative in \eqref{eq:partial_C_single} to zero for optimality:
    \begin{equation}
        \sum_{t \in \pazocal{T}} \lambda_t = d + \added{\int_{0}^{I} D(t)\,dt}\deleted{d' \cdot I}
        \label{eq:lambda_sum_single}
    \end{equation}
\end{itemize}

We now use the fact that \( |\pazocal S'| = 1 \), so the constraint in \eqref{eq:optimization_single} implies:
\begin{equation}
    h_{i, \pazocal S}^{t*} = C, \quad \forall t \in \pazocal{T}
    \label{eq:h_equals_C_single}
\end{equation}

Substitute \eqref{eq:h_equals_C_single} into \eqref{eq:optimality_h_single}:
\begin{equation}
    \lambda_t = \xi \beta_i \deleted{\bar{l}}\added{\bar{f}}_i^t \cdot e^{-\xi C}, \quad \forall t
    \label{eq:lambda_expression_single}
\end{equation}

Substituting \eqref{eq:lambda_expression_single} into \eqref{eq:lambda_sum_single}:
\begin{equation}
    \sum_{t \in \pazocal{T}} \xi \beta_i \deleted{\bar{l}}\added{\bar{f}}_i^t \cdot e^{-\xi C} = d + \added{\int_{0}^{I} D(t)\,dt}\deleted{d' \cdot I}
    \label{eq:lambda_substitution_single}
\end{equation}

Solving \eqref{eq:lambda_substitution_single} for \( C \):
\begin{equation}
    e^{-\xi C} = \frac{d + \added{\int_{0}^{I} D(t)\,dt}\deleted{d' \cdot I}}{\sum_{t \in \pazocal{T}} \xi \beta_i \deleted{\bar{l}}\added{\bar{f}}_i^t}
    \quad \Rightarrow \quad
    C_{\pazocal S}^* = \frac{1}{\xi} \ln \left( \frac{\sum_{t \in \pazocal{T}} \xi \beta_i \deleted{\bar{l}}\added{\bar{f}}_i^t}{d + \added{\int_{0}^{I} D(t)\,dt}\deleted{d' \cdot I}} \right)
    \label{eq:C_star_single}
\end{equation}

Using \eqref{eq:h_equals_C_single} and \eqref{eq:C_star_single}, the final allocation is:
\begin{equation}
    h_{i, \pazocal S}^{t*} = C_{\pazocal S}^* = \frac{1}{\xi} \ln \left( \frac{\sum_{t \in \pazocal{T}} \xi \beta_i \deleted{\bar{l}}\added{\bar{f}}_i^t}{d + \added{\int_{0}^{I} D(t)\,dt}\deleted{d' \cdot I}} \right)
    \label{eqaution:share_h_one_active}
\end{equation}

\subsection{Case 2:}\label{section:zeroload}
From the utility function~\eqref{equation: utility function} and \eqref{equation: expected_utility}, if~$\deleted{\bar{l}}\added{\bar{f}}_i^t = 0$, then~$\bar{u}_i^t(h_i^t) = 0$ for any~$h_i^t$. Since no utility is generated, the value of the objective function~\eqref{equation: optimal solution} is zero regardless of the installed capacity. Therefore, the optimizer will choose~$C_{\pazocal{S}}^* = 0$ and allocates no resources ($h_{i, \pazocal S}^{t*}=0, \forall t \in \pazocal T$) to avoid unnecessary costs.


\section{Proof of Proposition~\ref{Proposition:h_increase}}\label{appendix:h_increase_with_b_l}
\begin{itemize}
    \item Case 1:\\
    The allocation given by~\eqref{eqaution:share_h_active} can be rewritten as:
    \begin{equation}
    h_{i, \pazocal{S}}^{t*} = \frac{C_{\pazocal{S}}^*}{|\pazocal S'|} + \frac{1}{\xi} \left( \ln(\beta_i \cdot \deleted{\bar{l}}\added{\bar{f}}_i^t) - \frac{1}{|\pazocal S'|} \sum_{j \in \pazocal{S'}} \ln(\beta_j \cdot \deleted{\bar{l}}\added{\bar{f}}_j^t) \right)
    \label{equation:rewr}
    \end{equation}
    Taking the partial derivatives of \eqref{equation:rewr}:
    \begin{equation}
    \frac{\partial h_{i, \pazocal{S}}^{t*}}{\partial \beta_i} 
    = \frac{1}{\xi} \left( \frac{1}{\beta_i} - \frac{1}{|\pazocal S'|} \cdot \frac{1}{\beta_i} \right) 
    = \frac{1}{\xi} \left(1 - \frac{1}{|\pazocal S'|} \right) \cdot \frac{1}{\beta_i} > 0 \quad \text{for } \beta_i > 0    
    \end{equation}
    
    \begin{equation}
        \frac{\partial h_{i, \pazocal{S}}^{t*}}{\partial \deleted{\bar{l}}\added{\bar{f}}_i^t} 
    = \frac{1}{\xi} \left( \frac{1}{\deleted{\bar{l}}\added{\bar{f}}_i^t} - \frac{1}{|\pazocal S'|} \cdot \frac{1}{\deleted{\bar{l}}\added{\bar{f}}_i^t} \right)
    = \frac{1}{\xi} \left(1 - \frac{1}{|\pazocal S'|} \right) \cdot \frac{1}{\deleted{\bar{l}}\added{\bar{f}}_i^t} > 0 \quad \text{for } \deleted{\bar{l}}\added{\bar{f}}_i^t > 0
    \end{equation}
    So \( h_{i, \pazocal{S}}^{t*} \) is strictly increasing in both \( \beta_i \) and \( \deleted{\bar{l}}\added{\bar{f}}_i^t \).\\
    For the special case discussed in \ref{appendix:special}, we take the partial derivatives for~\eqref{eqaution:share_h_one_active}:
    \begin{equation}
        \frac{\partial h_{i, \pazocal{S}}^{t*}}{\partial \beta_i} = \frac{1}{\xi} \cdot \frac{\sum_{t \in \pazocal T} \deleted{\bar{l}}\added{\bar{f}}_i^t}{\sum_{t \in \pazocal T} \beta_i \cdot\deleted{\bar{l}}\added{\bar{f}}_i^t} = \frac{1}{\xi \cdot \beta_i} > 0
    \end{equation}
    \begin{equation}
        \frac{\partial h_{i, \pazocal{S}}^{t*}}{\partial \deleted{\bar{l}}\added{\bar{f}}_i^t} = \frac{1}{\xi} \cdot \frac{\xi \cdot \beta_i}{\sum_{t' \in \pazocal{T}} \xi \cdot \beta_i \cdot \deleted{\bar{l}}\added{\bar{f}}_i^{t'}} = \frac{1}{\xi} \cdot \frac{1}{\sum_{t' \in \pazocal{T}} \deleted{\bar{l}}\added{\bar{f}}_i^{t'}} > 0
    \end{equation}
    So again, $h_{i, \pazocal{S}}^{t*}$ increases strictly with \( \beta_i \) and \( \deleted{\bar{l}}\added{\bar{f}}_i^t \).
    \item Case 3:\\
    If \( \deleted{\bar{l}}\added{\bar{f}}_i^t = 0 \), then \( h_{i, \pazocal{S}}^{t*} = 0 \)~\eqref{equation:share_zero}, regardless of \( \beta_i \), so the proposition holds trivially.
\end{itemize}

\end{appendices}



\end{document}